\documentclass[a4paper,onecolumn,11pt,accepted=2026-02-06]{quantumarticle}
\pdfoutput=1
\usepackage[a4paper]{geometry}

\usepackage[utf8]{inputenc}
\usepackage[english]{babel}
\usepackage[T1]{fontenc}
\usepackage{hyperref}

\usepackage{url}

\usepackage{tikz}
\usepackage{lipsum}

\pdfoutput=1

\usepackage{graphicx}
\usepackage{floatrow}

\usepackage{dcolumn}
\usepackage{bm}
\usepackage{amsfonts}
\usepackage{dsfont}
\usepackage{bm}
\usepackage{bbold}
\usepackage[normalem]{ulem}

\usepackage[utf8]{inputenc}
\usepackage{amsmath}
\usepackage{physics}
\usepackage{amssymb}
\usepackage{braket}
\usepackage{mathtools}
\usepackage{color}
\usepackage{amsthm}
\usepackage{stmaryrd}

\usepackage{etoolbox}

\usepackage{cite}

\newtoggle{insidesection}
\togglefalse{insidesection}
\patchcmd{\section}{\bfseries}{\bfseries\toggletrue{insidesection}}{}{}

\DeclareMathOperator*{\E}{\mathbb{E}}

\theoremstyle{plain}
\newtheorem{theorem}{Theorem}
\newtheorem{corollary}{Corollary}

\newtheorem{lemma}{Lemma}
\newtheorem{definition}{Definition}
\newtheorem{remark}{Remark}

\definecolor{red}{RGB}{255,0,0}

\begin{document}
	
\title{On computational complexity and average-case hardness of shallow-depth Boson Sampling}
\author{Byeongseon Go}
\affiliation{NextQuantum Innovation Research Center, Department of Physics and Astronomy, Seoul National University, Seoul 08826, Republic of Korea}
\author{Changhun Oh}
\email{changhun0218@gmail.com}
\affiliation{Department of Physics, Korea Advanced Institute of Science and Technology, Daejeon 34141, Republic of Korea}
\author{Hyunseok Jeong}
\email{h.jeong37@gmail.com}
\affiliation{NextQuantum Innovation Research Center, Department of Physics and Astronomy, Seoul National University, Seoul 08826, Republic of Korea}

\begin{abstract}
Boson Sampling, a computational task believed to be classically hard to simulate, is expected to hold promise for demonstrating quantum computational advantage using near-term quantum devices. 
However, noise in experimental implementations poses a significant challenge, potentially rendering Boson Sampling classically simulable and compromising its classical intractability. 
Numerous studies have proposed classical algorithms that can efficiently simulate Boson Sampling under various noise models, particularly as noise rates increase with circuit depth. 
To address this challenge, we investigate the viability of achieving quantum computational advantage through Boson Sampling implemented with shallow-depth linear optical circuits.
In particular, as the average-case hardness of estimating output probabilities of Boson Sampling is a crucial ingredient in demonstrating its classical intractability, we make progress on establishing the average-case hardness of Boson Sampling confined to logarithmic-depth regimes. 
We also obtain the average-case hardness for logarithmic-depth Fock-state Boson Sampling subject to lossy environments and for the logarithmic-depth Gaussian Boson Sampling. 
By providing complexity-theoretical backgrounds for the classical simulation hardness of logarithmic-depth Boson Sampling, we expect that our findings will mark a crucial step towards a more noise-tolerant demonstration of quantum advantage with shallow-depth Boson Sampling.    
\end{abstract}

\maketitle

\tableofcontents

\section{Introduction}

\subsection{Backgrounds and motivation}

Boson Sampling is a computational task that is complexity-theoretically proven to be hard to classically simulate under plausible assumptions~\cite{aaronson2011computational, hamilton2017gaussian, deshpande2022quantum}. 
Accordingly, Boson Sampling has gathered significant attention, as it would possibly play a key role in the experimental demonstration of quantum computational advantage using near-term quantum devices.
However, the implementation of Boson Sampling in experimental settings with near-term quantum devices is inevitably subject to various sources of noise~\cite{zhong2020quantum, zhong2021phase, madsen2022quantum, deng2023gaussian}.
The problem is that those noises would possibly rule out the classical intractability of Boson Sampling, and thus potentially hinder the experimental demonstration of quantum advantage with Boson Sampling.
Indeed, both for finite-size near-term experiments and asymptotic limits as system size scales, numerous studies~\cite{renema2018classical, renema2018efficient, moylett2019classically, renema2020simulability, shi2022effect, van2024efficient, shchesnovich2019noise, oszmaniec2018classical, garcia2019simulating, qi2020regimes, brod2020classical, villalonga2021efficient, oh2021classical2, bulmer2022boundary, liu2023simulating, oh2023classical, oh2024classicalalgorithm, oh2025classical, oh2025recent} have proposed efficient classical simulation algorithms of Boson Sampling under various noise models, such as photon loss, partial distinguishability, random Gaussian noise, etc. 
Their results indicate that as the noise rate of Boson Sampling increases, it eventually renders such a noisy sampler classically simulable.
Moreover, as the noise is typically accumulated with each circuit depth, the quantum signal for classical intractability exhibits exponential decay with increasing circuit depth.
Hence, circuits with polynomially increasing depth with system size would suffer from significantly enlarged noise rates, posing substantial challenges to achieving quantum advantage in such settings.

A viable alternative to preclude the classical simulability due to the inevitable noise is to consider Boson Sampling with $\textit{shallow-depth}$ linear optical circuits, where the noise rate can be highly reduced.  
Specifically, among the shallow-depth regime, our primary focus is on investigating the simulation hardness for $\it{logarithmic}$ depth circuits; 
the intuition behind investigating logarithmic depth circuits lies in the potential to offer a ``sweet-spot" regime for the hardness of Boson Sampling.
Namely, this depth regime may avoid significant increases in noise rates to prevent classical simulability, while still being sufficiently large to generate quantum correlations and uphold classical intractability.
Despite such intuitive understanding, the hardness argument of Boson Sampling in this shallow-depth regime, particularly from a complexity-theoretical perspective, has been less studied so far and thus remains widely open. 
Hence, our goal is to establish the \textit{complexity-theoretical foundations} of the classical hardness of shallow-depth Boson Sampling, to suppress the classical simulability by noise in a rigorous manner and obtain a more noise-tolerant demonstration of quantum advantage with Boson Sampling.

In this work, we investigate the classical simulation hardness of Boson Sampling in shallow linear optical circuits. 
Specifically, as the average-case \#P-hardness of estimating output probabilities of Boson Sampling is a crucial ingredient to demonstrate the classical intractability of Boson Sampling, we make progress on establishing the average-case \#P-hardness confined in shallow-depth regimes. 
Similarly, we obtain the average-case hardness result in the shallow-depth regime for the Gaussian Boson Sampling scheme.
Finally, since noise is our main motivation for investigating shallow-depth Boson Sampling, we generalize our average-case hardness result to noisy Boson Sampling subject to a photon loss channel.

To avoid confusion, we note that the allowed imprecision level of our average-case \#P-hardness result is not sufficient to fully demonstrate the classical intractability of Boson Sampling in shallow-depth regimes.
However, to the best of our knowledge, the complexity-theoretical analysis on the average-case hardness of shallow-depth Boson Sampling has not yet been investigated. 
Hence, we believe that our hardness result in shallow-depth regimes will provide a first step toward a stronger hardness result and, ultimately, toward the full demonstration of the classical intractability of shallow-depth Boson Sampling. 


\subsection{Our results: Average-case hardness of shallow-depth Boson Sampling}\label{section: outline}

We set our goal as proving the hardness of classical simulation of Boson Sampling in the shallow-depth regime, specifically for \textit{approximate} simulation within total variation distance error.
Two key ingredients for the current hardness proof of the approximate simulation of Boson Sampling are $(i)$ average-case \#P-hardness of output probability approximation up to sufficiently large additive imprecision $\epsilon$, and $(ii)$ hiding property.
Informally, average-case hardness means that approximating the output probability of Boson Sampling with high probability over randomly chosen circuits (i.e., on average over circuits) is \#P-hard.
Here, by choosing random circuit ensembles that have the hiding property (i.e., symmetry over outcomes), one can reduce the average-case instances for the hardness from circuit instances to outcome instances, which is a crucial step to prove the hardness of approximate simulation within total variation distance error (See Appendix~\ref{previousfoundations averagecasehardness} for more details).

Most of the current theoretical foundations of the average-case hardness of Boson Sampling rely on global Haar random unitary circuits~\cite{aaronson2011computational, hamilton2017gaussian, deshpande2022quantum, grier2022complexity, bouland2022noise, bouland2023complexity, bouland2024average}, as they almost satisfy the two conditions described above. 
Namely, the outcome instances can be effectively hidden by global Haar random unitaries, and approximating the output probability within sufficiently large $\epsilon$ on average over global Haar random circuit instances is \#P-hard under some conjectures. 
However, the problem for our case is that implementing a global Haar random unitary requires at least polynomially large circuit depth (e.g., see~\cite{zyczkowski1994random, russell2017direct}), and thus is not implementable in sub-polynomial circuit depths.
This poses a challenge in proving the average-case hardness of Boson Sampling in the shallow-depth regime—specifically, the lack of the hiding property.
Due to the absence of hiding in the shallow-depth regime, the random outcome instances of Boson Sampling cannot be hidden by random circuit instances. 
Accordingly, even if we establish average-case hardness over randomly chosen circuits for a fixed outcome, it does not directly imply classical simulation hardness via Stockmeyer’s reduction, which requires average-case hardness over randomly chosen outcomes~\cite{aaronson2011computational}.

Another problem is that there already exist efficient classical algorithms that can approximately simulate shallow-depth Boson Sampling in certain circumstances, which directly rule out the classical simulation hardness in the shallow-depth regime for those cases.
While exact simulation of Boson Sampling is classically hard even for constant-depth circuits~\cite{brod2015complexity}, approximate simulation becomes easy for 1-dimensional local log-depth circuits~\cite{vidal2003efficient, qi2022efficient} and also for more general dimension local circuits under some constraints~\cite{deshpande2018dynamical, oh2022classical, kolarovszki2023simulating}.
According to their results, if we use circuits composed of only geometrically local gates, a polynomial circuit depth is generally required for a sufficiently large correlation to achieve approximate simulation hardness.
Those results indicate that we cannot expect the hardness results in the most general case of shallow-depth circuits composed of local gates only.

To deal with those problems, we take the following approach: 
First, we consider shallow linear optical circuit architectures that fully employ geometrically \textit{non-local} gates. 
In fact, the implementation of non-local gates is promising for near-term experimental settings; for example, experiments of linear optical systems based on trapped ions~\cite{shen2014scalable, chen2023scalable} and photonic architecture~\cite{madsen2022quantum} implemented long-range interactions. 
Because of the absence of hiding, we derive the average-case hardness over $\textit{both}$ randomly chosen outcomes and randomly chosen circuits over the shallow circuit architectures, by establishing a worst-to-average-case reduction for both outcome and circuit instances. 
In short, this reduction process can be done in two steps: $(i)$ from a given fixed outcome to a randomly chosen collision-free outcome, and $(ii)$ from a worst-case circuit to a randomly chosen circuit over a local random circuit distribution.

For random circuit ensembles, we employ a local random circuit ensemble in the shallow circuit architecture, inspired by the hardness results of random circuit sampling~\cite{bouland2019complexity, movassagh2023hardness, bouland2022noise, kondo2022quantum, krovi2022average}. 
Here, the local random distribution in this context means that each gate composing the circuit is independently chosen Haar random gate; we note that recent experimental setups of Boson Sampling~\cite{zhong2020quantum, zhong2021phase, madsen2022quantum, deng2023gaussian} follow a similar circuit distribution, but with geometrically local architectures.
Additionally, due to the absence of symmetry, the local random circuit ensemble itself does not guarantee the bosonic birthday paradox, which is crucial for the hardness of Boson Sampling when using only collision-free outcomes~\cite{aaronson2011computational}.
Hence, we also include a random permutation circuit at the input of our random circuit to guarantee the bosonic birthday paradox for our random circuit ensemble (see Appendix~\ref{section: appendix: bosonic birthday paradox}). 
In fact, this random permutation circuit can be implemented within the shallow depth regime (see, e.g., Lemma~\ref{permutation}), and can be experimentally realized by preparing a random input configuration instead.

To sum up, we show the average-case hardness over outcomes and circuit instances for shallow circuit architectures composed of geometrically non-local gates and employing the local random circuit ensemble.
We informally present here our average-case hardness result of Boson Sampling in the logarithmic depth regime, for photon number $N$ and mode number $M \propto N^{\gamma}$ with $\gamma \geq 2$.

\begin{theorem}[Informal]\label{averagehardnessinformal}
Consider a Boson Sampling problem with $N$ photons over $M = \Omega(N^{\gamma})$ mode, where $\gamma \geq 2$. 
There exists an $O(\log N)$-depth linear optical circuit architecture $\mathcal{A}$, consisting of $O(N^{\gamma}\log N)$ number of geometrically non-local gates, such that estimating the output probability to within additive error $2^{-O(N^{\gamma + 1}(\log N)^2)}$, with high probability over the choice of outcomes and circuits in $\mathcal{A}$, is \#$\rm{P}$-hard under $\rm{BPP}^{\rm{NP}}$ reduction.
\end{theorem}

Also, since our average-case hardness result regards both the random outcomes and random circuits, it is not straightforward to show the classical simulation hardness of shallow-depth Boson Sampling as in the original Boson Sampling proposal~\cite{aaronson2011computational}.
Accordingly, we show how our average-case hardness result over both the randomly chosen outcomes and circuits leads to the classical simulation hardness argument.
This implies that improving the additive imprecision for our average-case hardness result is the only remaining problem for the fully theoretically guaranteed classical intractability of shallow-depth Boson Sampling.

\begin{theorem}[Informal]
For $M \propto N^{\gamma}$ with $\gamma \geq 2$, if the allowed additive imprecision for the problem in Theorem~\ref{averagehardnessinformal} to be \#$\rm{P}$-hard is improved to $\epsilon = 2^{-(\gamma - 1)N\log N -O(N)}$, the approximate Boson Sampling for the shallow-depth circuit in Theorem~\ref{averagehardnessinformal}, up to constant total variation distance, is classically hard to simulate. 
\end{theorem}

\subsection{Paper organization}

Our paper is organized as follows, which is outlined in Fig.~\ref{fig:outlines}. 
We first define in Sec.~\ref{section:convention} a shallow-depth circuit architecture $(\mathcal{B}\mathcal{B}^*)^{q}$ composed of non-local gates, which we will use throughout our results. 
Next, in Sec.~\ref{Section:worstcase}, we prove the worst-case \#P-hardness of approximating output probability $p_{\bm{s}}(C)$ of a fixed outcome $\bm{s}$ of Boson Sampling, for any circuit $C$ in the shallow circuit architecture previously defined. 
In Sec.~\ref{Section:Averagecase}, we prove the average-case \#P-hardness of approximating output probability $p_{\bm{s}}(U)$ for randomly chosen outcome $\bm{s}$ and randomly chosen circuit $U$ in the shallow circuit architecture, by establishing worst-to-average-case reduction.
We prove in Sec.~\ref{Section:simulationhardness} how our average-case hardness results over both the random outcomes and random circuits lead to the classical simulation hardness.
We also extend our average-case hardness result to the Gaussian Boson Sampling scheme in Sec.~\ref{Section:Gaussian}, and to the lossy Boson Sampling subject to photon loss channels in Sec.~\ref{Section:lossy}. 
In Sec.~\ref{Section:remarks}, we conclude with several remarks.

\begin{figure*}[t]
\includegraphics[width=0.7\linewidth]{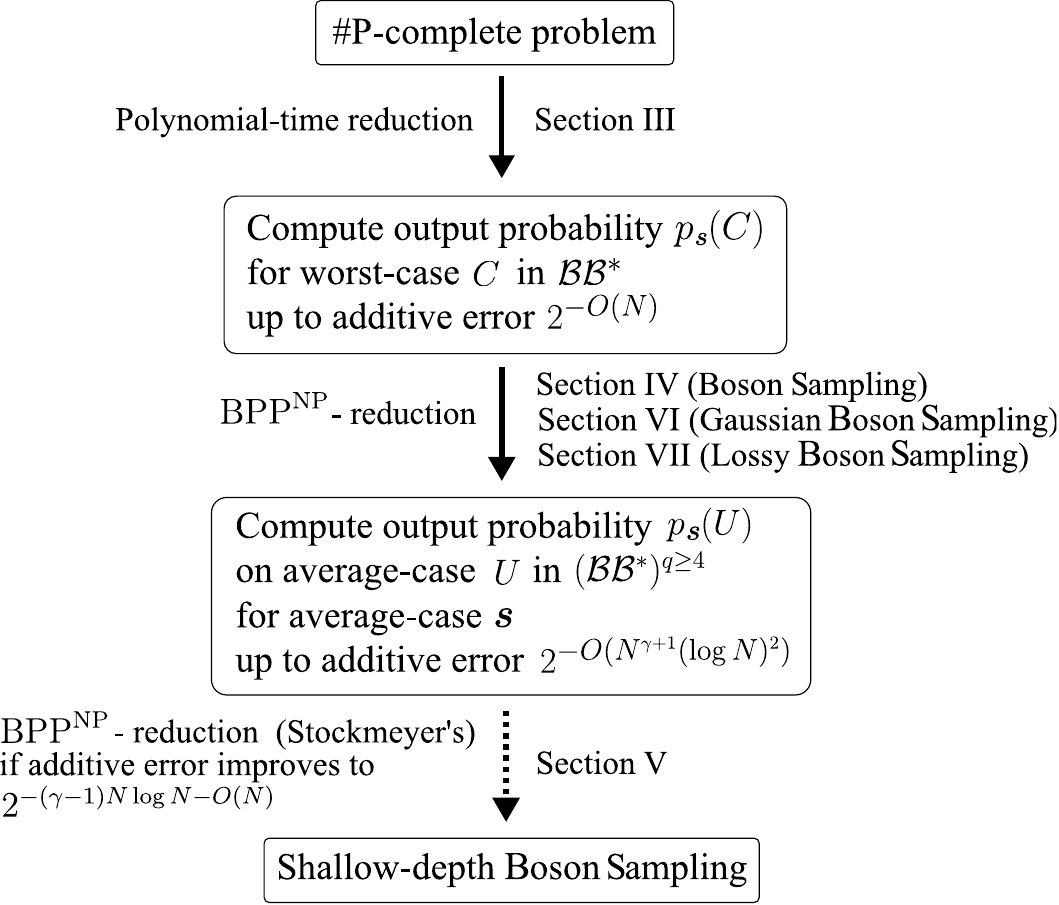}
\caption{Outlines of our result}
\label{fig:outlines}
\end{figure*}

\section{Our settings}\label{section:convention}

\subsection{Notations}

Let us define the total mode number as $M$, where we set $M$ as a power of 2 for simplicity. 
We set the output photon number $N$ polynomially related to $M$ as $M = c_0N^{\gamma}$ for a constant $c_0$ and $\gamma \geq 2$, to ensure that collision-free outcomes dominate the probability weight over all possible $N$-photon outcomes over $M$ modes (i.e., by bosonic birthday paradox; see Refs.~\cite{aaronson2011computational, arkhipov2012bosonic} and Appendix~\ref{section: appendix: bosonic birthday paradox}). 
We use the notation $\bm{s}$ as an $M$-dimensional output configuration vector representing $N$-photon outcome over $M$ modes, such that each element $s_i$ ($i\in [M]$) of $\bm{s}$ denotes the occupied number of photons in the $i$th mode.
Namely, $\bm{s}$ can be expressed as $\bm{s} = (s_1,\dots,s_M)$ where each $s_i \in [N]$ with the constraint $\sum_{i=1}^{M}s_i = N$.
Also, let $\bm{t}$ be an $M$-dimensional input configuration vector corresponding to an $N$-photon input state over $M$ modes, such that $\bm{t} = (t_1,\dots,t_M)$ where now each $t_i \in \{0,1\}$ with $\sum_{i=1}^{M}t_i = N$.

Hereafter, we use the terminology Boson Sampling representing Boson Sampling with Fock-state inputs as in the original Boson Sampling proposal~\cite{aaronson2011computational}; we will explicitly use the terminology Gaussian Boson Sampling for Boson Sampling with Gaussian state inputs~\cite{hamilton2017gaussian}.
We now define $p_{\bm{s}}(C)$ as the output probability of Boson Sampling with a linear optical circuit (unitary) matrix $C$ to obtain an outcome $\bm{s}$ from a fixed input configuration $\bm{t}$. 
More precisely, $p_{\bm{s}}(C)$ can be represented as~\cite{aaronson2011computational}
\begin{align}\label{outputprobability}
    p_{\bm{s}}(C) = \frac{1}{\prod_{i=1}^{M}s_i!}|\text{Per}(C_{\bm{s},\bm{t}})|^2,
\end{align} 
where $C_{\bm{s},\bm{t}}$ is an $N$ by $N$ matrix obtained by taking $s_i$ copies of the $i$th row and $t_j$ copies of the $j$th column of the matrix $C$.
Here, we omit the $\bm{t}$ dependence in the output probability $p_{\bm{s}}(C)$ for simplicity because we will fix the input configuration throughout the paper.

We note that an $M$-mode linear optical circuit can be represented by an $M$ by $M$ unitary matrix in U$(M)$, which unitarily transforms mode operators. 
Specifically, we can represent a single-mode gate (i.e., a phase shifter) as a U$(1)$ matrix to the mode, and a two-mode gate (i.e., a beam splitter) as a U$(2)$ matrix along the modes.
Also, the parallel application of gates can be represented as a unitary matrix with a block matrix form, and the serial application of gates can be represented as matrix multiplication of the unitary matrices. 
Accordingly to this correspondence, throughout this work, we will interchangeably use the terminology `(linear optical) circuit' and `(unitary) matrix'.

We first define the linear optical circuit architecture for a more rigorous analysis of the hardness proof.
	
\begin{definition}[Linear optical circuit architecture]
    The linear optical circuit architecture $\mathcal{A}$ is a linear optical circuit with fixed type (i.e., single- or two-mode) and fixed location of gates, where the coefficients of each gate are not specified. If the coefficients of each gate are specified with unitary matrices (in $\rm{U}$$(1)$ or $\rm{U}$$(2)$), then the circuit and the corresponding unitary matrix are specified. 
\end{definition}

Given the definition of the circuit architecture $\mathcal{A}$, we denote $m$ as the number of gates in the circuit architecture. 
Also, we denote $D$ as the depth of the circuit architecture, where we define a unit depth as a single step of \textit{parallel} application of gates, such that any serial application of gates is not allowed in the unit depth.
Here, we emphasize that single-mode gates represented by U(1) matrices do not contribute to the depth count since they can be absorbed by the nearest two-mode gates represented by U(2) matrices.

In this work, we consider the shallow-depth circuit architecture, specifically in logarithmic depth regime $D = O(\log N)$. 
As we only consider the binary logarithm throughout this work, we will omit the base of the logarithm such that $\log(\cdot)$ implies $\log_{2}(\cdot)$ hereafter.
We now define the shallow linear optical circuit architecture of circuit depth $D = \log M$ with the number of gates $m = \frac{M}{2}\log M$, using the convention used in~\cite{dao2020kaleidoscope, go2024exploring}. 
	
\begin{definition}[Butterfly circuit architecture]\label{butterfly}
    We define butterfly circuit architecture $\mathcal{B}$ as follows: for each layer $L = 1,2,\dots,D = \log M$ of the circuit architecture, allocate two-mode gate between mode number $2^L(j-1) + k$ and $2^L(j-1) + k + 2^{L-1}$, for all $j = 1,2,\dots,2^{D-L}$ and $k = 1,2,\dots,2^{L-1}$. Also, we define inverse butterfly circuit architecture $\mathcal{B}^{*}$ as a butterfly circuit architecture with the inverse sequence of gate application along the depth. 
\end{definition}

\begin{figure*}[t]
\includegraphics[width=0.85\linewidth]{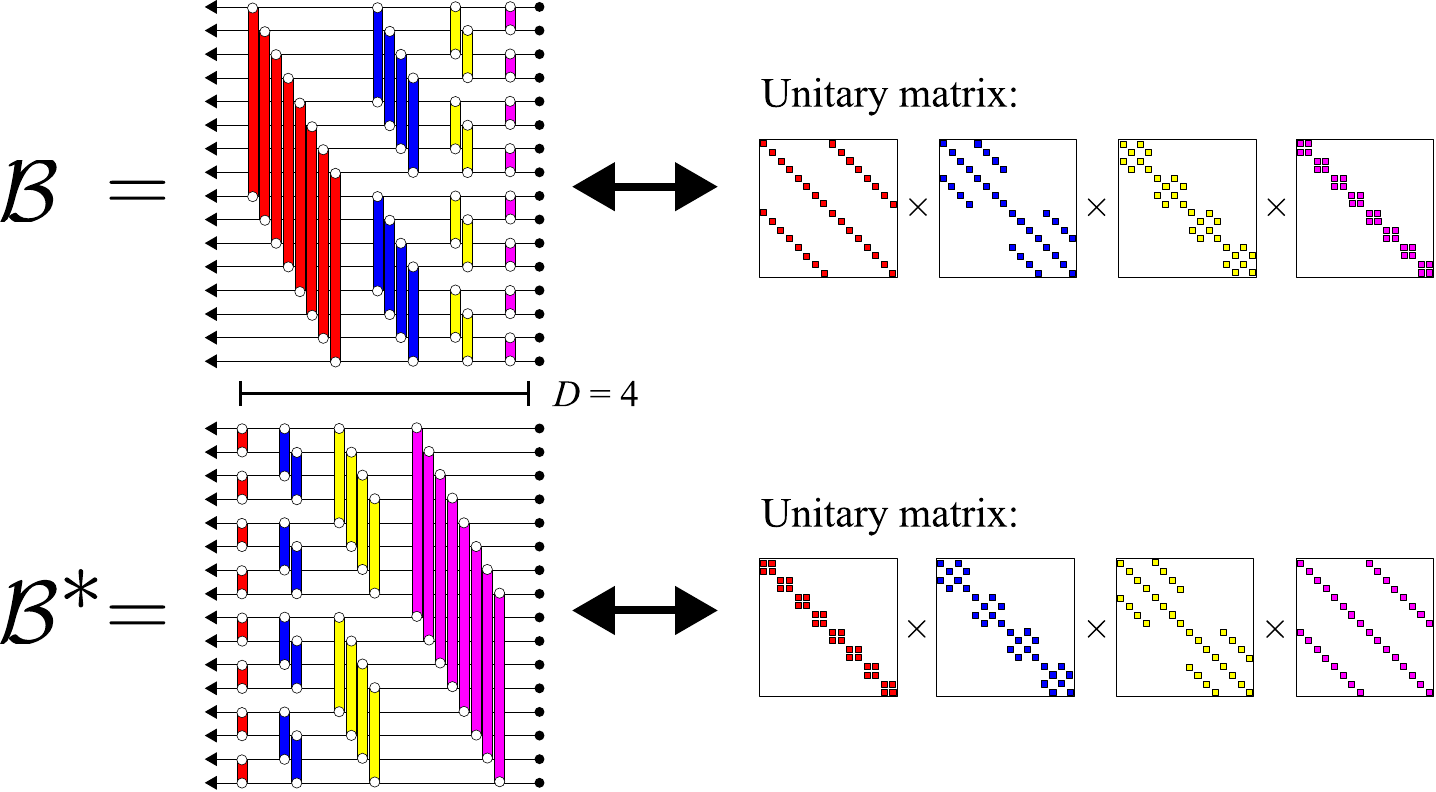}
\caption{Schematics of the butterfly circuit architectures in Definition~\ref{butterfly} and their unitary matrix form, for mode number $M = 2^4 = 16$.}
\label{butterflyfigure}
\end{figure*}

We emphasize that in Definition~\ref{butterfly}, gates are applied in parallel for all $j = 1,2,\dots,2^{D-L}$ and $k = 1,2,\dots,2^{L-1}$. Hence, by definition of the circuit depth, the total depth for $\mathcal{B}^{*}$ (and thus also $\mathcal{B}^{*}$) is given by $D = \log M$. 
We illustrate in Fig.~\ref{butterflyfigure} the circuit architecture $\mathcal{B}$ and $\mathcal{B}^*$ for $2^4 = 16$ modes (such that the circuit depth is $D = 4$ each), and the form of their corresponding unitary matrix. 

Next, we define the Kaleidoscope circuit architecture proposed in~\cite{dao2020kaleidoscope}, using the butterfly circuit architecture defined above.

\begin{definition}[Kaleidoscope circuit architecture]\label{def: kaleidoscope}
    We define Kaleidoscope circuit architecture $\mathcal{B}\mathcal{B}^*$ as a serial application of $\mathcal{B}$ over $\mathcal{B}^{*}$ in Definition~\ref{butterfly}. 
    We define $q$-Kaleidoscope circuit architecture $(\mathcal{B}\mathcal{B}^*)^{q}$ as a repeat of the Kaleidoscope circuit architecture, for a repetition number $q \in \mathbb{N}$.
\end{definition}

Here, the circuit depth of the $q$-Kaleidoscope architecture is $D = 2q\log M$, which is indeed a logarithmic depth $D = O(\log N)$ for $q = O(1)$.
Also, the number of gates for the $q$-Kaleidoscope architecture is $m = qM\log M = O(N^{\gamma}\log N)$ for $q = O(1)$.
Note that, because all the gates are geometrically non-local except for the gates in the first and the last depth for each $\mathcal{B}\mathcal{B}^*$, the number of non-local gates also scales $O(N^{\gamma}\log N)$ for $q = O(1)$.
Throughout this work, we will focus on the $q$-Kaleidoscope circuit architecture with $q = O(1)$ to demonstrate the hardness results for shallow-depth circuits.
The motivation for employing this linear optical circuit architecture is that it enjoys a useful property that is crucial for our analysis, which is that for $M$ a power of $2$, any $M$-mode permutation circuit can be implemented within $\mathcal{B}\mathcal{B}^*$:

\begin{lemma}[Dao et al~\cite{dao2020kaleidoscope}]\label{permutation}
Let $\bm{P}$ be an arbitrary $M \times M$ permutation matrix with $M$ a power of 2. Then $\bm{P}$ can be efficiently implemented in $ \mathcal{B}\mathcal{B}^*$ using $M\log M$-number of two-mode permutation gates, i.e., $\begin{pmatrix}
    0 & 1\\ 
    1 & 0
\end{pmatrix}$ and $\begin{pmatrix}
    1 & 0\\ 
    0 & 1
\end{pmatrix}$. 
\end{lemma}

\subsection{Definitions for random circuit and outcome ensemble}

In this work, we show the average-case hardness over both randomly chosen outcomes and randomly chosen circuits by establishing a worst-to-average-case reduction for both outcomes and circuit instances. 
Here, for the average-case instances, it is crucial to choose a proper circuit ensemble and outcome ensemble to successfully establish the worst-to-average-case reduction.

For the random circuit instances, we consider the \textit{product} of two random circuit ensembles.
The first is the local random circuit ensemble, which is the circuit ensemble with each gate being independently drawn from the local Haar measure.

\begin{definition}[Local random circuit ensemble]\label{randomcircuit}
Let $\mathcal{A}$ be the circuit architecture with $m$ number of gates. We define $\mathcal{H}_{\mathcal{A}}$ as the distribution over circuits with architecture $\mathcal{A}$, whose gates are independently distributed local Haar random matrices $\{H_i\}_{i=1}^{m}$.
\end{definition}

Indeed, this circuit ensemble is similar to the random circuit ensemble employed by the previous random circuit sampling proposals~\cite{bouland2019complexity, movassagh2023hardness, bouland2022noise, kondo2022quantum, krovi2022average} which employ gate-wise Haar random unitary circuits.
Next, we introduce our second random circuit ensemble, which is the random permutation circuit ensemble over the fixed shallow-depth circuit architecture $\mathcal{B}\mathcal{B}^*$.

\begin{definition}[Random permutation circuit ensemble]\label{randompermutation}
We define $\mathcal{P}$ as the uniform distribution over $M$-mode permutation circuits over the fixed Kaleidoscope circuit architecture $\mathcal{B}\mathcal{B}^*$ composed of $M\log M$ number of two-mode permutation gates $\begin{pmatrix}
    0 & 1\\ 
    1 & 0
\end{pmatrix}$ and $\begin{pmatrix}
    1 & 0\\ 
    0 & 1
\end{pmatrix}$. 
Each $\bm{P} \sim \mathcal{P}$ represents an $M$ by $M$ permutation matrix, uniformly chosen from all possible permutation matrices. 
\end{definition}

Note that this random permutation circuit ensemble can be implemented in the shallow-depth circuit architecture $\mathcal{B}\mathcal{B}^*$ by Lemma~\ref{permutation}.
Before proceeding, we briefly remark on the motivation for including the permutation circuit ensemble in our random circuit construction.
Specifically, for any $M$-mode circuit architecture $\mathcal{A}$, including $\mathcal{A} = (\mathcal{B}\mathcal{B}^*)^{q}$ considered in this work, a circuit drawn from the product circuit ensemble $\mathcal{H}_{\mathcal{A}}\mathcal{P}$ guarantees the \textit{bosonic birthday paradox}~\cite{aaronson2011computational, arkhipov2012bosonic}.
The bosonic birthday paradox ensures that the probability of obtaining only collision outcomes is sufficiently large given that the photon number $N$ is sufficiently smaller than the mode number $M$, 
as stated in the following lemma.

\begin{lemma}[Bosonic birthday paradox for the product circuit ensemble]\label{lemma: BBP}
Let $M \geq 2N^2$, and let $B_{M,N}$ be a set of collision outcomes for $M$ modes and $N$ photons. 
Then, for a randomly chosen unitary circuit $U \sim \mathcal{H}_{\mathcal{A}}\mathcal{P}$, the probability to obtain collision outcomes is upper bounded by
\begin{align}
    \E_{U \sim \mathcal{H}_{\mathcal{A}}\mathcal{P}}\left[ \Pr_{\mathcal{D}_{U}}\left[ \bm{s} \in B_{M,N} \right] \right] < \frac{2N^2}{M} ,
\end{align}
where $\mathcal{D}_{U}$ is the output distribution according to $p_{\bm{s}}(U)$ defined in Eq.~\eqref{outputprobability}. 
\end{lemma}

The proof of Lemma~\ref{lemma: BBP} is provided in Appendix~\ref{section: appendix: bosonic birthday paradox}.
Accordingly, the total probability weight of collision outcomes in the Boson Sampling output distribution can be made sufficiently small in the dilute regime $M = \Omega(N^2)$, which is the primary regime of interest in this work. 
Therefore, when using the circuit ensemble $\mathcal{H}_{\mathcal{A}}\mathcal{P}$, the bosonic birthday paradox in Lemma~\ref{lemma: BBP} ensures that collision-free outcomes dominate the probability weight, thereby enabling us to confine the outcome space only to collision-free outcomes.

It is worth emphasizing that, for the bosonic birthday paradox, the circuit ensemble need not be restricted to $\mathcal{H}_{\mathcal{A}}$. 
In fact, regardless of the underlying circuit distribution, the initial random permutation layer $\mathcal{P}$  prior to the circuit distribution ensures the bosonic birthday paradox. 
From an experimental perspective, rather than physically implementing the permutation circuit $\mathcal{P}$, one may instead prepare a uniformly random input configuration, which is operationally equivalent.

\begin{remark}[Conditions for the bosonic birthday paradox]
The bosonic birthday paradox is ensured in the dilute regime as long as the input collision-free configuration is prepared uniformly random by any means.
This observation opens the possibility of improving hardness results for shallow-depth Boson Sampling by considering more general circuit ensembles beyond $\mathcal{H}_{\mathcal{A}}$. 
\end{remark}



For the random outcome instances, as we can now only focus on collision-free outcomes for the hardness argument, we also define the random outcome ensemble as the uniform distribution over collision-free outcomes.

\begin{definition}[Random collision-free outcome ensemble]\label{randomoutcome}
We define $\,\mathcal{G}_{M,N}$ as the uniform distribution over $\binom{M}{N}$ number of collision-free outcomes of Boson Sampling with $M$ modes and $N$ photons.
Each outcome $\bm{s} \sim \mathcal{G}_{M,N}$ is an $M$-dimensional output configuration vector representing the collision-free outcome, such that $\bm{s} = (s_1,\dots,s_M)$ where each $s_i \in \{0,1\}$ with $\sum_{i=1}^{M}s_i = N$.
\end{definition}

\section{Worst-case hardness of output probability estimation}\label{Section:worstcase}
	
In this section, we find the worst-case \#P-hardness of output probability estimation of shallow-depth linear optical circuits in the circuit architecture $\mathcal{B}\mathcal{B}^*$ defined in Definition~\ref{def: kaleidoscope}, for a fixed input $\bm{t}$ and a fixed collision-free output $\bm{s}_0$, to within a certain additive imprecision.
Our worst-case hardness result for our shallow linear optical circuit architecture can be represented as follows.
	
\begin{theorem}[Worst-case hardness]\label{worstcase}
    Given $M = \Omega(N)$, for a fixed collision-free outcome $\bm{s}_0$, approximating an output probability $p_{\bm{s}_0}(C_0)$ to within additive imprecision $2^{-O(N)}$ for any $C_0$ over linear optical circuit architecture $\mathcal{B}\mathcal{B}^*$ is \#$\rm{P}$-hard in the worst case.
\end{theorem}

We briefly sketch the proof of our worst-case hardness result for the shallow circuit architecture $\mathcal{B}\mathcal{B}^*$. 
The proof is based on the result by~\cite{brod2015complexity}, which showed the simulation hardness of exact Boson Sampling with constant-depth linear optical circuits. 
Specifically, as shown in~\cite{brod2015complexity}, based on measurement-based quantum computation (MBQC) scheme, there exists a constant-depth linear optical circuit that can simulate an arbitrary given quantum circuit using post-selection.
To extend this result to our case, we show that the constant-depth worst-case MBQC circuit in~\cite{brod2015complexity} can be embedded in our circuit architecture $\mathcal{B}\mathcal{B}^*$. 
Combining these arguments, additively approximating the output probability of any quantum circuit can be reduced to additively approximating the output probability of any circuit in $\mathcal{B}\mathcal{B}^*$, with imprecision blowup up to the inverse of post-selection probability. 
Using the fact that the additive approximation of any quantum circuit is \#P-hard for certain additive imprecision~\cite{kondo2022quantum}, we can obtain the worst-case hardness of our shallow circuit architecture $\mathcal{B}\mathcal{B}^*$.

\begin{proof}[Proof of Theorem~\ref{worstcase}]
    See Appendix~\ref{proof of worstcase hardness}. 
\end{proof}


We conclude this section by noting that, it is not necessary to derive worst-case hardness of shallow-depth Boson Sampling exclusively for the circuit architecture $\mathcal{BB}^*$. 
That is, embedding the constant-depth worst-case MBQC circuit~\cite{brod2015complexity} into the logarithmic-depth $\mathcal{BB}^*$ circuit, as done in our proof of Theorem~\ref{worstcase}, is not essential for the worst-case hardness of shallow-depth Boson Sampling in principle (and, as will be discussed later, not even for the worst-to-average-case reduction).
Although not mandatory, however, since we later make use of $\mathcal{BB}^*$ to implement the global permutation during the worst-to-average-case reduction (due to the absence of hiding), we establish worst-case hardness within $\mathcal{BB}^*$ so as to simplify the overall circuit architecture to the repeated structure $(\mathcal{BB}^*)^{q}$ defined in Definition~\ref{def: kaleidoscope}.


\section{Average-case hardness of output probability estimation}\label{Section:Averagecase}
	
In this section, we introduce our average-case hardness result of approximating output probabilities of shallow-depth Boson Sampling.
We focus on $q$-Kaleidoscope circuit architecture $(\mathcal{B}\mathcal{B}^*)^{q}$ in Definition~\ref{def: kaleidoscope} with $q = O(1)$, whose circuit depth is $D = 2q\log M$ and gate number is $m = qM\log M$. 
By Theorem~\ref{worstcase}, $\mathcal{B}\mathcal{B}^*$ has the worst-case hardness, and thus  $(\mathcal{B}\mathcal{B}^*)^{q}$ with $q \ge 1$ also has the worst-case hardness by adding trivial gates (identity gates) over the worst-case circuit in $\mathcal{B}\mathcal{B}^*$.

As introduced, our main strategy for the average-case hardness is the worst-to-average-case reduction, using the previous worst-case hardness result in Theorem~\ref{worstcase}. 
In other words, we show that if we can well-estimate the output probability $\it{on\, average}$, we can also well-estimate the worst-case output probability in Theorem~\ref{worstcase}, which means that well-estimating the average-case output probability is also \#P-hard.
{Here, as we have stated previously, our average-case approximation regards both the randomly chosen outcome and the randomly chosen circuit.}

We now formally state our average-case hardness result for shallow-depth linear optical circuit architecture $(\mathcal{B}\mathcal{B}^*)^{q\ge 4}$ with $q = O(1)$.

\begin{theorem}[Average-case hardness]\label{averagehardness}
The following problem is \#$\rm{P}$-hard under a $\rm{BPP}^{\rm{NP}}$ reduction: for any constant $\delta, \eta \ge 0$ with $\delta + \eta < \frac{1}{4}$, {on input an $O(\log N)$-depth random circuit $U \sim \mathcal{H}_{\mathcal{A}}\mathcal{P}$ with $\mathcal{A} = (\mathcal{B}\mathcal{B}^*)^{q\ge 3}$ and a random collision-free outcome $\bm{s} \sim \mathcal{G}_{M,N}$}, compute the output probability $p_{\bm{s}}(U)$ within additive imprecision $\epsilon = 2^{-O(N^{\gamma+1}(\log N)^2)}$, with probability at least $1-\delta$ over the choice of $U$ for at least $1-\eta$ over the choice of $\bm{s}$. 
\end{theorem}

We sketch the proof of Theorem~\ref{averagehardness} in the following by briefly describing the worst-to-average-case reduction process; we leave in Appendix~\ref{proof of averagecase hardness} a detailed proof of Theorem~\ref{averagehardness}.
{Now, we clarify how to establish the reduction: 
The goal is to well-estimate worst-case output probability $p_{\bm{s}_0}(C_0)$ for a fixed collision-free outcome $\bm{s}_0$ and worst-case circuit $C_0$ in $(\mathcal{B}\mathcal{B}^*)^{q \geq 1}$ (Theorem~\ref{worstcase}), given oracle access to well-estimated values $p_{\bm{s}}(U)$ for average-case collision-free outcomes over $\bm{s} \sim \mathcal{G}_{M,N}$ and average-case circuits over $U \sim \mathcal{H}_{\mathcal{A}}\mathcal{P}$ with $\mathcal{A} = (\mathcal{B}\mathcal{B}^*)^{q\ge 3}$ (in other words, $U$ is a randomly chosen circuit over $(\mathcal{B}\mathcal{B}^*)^{q\ge 4}$). }

{
First, since our average-case argument regards both outcomes and circuits, we begin by describing how to \textit{fix} the outcome in establishing the worst-to-average-case reduction.
The remaining problem is then to establish a ``circuit-level" worst-to-average-case reduction, from a worst-case circuit to randomly chosen average-case circuits for a fixed output probability, similarly to previous hardness results of random circuit sampling~\cite{bouland2019complexity, movassagh2023hardness, bouland2022noise, kondo2022quantum, krovi2022average}. }
To fix the outcome, our strategy is to randomly \textit{permute} both a given worst-case circuit $C_0$ in $(\mathcal{B}\mathcal{B}^*)^{q \geq 1}$ and a given collision-free outcome $s_0$. 
That is, we sample random $M$-mode permutation $\bm{P}_0$ uniformly, and permute the worst-case circuit $C_0$ and the fixed outcome $\bm{s}_0$ in Theorem~\ref{worstcase} equally with $\bm{P}_0$.
Then, the permuted outcome $\bm{s} = \bm{P}_0\bm{s}_0$ essentially follows the random outcome ensemble $\bm{s} \sim \mathcal{G}_{M,N}$, and the permuted circuit $\bm{P}_0C_0$ is in $(\mathcal{B}\mathcal{B}^*)^{q \geq 2}$ by Lemma~\ref{permutation}.
Note that the fixed worst-case output probability $p_{\bm{s}_0}(C_0)$ is now equivalent to $p_{\bm{s}}(\bm{P}_0C_0)$, and thus our new goal is to estimate $p_{\bm{s}}(\bm{P}_0C_0)$ value, given estimated average-case values of $p_{\bm{s}}(U)$.


{
Next, we describe how to estimate $p_{\bm{s}}(\bm{P}_0C_0)$ given average-case estimation values $p_{\bm{s}}(U)$ for $U \sim \mathcal{H}_{\mathcal{A}}\mathcal{P}$ with $\mathcal{A} = (\mathcal{B}\mathcal{B}^*)^{q\ge 3}$.
For the permutation part $\mathcal{P}$ in the average-case circuit ensemble $\mathcal{H}_{\mathcal{A}}\mathcal{P}$, we newly sample an $M$-mode random permutation $\bm{P}_{1} \sim \mathcal{P}$ again, and using the sampled $\bm{P}_1$ we set the circuit $\bm{P}_0C_0\bm{P}_{1}^{-1}$ as a ``revised" worst-case circuit in $(\mathcal{B}\mathcal{B}^*)^{q \geq 3}$.
For the local random circuit part $\mathcal{H}_{\mathcal{A}}$ in the average-case circuit ensemble $\mathcal{H}_{\mathcal{A}}\mathcal{P}$, we sample a random circuit $V$ from $\mathcal{H}_{\mathcal{A}}$ and \textit{perturb} each gate of $V$ by each gate of $\bm{P}_0C_0\bm{P}_{1}^{-1}$ parameterized by a constant $\theta \in [0,1]$ using Cayley transform method, and obtain a perturbed circuit $V(\theta)$ (see Appendix~\ref{appendix: section: cayley} for details).
Here, by Cayley transform, $V(\theta)$ lies on a continuous path over unitary matrices interpolating $V(0)$ and $V(1)$, where $\theta = 0$ corresponds to the random circuit drawn from local random circuit ensemble $V(0) \sim \mathcal{H}_{(\mathcal{B}\mathcal{B}^*)^{q\ge 3}}$ and $\theta = 1$ corresponds to the worst-case circuit $V(1) = \bm{P}_0C_0\bm{P}_{1}^{-1}$.
Notably, $V(0)\bm{P}_{1}$ follows the average-case circuit distribution $\mathcal{H}_{\mathcal{A}}\mathcal{P}$ with $\mathcal{A} = (\mathcal{B}\mathcal{B}^*)^{q\ge 3}$, and thus, one can expect that the distribution for $V(\theta)\bm{P}_{1}$ for sufficiently small $\theta$ would not be largely deviated from the average-case circuit distribution. 
Accordingly, one can obtain well-estimated values of $p_{\bm{s}}(V(\theta)\bm{P}_{1})$ for sufficiently small $\theta$ with high probability, given oracle access to the average-case output probabilities (i.e., Theorem~\ref{averagehardness}).

Importantly, we find that $p_{\bm{s}}(V(\theta)\bm{P}_{1})$ can be expressed as a finite-degree rational function $\frac{P(\theta)}{Q(\theta)}$ for $\theta \in [0,1]$, where $Q(\theta)$ is efficiently computable. 
Hence, for sufficiently small $\theta$ values, one can also obtain the well-estimated values for the finite degree \textit{polynomial} $P(\theta)$, by multiplying the well-estimated values of $p_{\bm{s}}(V(\theta)\bm{P}_{1})$ with $Q(\theta)$.
Then, given the estimated values of $P(\theta)$ for different $\theta$ values, one can use \textit{polynomial interpolation} for the polynomial $P(\theta)$, and obtain the estimate of $P(1) = p_{\bm{s}}(V(1)\bm{P}_1)Q(1)$. 
Therefore, one can finally infer the value $p_{\bm{s}}(V(1)\bm{P}_1)$ by multiplying $Q(1)^{-1}$ to the estimate of $P(1)$, which corresponds to the worst-case output probability $p_{\bm{s}}(V(1)\bm{P}_1 ) = p_{\bm{s}}(\bm{P}_0C_0\bm{P}_1^{-1}\bm{P}_1 ) = p_{\bm{s}}(\bm{P}_0C_0) = p_{\bm{s}_0}(C_0)$, concluding the worst-to-average-case reduction.

}

To sum up, for the worst-case output probability $p_{\bm{s}_0}(C_0)$ in Theorem~\ref{worstcase} and the average-case output probability $p_{\bm{s}}(U)$ in Theorem~\ref{averagehardness}, the overall reduction process can be schematized by
\begin{align}\label{eq: sketch of reduction}
     p_{\bm{s}_0}(C_0) =  p_{\bm{s}}(\bm{P}_0C_0) = p_{\bm{s}}(\bm{P}_0C_0\bm{P}_{1}^{-1}\bm{P}_{1} ) \; \xleftarrow[\substack{V \rightarrow \bm{P}_0C_0\bm{P}_{1}^{-1} \\ \text{Cayley transform} \\ \text{\&} \\ 
    \text{polynomial} \\ \text{interpolation}  } ]{} \; p_{\bm{s}}(V\bm{P}_{1}) =  p_{\bm{s}}(U) \approx  \mathcal{O}(\bm{s},U) 
    ,
\end{align}
where
\begin{align}
    &\mathcal{O}: \text{Oracle for average-case output probabilities (Theorem~\ref{averagehardness})}\nonumber \\
    &C_0: \text{Worst-case circuit in } (\mathcal{B}\mathcal{B}^*)^{q\ge 1} \nonumber\\
    &\bm{s}_0: \text{A fixed collision-free outcome} \nonumber\\
    &\bm{s} \sim  \mathcal{G}_{M,N}: \text{Random collision-free outcome (Definition~\ref{randomoutcome})} \nonumber\\
    & V \sim \mathcal{H}_{(\mathcal{B}\mathcal{B}^*)^{q\ge 3}} : \text{Local random circuit over }(\mathcal{B}\mathcal{B}^*)^{q\ge 3} \,\text{(Definition~\ref{randomcircuit})}\nonumber\\
    &\bm{P}_0,\,\bm{P}_{1} \sim \mathcal{P} : \text{Random permutation circuit in } \mathcal{B}\mathcal{B}^* \, \text{(Definition~\ref{randompermutation})} . \nonumber
\end{align}

So far, we have outlined the proof of Theorem~\ref{averagehardness}, describing how to establish worst-to-average-case reduction. 
In Appendix~\ref{appendix: section: cayley}, we introduce in detail the Cayley transform, the circuit perturbation method for the worst-to-average-case reduction in our work. 
Subsequently, in Appendix~\ref{proof of averagecase hardness}, we provide a detailed proof of Theorem~\ref{averagehardness}, including the imprecision analysis we have missed out in the main text (i.e., why the average-case imprecision level $\epsilon = 2^{-O(N^{\gamma+1}(\log N)^2)}$ is necessary to achieve the worst-case imprecision level $2^{-O(N)}$).

\begin{proof}[Proof of Theorem~\ref{averagehardness}]
See Appendix~\ref{proof of averagecase hardness}
\end{proof}

It is worth mentioning that, although we have concentrated on the dilute regime such that $\gamma \geq 2$, our average-case hardness result does not rule out the case $1 \leq \gamma < 2$ because it is derived from the worst-to-average-case reduction and our worst-case hardness result allows $M = \Omega(N)$.
However, since our average-case hardness result only considers collision-free outcomes, it does not directly lead to the sampling hardness result (i.e., quantum advantage claim for shallow-depth Boson Sampling) by only increasing the allowed additive imprecision for the average-case hardness because collision-free outcomes are expected to occupy a small portion of outcomes when $1 \leq \gamma < 2$.
Therefore, to lead to the sampling hardness argument when $1 \leq \gamma < 2$, our average-case hardness result should be extended to also contain collision outcomes for the average-case outcome instances (i.e., instead of $\mathcal{G}_{M, N}$, now uniformly choose over all possible outcomes). 
For a more complete analysis, in Appendix~\ref{appendix: section: collision}, we present a promising approach that can be extended to the average-case hardness result that includes collision outcomes.


We also remark again that one could further reduce the circuit depth by taking the worst-case circuit $C_0$ as the constant-depth MBQC circuit in~\cite{brod2015complexity} directly, without encoding into the logarithmic-depth architecture $(\mathcal{B}\mathcal{B}^*)^{q \geq 1}$ as in Theorem~\ref{worstcase}.
In this case, the overall circuit architecture $\mathcal{A}$ for the average-case hardness consists of the constant-depth circuit (for MBQC) sandwiched between the $\mathcal{BB}^*$ architectures used for implementing required permutations (i.e., $\bm{P}_0$ and $\bm{P}_1$), yielding a slightly shallower circuit architecture than $\mathcal{A} = (\mathcal{BB}^*)^{q \geq 3}$ used in our case.
Nevertheless, since we are already in the logarithmic-depth regime (by implementing permutations with $\mathcal{BB}^*$), we adopt $\mathcal{A} = (\mathcal{BB}^*)^{q \geq 3}$ as our overall circuit architecture for the ensemble $\mathcal{H}_{\mathcal{A}}$ in order to provide a unified description of the circuit architecture required for the hardness of shallow-depth Boson Sampling, as well as a clearer layout for experimental implementations.


\section{Classical simulation hardness of shallow-depth boson sampling}\label{Section:simulationhardness}
	

As we have previously discussed, since our average-case hardness result considers both the random outcomes and random circuits, it is not straightforward to show the classical simulation hardness of shallow-depth Boson Sampling as in the original Boson Sampling proposal~\cite{aaronson2011computational}.
Therefore, in this section, we provide a self-contained analysis of how our average-case hardness result leads to the classical simulation hardness arguments of shallow-depth Boson Sampling.
Specifically, we show that if the allowed additive error in Theorem~\ref{averagehardness} for the hardness is improved to a certain imprecision level, an efficient classical algorithm that can approximately simulate the shallow-depth Boson Sampling is unlikely to exist. 
This emphasizes that improving the imprecision level of the average-case hardness in Theorem~\ref{averagehardness} is a crucial step for the classical intractability of shallow-depth Boson Sampling. 
	
Similarly to Refs.~\cite{aaronson2011computational, bouland2019complexity}, we define an approximate boson sampler as follows.
	
\begin{definition}[Approximate boson sampler]\label{approximatebosonsampler}
	Approximate boson sampler is a classical randomized algorithm that takes input linear optical circuit $C$ and outputs a sample from the output distribution $\mathcal{D}_{C}'$ such that
		\begin{equation}
			|| \mathcal{D}_{C}' - \mathcal{D}_{C} || \le \beta ,
		\end{equation}
		where $\mathcal{D}_{C}$ is the ideal output distribution of the circuit $C$ and $||\cdot||$ represents total variation distance. 
\end{definition}

Given the total variation distance error, the above approximate sampler can have an arbitrarily large additive error for a fixed output probability.  
Nevertheless, it still has a comparably small additive error for \textit{most} of the output probabilities due to Markov's inequality.
Accordingly, finding the average-case solution of the output probability of the ideal sampler over randomly chosen collision-free outcome $\bm{s} \sim \mathcal{G}_{M,N}$, up to a certain additive imprecision, is in complexity class $\rm{BPP}^{\rm{NP}}$ by Stockmeyer's theorem~\cite{stockmeyer1985approximation}.

\begin{lemma}[Average-case approximation~\cite{aaronson2011computational}]\label{aaronsonandarkhipov}
If there exists an approximate boson sampler $\mathcal{S}$ with total variation distance $\beta$, then for any linear optical circuit $C$, the following problem is in $\rm{BPP}^{\rm{NP}^{\mathcal{S}}}$: find the average-case approximate solution $\tilde{p}_{\bm{s}}(C)$ of $p_{\bm{s}}(C)$, which satisfies
\begin{equation}\label{averagecaseapproximation}
    \Pr_{\bm{s} \sim \mathcal{G}_{M,N}}\left[|\tilde{p}_{\bm{s}}(C) - p_{\bm{s}}(C)| \ge \frac{\kappa}{\binom{M}{N}}\right] \le \xi,
\end{equation}
where $\bm{s}$ is over all collision-free outcomes, and $\kappa, \xi > 0$ are the fixed error parameters satisfying $\beta = \kappa \xi/{12}$.
\end{lemma}

We leave the proof of Lemma~\ref{aaronsonandarkhipov} in Appendix~\ref{proof of averagecase approximation} for a more self-contained analysis. 
The complexity $\rm{BPP}^{\rm{NP}}$ is known to be inside the finite level of PH, i.e., $\rm{BPP}^{\rm{NP}} \subseteq PH$. 
Also, by Toda's theorem~\cite{toda1991pp}, PH problems can be solved given the ability to solve any \#P problem, i.e., $\rm{BPP}^{\rm{NP}} \subseteq 
\rm{PH} \subseteq \rm{P}^{\rm{\#P}}$.
If finding the average-case solution of output probabilities of sampler $\mathcal{S}$ is \#P-hard, then $\rm{P}^{\rm{\#P}} \subseteq \rm{BPP}^{\rm{NP}^{\mathcal{S}}}$. 
Therefore, if an efficient classical algorithm exists that can simulate $\mathcal{S}$, then $\rm{P}^{\rm{\#P}} \subseteq \rm{BPP}^{\rm{NP}}$ which implies the collapse of PH.
Consequently, under the assumption of the non-collapse of the PH, there is no efficient classical algorithm capable of simulating $\mathcal{S}$. 
	

Based on the above arguments, we show that for the case where the allowed additive imprecision of Theorem~\ref{averagehardness} for the hardness can be improved, then it is classically hard to simulate shallow-depth Boson Sampling within a constant total variation distance.

\begin{theorem}\label{simulationhardness}
    Suppose that the allowed additive imprecision for the problem in Theorem~\ref{averagehardness} to be \#$\rm{P}$-hard can be improved to $\epsilon = 2^{-(\gamma - 1)N\log N -O(N)}$. Then, {in the dilute regime such that $\gamma \geq 2$,} the efficient classical simulation of approximate boson sampler $\mathcal{S}$ with respect to circuits from the shallow architecture {$\mathcal{A} = (\mathcal{B}\mathcal{B}^*)^{q\ge 4}$} implies the collapse of $\rm{PH}$. 
\end{theorem}
\begin{proof}
We establish a reduction from the problem in Theorem~\ref{averagehardness} with allowed additive error $\epsilon = 2^{-(\gamma - 1)N\log N -O(N)}$ to the problem in Lemma~\ref{aaronsonandarkhipov}.
Let $\mathcal{O}$ be an oracle that solves the problem in Lemma~\ref{aaronsonandarkhipov}, i.e., on input a circuit $C$ and a randomly chosen collision-free outcome $\bm{s} \sim \mathcal{G}_{M,N}$, $\mathcal{O}$ outputs an estimate of output probability $p_{\bm{s}}(C)$ up to additive imprecision $\kappa\binom{M}{N}^{-1}$, with probability at least $1 - \xi$ over outcomes for any circuit. 
For convenience, we refer to the outcomes which $\mathcal{O}$ estimates with error larger than $\kappa\binom{M}{N}^{-1}$ as \textit{bad} outcomes, such that the portion of bad outcomes over possible collision-free outcomes is at most $\xi$ for any circuit.

{Note that, Lemma~\ref{aaronsonandarkhipov} holds for a randomly chosen input circuit drawn from an arbitrary circuit ensemble, such that the success probability of $\mathcal{O}$ is at least $1 - \xi$ over randomly chosen outcomes and circuits. 
However, for the randomly chosen circuit input, bad outcomes can vary with the circuit instances, as the sampler $\mathcal{S}$ has the freedom to choose 
its error distribution according to the input circuit (e.g., $C$ and $\bm{P}C$ would have different bad outcomes for permutation $\bm{P}$). 
Nevertheless,} no matter how the bad outcomes vary with circuit instances, $\mathcal{O}$ succeeds at least $1 - \frac{\xi}{\eta}$ fraction over circuit instances for at least $1 - \eta$ fraction of the outcomes, for any $\eta$ satisfying $\xi < \eta < 1$.
Otherwise, the failure probability over outcomes and circuits would exceed $\xi$, which contradicts the proposition that the success probability of $\mathcal{O}$ is at least $1 - \xi$ over randomly chosen outcomes and circuits.
As the above argument holds for any random circuit families, we choose the random circuit distribution as {$\mathcal{H}_{\mathcal{A}}\mathcal{P}
$ with the shallow architecture $\mathcal{A} = (\mathcal{B}\mathcal{B}^*)^{q\ge 3}$. }

To sum up, on input a random circuit $H \sim \mathcal{H}_{\mathcal{A}}$ and a random outcome $\bm{s} \sim \mathcal{G}_{M,N}$, the oracle $\mathcal{O}$ estimates the output probability $p_{\bm{s}}(H)$ up to imprecision $\kappa\binom{M}{N}^{-1}$, with probability at least $1-\frac{\xi}{\eta}$ over the choice of $H$ for at least $1-\eta$ over the choice of $\bm{s}$.
Here, the additive imprecision can be bounded as 
	\begin{align}
		\kappa\binom{M}{N}^{-1} &= \kappa \frac{N!(M-N)!}{M!} = 2^{-(\gamma - 1)N\log N -O(N)},
	\end{align}
for constant $\beta$ and so as $\kappa$, where we used the relation $M = c_0N^{\gamma}$ with a constant $c_0$ and $\gamma \ge 1$. 
By setting $\eta$ and $\xi$ small constant satisfying $\eta + \frac{\xi}{\eta} < \frac{1}{4}$, we can solve the problem in Theorem~\ref{averagehardness} up to additive imprecision $\epsilon = 2^{-(\gamma - 1)N\log N -O(N)}$ using the oracle $\mathcal{O}$.
Hence, assuming that the above problem is \#P-hard under $\rm{BPP}^{\rm{NP}}$ reduction, we can obtain the complexity-theoretical relation $\rm{P}^{\#P} \subseteq \rm{BPP}^{\rm{NP}^{\mathcal{S}}}$, which implies the collapse of PH if $\mathcal{S}$ with respect to shallow-depth circuit architecture {$(\mathcal{B}\mathcal{B}^*)^{q\ge 4}$} can be done in classical polynomial time. 
This completes the proof.
\end{proof}

{

We lastly remark that, because our average-case hardness result in Sec.~\ref{Section:Averagecase} solely relies on collision-free outcomes, it must be ensured that these collision-free outcomes occupy a sufficient portion of all possible outcomes of Boson Sampling (for more details, see Appendix~\ref{section: appendix: bosonic birthday paradox}). 
Hence, we focus on the dilute regime $\gamma \geq 2$ for the classical hardness result in this section, to ensure the dominance of collision-free outcomes. 
To similarly obtain the classical hardness result in the saturated regime $1 \leq \gamma < 2$, extending our average-case hardness result for the collision outcomes is crucial.

}

\section{Average-case hardness for shallow-depth Gaussian Boson Sampling}\label{Section:Gaussian}
In this section, we show that our hardness results of the shallow-depth Boson Sampling can be generalized to the Gaussian Boson Sampling scheme ~\cite{hamilton2017gaussian}. 
Our specific setup for the Gaussian Boson Sampling is as follows.
Let the total mode number $M$ of the circuit be a power of 2, and now the input state is an $M$ product of single-mode squeezed vacuum (SMSV) state $\ket{\text{SMSV}}^{\otimes M}$ with equal squeezing parameter $r$ and equal squeezing direction.
Also, let us define the output mean photon number as an integer $N$ (i.e., $N = M\sinh^{2}r$) where $M$ and $N$ are polynomially related as $M = c_1N^{\gamma}$ for a constant $c_1$ and $\gamma \ge 2$. 
We define $q_{\bm{s}}(C)$ as an output probability of the Gaussian Boson Sampling, for an $N$-photon outcome $\bm{s}$ from an $M$-mode linear optical circuit matrix $C$ on input $M$ SMSV states. 
For collision-free outcome $\bm{s}$, $q_{\bm{s}}(C)$ can be expressed as~\cite{hamilton2017gaussian} 
\begin{equation}\label{outputprobabilitygaussian}
	q_{\bm{s}}(C) = \left| \bra{\bm{s}}\hat{\mathcal{U}}(C)\ket{\text{SMSV}}^{\otimes M}\right|^2 = \frac{\tanh^{N}r}{\cosh^{M}r}|\text{Haf}((CC^T)_{\bm{s}})|^2,
\end{equation}
where $\ket{\bm{s}}$ is an $M$-mode Fock-state corresponding to the outcome $\bm{s}$, $\hat{\mathcal{U}}(C)$ is a unitary operator corresponding to the circuit $C$, and $(CC^T)_{\bm{s}}$ is an $N$ by $N$ matrix obtained by taking $s_i$ copies of the $i$th row and column of the matrix $CC^T$.
	
Using the above settings, we first prove the worst-case hardness of Gaussian Boson Sampling for a fixed outcome $\bm{s}$, with the shallow-depth circuit architecture $\mathcal{B}\mathcal{B}^*$.
	
\begin{theorem}\label{gaussianworstcase}
	Approximating the output probability $q_{\bm{s}_0}(C)$ of Gaussian Boson Sampling to within additive error $2^{-\frac{\gamma - 1}{2}N\log N -O(N)}$ for any $C$ over linear optical circuit architecture $\mathcal{B}\mathcal{B}^*$ is \#$\rm{P}$-hard in the worst case.
\end{theorem}
\begin{proof}
We establish a reduction from the worst-case hardness of Boson Sampling in Theorem~\ref{worstcase} to the problem in Theorem~\ref{gaussianworstcase}. 
Let $p_{\bm{s}_0}(C_0)$ be the output probability of a fixed input and output $\bm{s}_0$ of Boson Sampling in Theorem~\ref{worstcase}, for mode number $M_0$, photon number $N_0$, and the circuit $C_0$ in $M_0$ mode circuit architecture $\mathcal{B}\mathcal{B}^*$. 
In the following, we show that $p_{\bm{s}_0}(C_0)$ can be efficiently reduced to the output probability $q_{\bm{s}}(C)$ of Gaussian Boson Sampling, for mode number $M = 2M_0$ and mean photon number $N = 2N_0$, with output $\bm{s}$ and circuit $C$ determined by $\bm{s}_0$ and $C_0$ each.
		
Our strategy is to employ the scheme in Ref.~\cite{lund2014boson}, which used $M_0$ product of equally squeezed two-mode squeezed vacuum (TMSV) state as an input state to perform $M_0$ mode Boson Sampling task.
Specifically, a single TMSV state with squeezing parameter $r$ can be represented as 
\begin{equation}
\ket{\text{TMSV}} = \frac{1}{\cosh r}\sum_{n=0}^{\infty} \tanh^{n}r\ket{n}\ket{n},
\end{equation}
and thus $M_0$ product of the TMSV state is 
\begin{align}\label{productofTMSV}
\begin{split}
\ket{\text{TMSV}}^{\otimes M_0} &= \frac{1}{\cosh^{M_0} r} \left( \sum_{n=0}^{\infty} \tanh^{n}r\ket{n}_{(1)}\ket{n}_{(2)} \right)^{\otimes M_0} \\
&= \frac{1}{\cosh^{M_0} r} \sum_{n = 0}^{\infty}\tanh^{n}r\sum_{\bm{s}_{n}} \ket{\bm{s}_n}_{(1)}\ket{\bm{s}_n}_{(2)},
\end{split}
\end{align}
where the summation of $\bm{s}_n$ is over all possible configurations of Fock-state with a total $M_0$ mode and $n$ photon.

For each mode in the given $M_0$ mode circuit $C_0$, one-half of the TMSV state (i.e., subscript (2) in Eq.~\eqref{productofTMSV}) is input into it, and the other half of each state (i.e., subscript (1) in Eq.~\eqref{productofTMSV}) is sent directly to a photon counter. 
By setting each $\ket{\bm{s}_{\text{in}}}$ and $\ket{\bm{s}_{\text{out}}}$ as a total $M_0$ mode and total $N_0$ photon Fock-state, the output probability can be represented as 
\begin{equation} \left|\bra{\bm{s}_{\text{in}}}_{(1)}\bra{\bm{s}_{\text{out}}}_{(2)}\hat{\mathcal{U}}_{(2)}(C_0)\ket{\text{TMSV}}^{\otimes M_0}\right|^2 = \frac{\tanh^{2N_0}r}{\cosh^{2M_0}r} \left|\bra{\bm{s}_{\text{in}}}\hat{\mathcal{U}}(C_0)\ket{\bm{s}_{\text{out}}}  \right|^2,
\end{equation}
which is the output probability of $M_0$ mode and $N_0$ photon Boson Sampling in circuit $C_0$, with an additional multiplicative factor. 
		
Note that two $M_0$ mode $\mathcal{B}\mathcal{B}^*$ architecture can be embedded in the middle of an $M = 2M_0$ mode $\mathcal{B}\mathcal{B}^*$ architecture.
Accordingly, we define a circuit $C$ in an $M$-mode $\mathcal{B}\mathcal{B}^*$ by embedding the given $M_0$ mode circuit $C_0$ in one side of $\mathcal{B}\mathcal{B}^*$, setting gates located right in front of the input ports as balanced beam splitters, and setting the remaining gates as identity gates; 
we leave in Fig.~\ref{butterflyforgaussian} an illustration of $M = 16$ mode circuit $C$ for more clarity. 
Here, the input $M$ SMSV states with squeezing parameter $r$ combined with the balanced beam splitters at the front becomes $M_0$ TMSV states with squeezing parameter $r$.
Therefore, our overall setup exactly follows the scheme in Ref.~\cite{lund2014boson}, such that the first $M_0$ mode is the photon counter sector to determine the input configuration of Boson Sampling, and the last $M_0$ mode is to simulate the Boson Sampling for the given circuit $C_0$.

\begin{figure*}[t]
\includegraphics[width=0.75\linewidth]{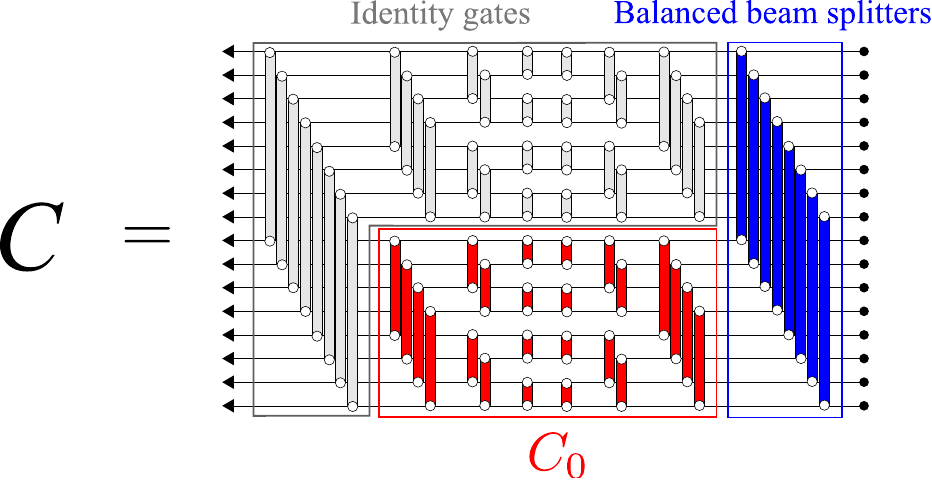}
\caption{Schematics of an mode number $M = 16$ circuit $C$ in $\mathcal{B}\mathcal{B}^*$ which contains a given $M_0 = 8$ mode circuit $C_0$ in $\mathcal{B}\mathcal{B}^*$ }
\label{butterflyforgaussian}
\end{figure*}
		
We also define an $M$-dimensional vector $\bm{s}$ as a serial concatenation of two $\bm{s}_0$ vectors, so that $\bm{s}$ represents an $N$-photon outcome for $M$ modes.
Then the output probabilities $p_{\bm{s}_0}(C_0)$ and $q_{\bm{s}}(C)$ are related as 
\begin{align}
\begin{split}
    q_{\bm{s}}(C) &= \left| \bra{\bm{s}}\hat{\mathcal{U}}(C)\ket{\text{SMSV}}^{\otimes M} \right|^2 \\
    &= \left|\bra{\bm{s}_{0}}_{(1)}\bra{\bm{s}_{0}}_{(2)}\hat{\mathcal{U}}_{(2)}(C_0)\ket{\text{TMSV}}^{\otimes M_0}\right|^2 \\
    &= \frac{\tanh^{2N_0}r}{\cosh^{2M_0}r} p_{\bm{s}_0}(C_0).
\end{split}
\end{align}

Hence, approximating the output probability $p_{\bm{s}_0}(C_0)$ can be reduced to approximating the output probability $q_{\bm{s}}(C)$ of Gaussian Boson Sampling, with a blowup in the additive imprecision. 
The size of the additive imprecision blowup is
\begin{align}
	\frac{\cosh^{2M_0}r}{\tanh^{2N_0}r} 
	= \left( \frac{M+N}{M} \right)^{M_0 + N_0}\left( \frac{M}{N} \right)^{N_0} 
	= 2^{\frac{\gamma-1}{2}N\log N + O(N) },
\end{align}
using the relation $N = M\sinh r$ and $M \propto N^{\gamma}$. 
Since the allowed additive error for the worst-case hardness of $p_{\bm{s}_0}(C_0)$ is $2^{-O(N)}$, the allowed additive error for the reduction is $2^{-O(N)}2^{-\frac{\gamma-1}{2}N\log N - O(N) }$ = $2^{-\frac{\gamma-1}{2}N\log N - O(N) }$. 
This completes the proof. 
\end{proof}
	
Using the results of Theorem~\ref{gaussianworstcase} and the previous proof of the average-case hardness of Boson Sampling in Theorem~\ref{averagehardness}, it is straightforward to find the average-case hardness of Gaussian Boson Sampling, for a randomly chosen $N$-photon outcomes $\bm{s} \sim \mathcal{G}_{M,N}$ and a randomly chosen circuit $U \sim \mathcal{H}_{\mathcal{A}}$ in shallow-depth architecture $\mathcal{A} = (\mathcal{B}\mathcal{B}^*)^{q\ge 2}$.

\begin{theorem}\label{gaussianaveragehardness}
    The following problem is \#$\rm{P}$-hard under a $\rm{BPP}^{\rm{NP}}$ reduction: for any constant $\delta, \eta \ge 0$ with $\delta + \eta < \frac{1}{4}$, on input a $O(\log N)$-depth random circuit $U \sim \mathcal{H}_{\mathcal{A}}$ with $\mathcal{A} = (\mathcal{B}\mathcal{B}^*)^{q\ge 2}$ and a random collision-free outcome $\bm{s} \sim \mathcal{G}_{M,N}$, compute the output probability $q_{\bm{s}}(U)$ of Gaussian Boson Sampling within additive imprecision $\epsilon = 2^{-O(N^{\gamma+1}(\log N)^2)}$, with probability at least $1-\delta$ over the choice of $U$ for at least $1-\eta$ over the choice of $\bm{s}$. 
\end{theorem}

It is worth emphasizing that for our setup of the Gaussian Boson Sampling, the bosonic birthday paradox is already guaranteed in the dilute regime $M = \Omega(N^2)$, in contrast to the Boson Sampling case.
More specifically, the bosonic birthday paradox requires the permutation symmetry of inputs over the unitary circuits, in the sense that the input pattern should be invariant under the permutations (see Appendix~\ref{section: appendix: bosonic birthday paradox} for more details). 
Since we are considering ``full mode" input, our setup already satisfies this required input symmetry. 
To see this input symmetry clearly, one can easily check that for any unitary circuit $U$ and permutation matrix $\bm{P}$, the output probability of Gaussian Boson Sampling for $U\bm{P}$ is equivalent to
\begin{align}
    q_{\bm{s}}(U\bm{P}) =  \frac{\tanh^{N}r}{\cosh^{M}r}|\text{Haf}((U\bm{P}\bm{P}^{T}U^T)_{\bm{s}})|^2 =  \frac{\tanh^{N}r}{\cosh^{M}r}|\text{Haf}((UU^T)_{\bm{s}})|^2 = q_{\bm{s}}(U).
\end{align}
Therefore, we do not necessarily have to include the random permutation part $\mathcal{P}$ for the average-case hardness result; only the local random circuit for $\mathcal{A} = (\mathcal{B}\mathcal{B}^*)^{q\ge 2}$ suffices.

\begin{proof}[Proof of Theorem~\ref{gaussianaveragehardness}]
The procedure is essentially the same as the proof of Theorem~\ref{averagehardness}, namely, establishing a worst-to-average-case reduction from the problem in Theorem~\ref{gaussianworstcase} to the problem in Theorem~\ref{gaussianaveragehardness}: 
{First, given worst-case circuit $C\in (\mathcal{B}\mathcal{B})^{q\geq1}$ and outcome $\bm{s}_0$ depicted in Theorem~\ref{gaussianworstcase}, sample permutation $\bm{P}$, such that the revised worst-case circuit and outcome become $C_{p} = \bm{P}C \in (\mathcal{B}\mathcal{B})^{q\geq2}$ and $\bm{s} = \bm{P}\bm{s}_0  \sim \mathcal{G}_{M,N}$ respectively. 
Then, we sample $V(\theta)$ from local random circuit ensemble with circuit architecture $\mathcal{A} = (\mathcal{B}\mathcal{B})^{q\geq2}$ and perturb by $C_{p} = \bm{P}C$ via Cayley transform, such that $V(\theta)$ follows the distribution $\mathcal{H}_{\mathcal{A},\theta}^{C_p}$. 
Given access to the oracle that solves the problem in Theorem~\ref{gaussianaveragehardness}, we infer the output probability $q_{\bm{s}}(V(\theta))$ using different values of $\theta$, and thus obtain the worst-case output probability $q_{\bm{s}}(V(1)) = q_{\bm{s}}(C_{p}) = q_{\bm{s}_0}(C)$
(for more details, see Appendix~\ref{proof of averagecase hardness}, where we prove the average-case hardness for the Boson Sampling case).

}

Hence, the only different part for the Gaussian Boson Sampling case is the functional form of the output probability $q_{\bm{s}}(V(\theta))$ parameterized by $\theta$ via Cayley transformation. 
Accordingly, we show that $q_{\bm{s}}(V(\theta))$ can also be represented as a degree $(4mN, 4mN)$ rational function in $\theta$, the same degree as the Boson Sampling case in Lemma~\ref{rationalfunction}. 
		
From Eq.~\eqref{outputprobabilitygaussian}, the output probability $q_{\bm{s}}(V(\theta))$ has the form of
		\begin{align}
			\begin{split}
				q_{\bm{s}}(V(\theta)) &= \tanh^{N}r\sech^{M}r \left|\sum_{\mu \in \text{PMP}}\prod_{j = 1}^{N/2} \left[(V(\theta)V(\theta)^{T})_{\bm{s}}\right]_{\mu(2j-1), \mu(2j)} \right|^2,
			\end{split}
		\end{align}
		where $\mu$ is along all possible perfect matching permutations over $N$ modes. 
		From the proof of Lemma~\ref{rationalfunction}, using reduction to the common denominator for all of the $m = qM\log M$ gates, $\left[V(\theta)\right]_{j,k}$ can be represented as $(2m,2m)$ rational function in $\theta$ with the common denominator $\prod_{i=1}^{m} q_i(\theta)$.
		Using this fact, one can easily check that $\prod_{j = 1}^{N/2} \left[(V(\theta)V(\theta)^{T})_{\bm{s}}\right]_{\mu(2j-1), \mu(2j)}$ can be represented as $(2mN,2mN)$ rational function in $\theta$, with the common denominator $[\prod_{i=1}^{m} q_i(\theta)]^N$ which does not change with $\mu$. 
		Therefore, the output probability can be represented as $q_{\bm{s}}(V(\theta)) = \frac{P(\theta)}{Q(\theta)}$, with each $Q(\theta) = [\prod_{i=1}^{m} |q_i(\theta)|^2]^N$ and $P(\theta)$ a degree $4mN$ polynomial function in $\theta$. 
		
		Given that $q_{\bm{s}}(V(\theta))$ can be represented as a degree $(4mN, 4mN)$ rational function with the same denominator $Q(\theta) = [\prod_{i=1}^{m} |q_i(\theta)|^2]^N$ from the Boson Sampling case in Lemma~\ref{rationalfunction}, we can repeat all the steps identically to the proof of Theorem~\ref{averagehardness} and obtain the same result.
	\end{proof}

\section{Hardness analyses in lossy environments}\label{Section:lossy}

In this section, we generalize our hardness results for $\it{lossy}$ environments, namely, shallow-depth linear optical circuits suffering from photon loss channels after each gate implementation. 
The reason we consider such a noise channel is that photon loss is indeed a major source of error in optical systems~\cite{zhong2020quantum, zhong2021phase, madsen2022quantum, deng2023gaussian}. 
Also, photon loss ruins the classical intractability of Boson Sampling, as there exist many efficient classical algorithms that can simulate lossy Boson Sampling within a constant total variation distance~\cite{oszmaniec2018classical, garcia2019simulating, qi2020regimes}. 
Therefore, we mainly deal with the photon loss error here; 
our goal is to provide evidence for the hardness of the approximate simulation of Boson Sampling in lossy shallow circuits within total variation distance error. 
For simplicity, we do not consider any photon gain error here, such as thermal radiation noise subjected to the circuits. 

To proceed, we start with a brief review of the results presented by Ref.~\cite{fujii2016noise}, which shows the hardness of simulating noisy quantum circuits.
Specifically, one can simulate a noiseless circuit using a larger noisy circuit up to the desired imprecision, by establishing error-detecting code in the noisy circuit and post-selecting null syndrome measurements.
Therefore, given the probability to post-select the no-error syndromes, one can approximate the output probability of the noiseless circuit from the output probability of the noisy circuit.
Based on this argument, Ref.~\cite{bouland2022noise} demonstrates the average-case hardness of approximating output probabilities of noisy quantum circuits, under some plausible assumptions of the noise model.
This result gives evidence of the approximate simulation hardness of noisy quantum circuits, within total variation distance error.

The main strategy of the above hardness results is approximating ideal output probabilities by post-selecting error-free results from noisy circuits.
Here, we can directly apply their strategy to our case, i.e., lossy shallow-depth linear optical circuits.
The crucial observation is that considering photon loss error on Boson Sampling, the error syndrome is the output photon number $\it{itself}$. 
Specifically, if the output photon number is the same as the input photon number, this implies that no loss occurred throughout the circuit. 
Therefore, by post-selecting the event that the measured output photon number is the same as the input photon number, we can infer ideal output probabilities.

For a more detailed analysis, we set the loss model as follows.
Let the photon loss model $\mathcal{N}$ be local and stochastic.
Specifically, $\mathcal{N}$ is a set of loss channels $\{\mathcal{N}_i\}_{i=1}^{l}$, such that after each unitary gate is applied, loss channel $\mathcal{N}_i$ acts on each mode participated in the unitary gate.  
Hence, the number of loss channels is $l = O(m)$ for gate number $m$ in a given circuit architecture. 
We can decompose each noise channel $\mathcal{N}_i$ as follows:
\begin{equation}\label{noise}
    \mathcal{N}_i = (1-\rho_i)\mathcal{I} + \rho_i\mathcal{E}_i ,
\end{equation}
where $\mathcal{I}$ is identity, $\mathcal{E}_i$ is an CPTP map representing photon loss, and $\rho_i$ is a loss rate for each channel satisfying $\rho_i \le \rho$ for a constant $\rho$. 
The validity of such modeling for photon loss channel is represented in~\cite{oszmaniec2018classical, deshpande2022quantum}.

To simplify, we assume that we know $\it{a\,priori}$ each error rate $\rho_i$ for all $i \in [l]$, and the noise model $\mathcal{N}$ is fixed so that it does not change with random circuit instances.
Then we can obtain the hardness of approximating output probabilities of lossy shallow circuits, from our previous hardness proposals. 
To do so, let $p_{\bm{s}}(C,\mathcal{N})$ be the output probability of an $N$-photon outcome $\bm{s}$ from an $M$-mode linear optical circuit $C$ which undergoes loss model $\mathcal{N}$ we set. 
By post-selecting `no loss event', which can be accomplished by counting the output photon number, the ideal output probability $p_{\bm{s}}(C)$ can be inferred from $p_{\bm{s}}(C,\mathcal{N})$ by 
\begin{equation}\label{noisyprobability}
    p_{\bm{s}}(C) = \frac{p_{\bm{s}}(C,\mathcal{N})}{\Pr[\text{`no loss event'}]}. 
\end{equation}
From Eq.~\eqref{noise}, the probability of `no loss event' is $\prod_{i=1}^{l}(1-\rho_i)$, which can be efficiently calculated. 
This implies that approximating $p_{\bm{s}}(C,\mathcal{N})$ can be reduced from approximating $p_{\bm{s}}(C)$, with at most $\Pr[\text{`no loss event'}]^{-1} = \prod_{i=1}^{l}(1-\rho_i)^{-1} \le (1-\rho)^{-l} = 2^{O(\rho m)}$ blowup in the additive imprecision.

Given Eq.~\eqref{noisyprobability}, we can repeat the same steps from the previous hardness arguments, for the lossy shallow-depth Boson Sampling; the only difference is the imprecision blowup by $2^{O(\rho m)}$.
For our shallow-depth architecture $(\mathcal{B}\mathcal{B}^*)^q$, the gate number $m$ is $qM\log M$, so the size of imprecision blowup is $2^{O(N^{\gamma}\log N)}$ in our case. 
Such imprecision blowup does not affect the allowed additive accuracy $\epsilon = 2^{-O(N^{\gamma+1}(\log N)^2)}$ for our average-case hardness result. 
Based on the arguments so far, the following corollary is straightforward from Theorem~\ref{averagehardness}.


\begin{corollary}\label{noisyaveragehardness}
    Suppose we have the photon loss model $\mathcal{N}$ with each loss rate $\rho_i \le \rho$ for a constant $\rho$. 
    Then the following problem is \#$\rm{P}$-hard under a $\rm{BPP}^{\rm{NP}}$ reduction: for any constant $\delta, \eta \ge 0$ with $\delta + \eta < \frac{1}{4}$, {on input a $O(\log N)$-depth random circuit $U \sim \mathcal{H}_{\mathcal{A}}\mathcal{P}$ with $\mathcal{A} = (\mathcal{B}\mathcal{B}^*)^{q\ge 3}$ and a random collision-free outcome $\bm{s} \sim \mathcal{G}_{M,N}$}, compute the lossy output probability $p_{\bm{s}}(U, \mathcal{N})$ within additive imprecision $\epsilon = 2^{-O(N^{\gamma+1}(\log N)^2)}$, with probability at least $1-\delta$ over the choice of $U$ for at least $1-\eta$ over the choice of $\bm{s}$. 
\end{corollary}

We remark that for our noise model, the imprecision blowup grows exponentially with the gate number $m$. 
Therefore, shallow-depth circuits can be more advantageous in this perspective, since they are likely to have less gate number and thus have small imprecision blowup.  
For example, the current hardness results are based on $M$ by $N$ submatrices of $M$-dimensional Haar random unitaries, and the implementation of such matrices requires gate number $m = \Omega(N^{\gamma+1})$.
This arouses the imprecision blowup at least $2^{O(N^{\gamma + 1})}$, which restricts the allowed additive error for the average-case hardness at most $2^{-O(N^{\gamma + 1})}$.

\section{Concluding remarks}\label{Section:remarks}
	
	
Here we provide a few remarks about our overall results and related open questions.

1. Our result demonstrates the average-case hardness for additive imprecision $2^{-O(N^{\gamma + 1}(\log N)^2)}$.
Indeed, there still remains a gap to the desired additive imprecision for the simulation hardness $2^{-(\gamma - 1)N\log N - O(N)}$ in Theorem~\ref{simulationhardness}. 
Hence, closing this gap would be an ultimate challenge to the full achievement of classical intractability; more advanced proof techniques are required to reduce this gap. 
Here, one can take the following approach: finite-size numerical experiments suggest that the output distributions of local random circuits in our circuit architecture (i.e., Definition~\ref{butterfly}) are close enough to those of global Haar random circuits {(see, e.g., Ref.~\cite{go2024exploring} and Appendix~\ref{appendix: section: numerical evidence of bbp})}. 
Accordingly, if one can analytically prove that the distance between those output distributions is close enough, we can directly obtain a better imprecision level $2^{-O(N\log N)}$ by results in~\cite{bouland2022noise, deshpande2022quantum, bouland2023complexity, bouland2024average}, which employed the global Haar random circuits.

Another possible approach for reducing the imprecision gap is to perturb a random circuit matrix in a different way from the Cayley transform (i.e., Definition~\ref{cayley}), i.e., as depicted in~\cite{bouland2023complexity}. 
Specifically, instead of perturbing each random gate, one can perturb a submatrix $X$ of our random circuit matrix $U$ with a worst-case matrix $A$ as $X(\theta) = (1-\theta)X + \theta A$ for $\theta\in[0,1]$. 
Here, a degree of polynomial $|\text{Per}(X(\theta))|^2$ is $2N$, which is lower than ours derived by the Cayley transform.
Therefore, if one can prove that $X(\theta)$ is distributed similarly to $X$ for small $\theta$, we expect that we can also obtain a better imprecision level by using the same interpolation method.  
However, the above approach requires one to figure out a global circuit distribution generated by the convolution of local circuit distributions. 
Although we believe that this problem can be resolved using techniques from random matrix theories, we have not yet developed a complete analysis.
Hence, we leave it as an open question.

2. Another important challenge that should be addressed is the extension of the result in Sec.~\ref{Section:lossy}, i.e., finding the classical simulation hardness of noisy Gaussian Boson Sampling with photon loss, or Boson Sampling for more general types of noise.
As described in~\cite{fujii2016noise, bouland2022noise, deshpande2022quantum}, employing the threshold theorem would be a viable choice for this goal, which requires efficient error detection codes using only linear optical elements for these setups. 
However, to the best of our knowledge, such an error detection code does not yet exist. 
That is, in contrast to the Boson Sampling case, we cannot post-select the so-called `no loss' event in the Gaussian Boson Sampling case because the input photon number is not fixed in this case. 
This is also for the Boson Sampling with general types of noise because we cannot post-select no error event by using only photon number measurements. 
Hence, constructing any of these error detection codes would be a crucial step toward the classical hardness of noisy (Gaussian) Boson Sampling, which will contribute to a more noise-tolerant demonstration of quantum advantage with (Gaussian) Boson Sampling. 
We leave this problem as another open question.

3. One of the existing gaps between theoretical Boson Sampling hardness arguments and Boson Sampling experiments to date is the ratio between mode number $M$ and photon number $N$. 
Generally, most of the current theoretical Boson Sampling hardness arguments (including our current work) consider only collision-free outcomes and rule out collision outcomes due to the convenience of mathematical analysis~\cite{aaronson2011computational, deshpande2022quantum, grier2022complexity, bouland2022noise, bouland2024average}.
This calls for the ``dilute" regime $M = \Omega(N^2)$ to ensure the dominance of collision-free outcomes over the probability weight. 
In contrast, all the experiments to date~\cite{zhong2020quantum, zhong2021phase, madsen2022quantum, deng2023gaussian} operate in the ``saturated regime” where the number of modes is comparable to the number of photons. 
Therefore, to reduce this gap, extending our hardness result of shallow-depth Boson Sampling to the saturated regime $M = O(N^{\gamma})$ for $1 \leq \gamma < 2$ is crucial, similarly to~\cite{bouland2023complexity} for the general Boson Sampling case.
To do so, one also needs to consider collision outcomes in the analysis because collision outcomes now dominate the probability weight in this regime. 
Here, given that our worst-case hardness result in Theorem~\ref{worstcase} holds for a linear number of modes $M = \Omega(N)$, this opens the possibility of establishing average-case hardness in the saturated regime by extending the worst-to-average-case reduction from Theorem~\ref{worstcase} to include collision outcomes.
However, since the collision outcomes do not interchange with each other generally by permutations (in contrast to collision-free outcomes), our proof technique in Sec.~\ref{Section:Averagecase} using random permutation does not work for general collision outcomes.
Therefore, more advanced proof techniques are required to address this challenge. 
In Appendix~\ref{appendix: section: collision}, we present an approach that could offer a useful insight toward extending our average-case hardness result to include collision outcomes.
Also, another possible way is to analytically prove that the local random circuit for our circuit architecture $(\mathcal{B}\mathcal{B}^*)^{q}$ is close enough to those of global Haar random circuits (numerically shown in~\cite{go2024exploring} and Appendix~\ref{appendix: section: numerical evidence of bbp}). 
If this can be done, we can employ the previous argument~\cite{bouland2023complexity}, which used global Haar random circuits, for our circuit architecture $(\mathcal{B}\mathcal{B}^*)^{q}$ in the saturated regime.

4. We finally remark that if we only use the local random circuit ensemble $\mathcal{H}_{\mathcal{A}}$ for our shallow-depth circuit architecture $\mathcal{A} = (\mathcal{B}\mathcal{B}^*)^{q}$ as defined in Definition~\ref{randomcircuit}, the bosonic birthday paradox is not guaranteed because of the absence of translation invariance property (in contrast to the global Haar random unitary circuit). 
Hence, for the sake of analytical completeness, we introduced the additional (shallow-depth) random permutation circuit at the input to guarantee the bosonic birthday paradox and suppress the occurrence of collision outcomes.
In fact, from an experimental perspective, this setup can be realized by constructing the local random circuit $\mathcal{H}_{\mathcal{A}}$ and preparing a random input configuration. 
Furthermore, if one can analytically verify that the local random circuit ensemble $\mathcal{H}_{\mathcal{A}}$ has a permutation invariance property for $\mathcal{A} = (\mathcal{B}\mathcal{B}^*)^{q}$, then we can safely leave out this random permutation circuit part, thereby simplifying the average-case circuits and reducing the required circuit depth for the average-case hardness.
We leave in Appendix~\ref{appendix: section: numerical evidence of bbp} numerical analysis for the bosonic birthday paradox for our local random circuit ensemble $\mathcal{H}_{\mathcal{A}}$ with $\mathcal{A} = (\mathcal{B}\mathcal{B}^*)^{q}$, which gives evidence that only the local random circuit ensemble might be sufficient for the average-case hardness result.

\section*{Acknowledgements}	
We thank Bill Fefferman for insightful discussions. 
This work was supported by the National Research Foundation of Korea (NRF) grant funded by the Korea government (MSIT) (Nos. RS-2024-00413957, RS-2024-00438415, and RS-2023-NR076733) and by the Institute of Information \& Communications Technology Planning \& Evaluation (IITP) grants funded by the Korea government (MSIT) (IITP-2026-RS-2020-II201606 and IITP-2026-RS-2024-00437191). 
B.G. was supported by the education and training program of the Quantum Information Research Support Center, funded through the National Research Foundation of Korea (NRF) by the Ministry of Science and ICT (MSIT) of the Korean government (No.2021M3H3A1036573).
C.O. was supported by the National Research Foundation of Korea Grants (No. RS-2024-00431768 and No. RS-2025-00515456) funded by the Korean government (Ministry of Science and ICT (MSIT)) and the Institute of Information \& Communications Technology Planning \& Evaluation (IITP) Grants funded by the Korea government (MSIT) (No. IITP-2025-RS-2025-02283189 and IITP-2025-RS-2025-02263264).
This work was supported by Global Partnership Program of Leading Universities in Quantum Science and Technology (RS-2025-08542968) through the National Research Foundation of Korea~(NRF) funded by the Korean government (Ministry of Science and ICT(MSIT)).


\bibliographystyle{quantum}
\bibliography{reference}

\appendix
\section{Previous foundations on average-case hardness}\label{previousfoundations averagecasehardness}
{In this appendix, we introduce the existing proof technique used to establish the classical hardness of Boson Sampling, specifically in the context of approximate simulation within a bounded total variation distance error.~\cite{aaronson2011computational}.}
The current state-of-the-art proof technique for the hardness of sampling problems like Boson Sampling essentially builds upon Stockmeyer's algorithm about approximate counting~\cite{stockmeyer1985approximation}.
Specifically, given a classical sampler that outputs a sample from a given output distribution, Stockmeyer's algorithm enables one to multiplicatively estimate a fixed output probability of the sampler, within complexity class $\rm{BPP}^{\rm{NP}}$. 

	Now suppose there exists an approximate classical sampler capable of simulating ideal Boson Sampling up to total variation distance error, as in Definition~\ref{approximatebosonsampler}.
	This approximate sampler can have a large additive error for a fixed output probability, but have a comparably small additive error for \textit{most} of the output probabilities due to Markov's inequality.
	Then, Stockmeyer's algorithm, combined with the approximate sampler, can well approximate the ideal output probability of Boson Sampling within a certain additive error, with a high probability over randomly chosen outcomes (See Lemma~\ref{aaronsonandarkhipov} for more details). 
	For convenience, let us refer to this computational task as an \textit{average-case approximation} problem of Boson Sampling. 
	If the complexity of the average-case approximation problem is outside the Polynomial Hierarchy (PH), it implies the collapse of PH, since the complexity of Stockmeyer's algorithm is indeed inside the finite level of PH. 
	
Here, \textit{average-case hardness} comes into the proof of the classical simulation hardness argument, which means that approximating the ideal output probability of Boson Sampling with high probability over randomly chosen outcomes is \#P-hard. 
More precisely, if the average-case hardness holds up to the imprecision level of the average-case approximation problem, this comes down to the classical simulation hardness of the approximate sampling unless PH collapses, by the complexity-theoretical foundation $\rm{PH} \subseteq \rm{P}^{\#P}$~\cite{toda1991pp}. 

Moreover, by choosing random circuit instances that have symmetry over the outcomes, one can reduce the average-case instances for the hardness from \textit{outcome} instances to \textit{circuit} instances, which is called the \textit{hiding} property.
For the Boson Sampling case, global Haar random unitary (i.e., unitary matrix drawn from Haar measure on U(M), for mode number M) satisfies this condition. 
In detail, instead of randomly choosing the outcome, we can fix the outcome by applying a random permutation to the global Haar random unitary distribution, which is still Haar distributed due to its translation invariance property.
This hiding property plays an important role in the current proofs of the average-case hardness, as it enables one to establish worst-to-average-case reduction in terms of circuit instances. 
Specifically, as the output probability of Boson Sampling can be written as a low-degree polynomial of input circuit (matrix) values, it allows one to infer the value of a worst-case instance from the output probability of many average-case circuit instances.
Hence, the average-case hardness argument for Boson Sampling is typically used in this context, i.e., average-case hardness over random circuit instances, for a fixed outcome~\cite{aaronson2011computational, deshpande2022quantum, bouland2022noise}.

Accordingly, the crucial step for the classical simulation hardness of approximate sampling is to prove the average-case hardness for the desired imprecision level.
While there have been many impressive results and technical developments about the average-case hardness of Boson Sampling~\cite{aaronson2011computational, deshpande2022quantum, grier2022complexity, bouland2022noise, bouland2023complexity, bouland2024average}, the average-case hardness for the desired imprecision level is not yet fully demonstrated.
Still, there exists a gap between the imprecision level of average-case hardness in the strongest existing results and the imprecision level of average-case approximation problem. 
Hence, closing this imprecision gap remains the ultimate challenge for the complete theoretical demonstration of the classical hardness of approximate Boson Sampling.

\section{The bosonic birthday paradox for the shallow average-case circuits}\label{section: appendix: bosonic birthday paradox}

Throughout our work, we consider only collision-free outcomes in our average-case hardness arguments instead of all possible outcomes.
In this appendix, we examine whether the bosonic birthday paradox~\cite{aaronson2011computational} holds for our average-case circuit ensemble $\mathcal{H}_{\mathcal{A}}\mathcal{P}$.
More specifically, we investigate whether collision-free outcomes are observed with high probability in the so-called “dilute” regime, where the number of modes satisfies $M = \Omega(N^2)$.

To claim the classical hardness and quantum computational advantage of Boson Sampling based solely on collision-free outcomes, it must be ensured that these collision-free outcomes occupy at least a not-a-very-small portion of the probability weight over all possible outcomes. 
To understand why, let us denote the probability weight corresponding to the collision-free outcomes by $\varepsilon$.
For the probability weight of collision-free outcomes $\varepsilon$, the ``adversarial" classical sampler that outputs only collision outcomes is also a good sampler, if the allowed total variation distance error for the classical sampler is larger than $\varepsilon$. 
Using this adversarial sampler combined with Stockmeyer's algorithm that multiplicative approximates the output probability of the sampler, the obtained estimates of output probabilities of collision-free outcomes are always ``0", since this sampler does not output collision-free outcomes. 
Therefore, even if we derive the average-case \#P-hardness of estimating output probabilities for collision-free outcomes, we cannot solve this \#P-hard problem by Stockmeyer's reduction if the given sampler is adversarial; thereby, we cannot derive the classical hardness result of sampling (which allows total variation distance error) from the average-case \#P-hardness.


Accordingly, to avoid this issue and claim the classical hardness result using only collision-free outcomes, the bosonic birthday paradox is crucial, which suppresses the occurrence of collision outcomes over the experiments. 
This property holds for global Haar random circuits in the dilute regime $M = \Omega(N^2)$~\cite{aaronson2011computational}, but not guaranteed for any other random circuit ensemble, including our circuit ensemble. 
In fact, collision-free outcomes dominate the outcome space in the regime $M = \Omega(N^2)$ since the portion of collision-free outcomes over all possible outcomes can be bounded as $\binom{M}{N}/\binom{M+N-1}{N} > 1 - \frac{N^2}{M}$.
Accordingly, one might expect that the bosonic birthday paradox automatically holds for any circuit ensemble.
However, even in this regime, the probability weight might be highly concentrated on collision outcomes for some circuit ensembles (i.e., anticoncentration is not guaranteed for arbitrary circuit ensembles), such that the probability weight of the collision-free outcomes is still too small.
Hence, it is crucial to show that the bosonic birthday paradox for our circuit ensemble, such that our circuit ensemble outputs collision-free outcomes with high probability.

Based on this understanding, we now argue the bosonic birthday paradox for the circuit ensemble $\mathcal{H}_{\mathcal{A}}\mathcal{P}$ in the dilute regime $M = \Omega(N^2)$.
We first review a lemma given in Ref.~\cite[Lemma C.1]{aaronson2011computational}, which is indeed a crucial ingredient for the bosonic birthday paradox.

\begin{lemma}[Unitary Pigeonhole Principle~\cite{aaronson2011computational}]\label{appendix: lemma: pigeonhole}
Partition a finite set $[K]$ into a ``good part" $G$ and ``bad part" $B = [K]\setminus G$. Also, let $\hat{U}$ be an arbitrary $K$ by $K$ unitary matrix. Suppose we choose an element $x \in G$ uniformly at random, apply $\hat{U}$ to $\ket{x}$, and then measure $\hat{U}\ket{x}$ in the standard basis. Then, for measurement outcome $y$, we have $\Pr[y\in B] \leq |B|/|G|$.
\end{lemma}

Here, let us consider the finite set $[K]$ in the lemma as a set of all possible outcomes of Boson Sampling for $N$ photons and $M$ modes, i.e., the set of possible $\bm{s} = (s_1,\dots,s_M)$ where each $s_i \in [N]$ with the constraint $\sum_{i=1}^{M}s_i = N$ (such that, $|K| = \binom{M+N-1}{N}$).
Also, we consider $\hat{U}$ in the lemma as a unitary evolution over the Hilbert space of $N$ photons and $M$ modes, characterized by a given unitary circuit of Boson Sampling. 

We now consider the good part $G$ in the lemma as a set of collision-free $\bm{s}= (s_1,\dots,s_M)$, such that each $s_i \in \{0,1\}$ instead of $s_i \in [N]$, which we will denote as $G_{M,N}$. 
Likewise, let $B_{M,N}$ be the bad part in the lemma, i.e., a set of collision $\bm{s}$. 
Then, observe that $|G_{M,N}| = \binom{M}{N}$, $|B_{M,N}| = \binom{M+N-1}{N} - \binom{M}{N}$, and thus we have
\begin{align}
    \frac{|B_{M,N}|}{|G_{M,N}|} = \frac{ \binom{M+N-1}{N} - \binom{M}{N}}{\binom{M}{N}} < \frac{\frac{N^2}{M}}{1 - \frac{N^2}{M}} \leq \frac{2N^2}{M}, 
\end{align}
where we used the fact that $\binom{M}{N}/\binom{M+N-1}{N} > 1 - \frac{N^2}{M}$ and assumed $M \geq 2N^2$ in the last inequality for simplicity.
Therefore, by Lemma~\ref{appendix: lemma: pigeonhole}, if we uniformly choose collision-free input (i.e., $\bm{t}$ in Eq.~\eqref{outputprobability}), then the probability of observing collision outcomes is bounded by $\frac{2N^2}{M}$, for any unitary circuit.

Now, instead of using uniformly chosen input and fixed unitary circuit, we can rather fix the collision-free input $\bm{t}$ and draw random unitary circuits from $\mathcal{H}_{\mathcal{A}}\mathcal{P}$ (see Definition~\ref{randomcircuit} and Definition~\ref{randompermutation}), which is the random circuit ensemble for the average-case hardness in the shallow-depth regime.
Then, the random permutation circuit part from $\mathcal{P}$ permutes the collision-free input $\bm{t}$ uniformly random, and effectively makes input uniform distribution over possible collision-free configurations $G_{M,N}$. 
Hence, by choosing circuits from $\mathcal{H}_{\mathcal{A}}\mathcal{P}$, then the probability to observe collision outcomes is bounded by $\frac{2N^2}{M}$, and thus we have the following theorem.

\begin{theorem}[Bosonic birthday paradox 
for the product circuit ensemble $\mathcal{H}_{\mathcal{A}}\mathcal{P}$]
\label{theorem: BBP}

For a randomly chosen unitary circuit $U \sim \mathcal{H}_{\mathcal{A}}\mathcal{P}$, the probability to obtain collision outcomes is upper bounded by
\begin{align}
    \E_{U \sim \mathcal{H}_{\mathcal{A}}\mathcal{P}}\left[ \Pr_{\mathcal{D}_{U}}\left[ \bm{s} \in B_{M,N} \right] \right] < \frac{2N^2}{M} ,
\end{align}
where $\mathcal{D}_{U}$ is the output distribution according to $p_{\bm{s}}(U)$ defined in Eq.~\eqref{outputprobability}. 
\end{theorem}

We have now obtained the desired result; the above theorem implies that for most of the random circuits over $ \mathcal{H}_{\mathcal{A}}\mathcal{P}$, the sampler outputs a collision-free outcome with high probability in the dilute regime $M = \Omega(N^2)$.


Crucially, the proof of Theorem~\ref{theorem: BBP} does not rely on the circuit ensemble $\mathcal{H}_{\mathcal{A}}$; rather, the essential ingredient for the bosonic birthday paradox is the initial random permutation layer $\mathcal{P}$. 
Therefore, the bosonic birthday paradox is guaranteed in the dilute regime $M = \Omega(N^2)$ as long as the as long as the input (collision-free) configuration is uniformly random by any means.
In practice, this can also be achieved experimentally by directly preparing a uniformly random input configuration.
Nevertheless, we chose to include an explicit random permutation circuit $\mathcal{P}$ for mathematical clarity, since we fix the input configuration for Boson Sampling.


We additionally note that this bosonic birthday paradox is not guaranteed solely by the local random circuit ensemble in Definition~\ref{randomcircuit} itself, because of the absence of translation invariance property (in contrast to the global Haar random unitary circuit). 
Hence, the additional random permutation circuit at the input is crucial to guarantee the bosonic birthday paradox mathematically.
Nevertheless, if one can prove that the local random circuit ensemble has a permutation-invariance property for a specific circuit architecture, then we can leave out this random permutation circuit part at the input.
We leave in Appendix~\ref{appendix: section: numerical evidence of bbp} numerical evidence of the bosonic birthday paradox for solely the local random circuit ensemble $\mathcal{H}_{\mathcal{A}}$ with our shallow-depth circuit architecture $\mathcal{A} = (\mathcal{B}\mathcal{B}^*)^{q}$.

\section{Proof of Theorem~\ref{worstcase} (worst-case hardness)}\label{proof of worstcase hardness}

{
To prove Theorem~\ref{worstcase}, we first introduce a crucial lemma, which states that any BQP circuit can be simulated by a constant-depth linear optical circuit under post-selection.
}

\begin{lemma}[Brod~\cite{brod2015complexity}, revised]\label{Brod}
For an arbitrary given poly-sized $n$-qubit quantum circuit $Q$, there exists a constant depth linear optical circuit $C_0$ such that for $M = \Omega(N)$ and $N = {\rm{poly}}(n)$,
\begin{equation}\label{quantumcircuittolinearopticalcircuit}
    |\bra{J}\hat{\mathcal{U}}(C_0)\ket{I}|^2 = c_{Q}|\bra{0}^{\otimes n} Q \ket{0}^{\otimes n}|^2,
\end{equation}
where {$c_{Q} = 2^{-O(N)}$} is a $Q$-dependent constant which can be efficiently computed, $\hat{\mathcal{U}}(C_0)$ is a unitary operator corresponding to the circuit $C_0$, and each of $\ket{I}$ and $\ket{J}$ is an $M$-mode Fock-state composed of $N$ single photon states and vacuum states for the rest modes.  
\end{lemma}

{
\begin{proof}

In this proof, we revise the main result in~\cite{brod2015complexity}, incorporating a more quantitative analysis on the coefficient $c_{Q}$, which later determines the imprecision level allowed for the worst-case hardness of shallow-depth Boson Sampling.
To do so, we first go through the main result in~\cite{brod2015complexity}.

Given an arbitrary poly-sized quantum circuit $Q$ on $n$ qubits, there exists a measurement-based quantum computation (MBQC) scheme using constant depth brickwork graph state composed of maximally $m_Q = \text{poly}(n)$ number of qubits and $\Theta(m_Q)$ number of CZ gates, with an explicit measurement basis sequence for the universal gate set \{CX, T, H\} to simulate $Q$~\cite{raussendorf2001one, broadbent2009universal, childs2005unified} (see, e.g., Fig.~\ref{fig:mbqc_to_grid} (a)).
Since the measurement outcomes for gate implementation are random, one needs adaptive measurements based on these measurement outcomes to implement $Q$ (which is indeed the essence of MBQC).
However, if post-selecting desired measurement outcomes are allowed, one can rather post-select ``good" outcomes that do not require any adaptive measurements to implement $Q$.

This brickwork qubit circuit can also be implemented in linear optical systems via the celebrated KLM scheme~\cite{knill2001scheme}.
In their scheme, they use a dual-rail encoded state ($\ket{01}$ or $\ket{10}$ in the Fock basis) as a qubit, a beam splitter as a one qubit gate, and a 2-depth network of beam splitters for a probabilistic CZ gate which also requires additional two single photon ancillas for post-selecting successful CZ operation. 
Using this KLM scheme, Ref.~\cite{brod2015complexity} proposed a 4-depth linear optical circuit $C_0$, which implements the brickwork qubit state by post-selecting successful CZ operation (see Fig.~\ref{fig:mbqc_to_grid} (b) or Fig.~3 in Ref.~\cite{brod2015complexity} for more details).
Here, the required photon number is $N = \Theta(m_{Q})$ to implement the brickwork qubit state composed of $m_{Q}$ number of qubits and $\Theta(m_{Q})$ number of CZ gates, because each of qubit and CZ gate requires a constant number of photons in the KLM scheme. 
For more clarity, we illustrate in Fig.~\ref{fig:mbqc_to_grid} the brickwork qubit circuit and corresponding linear optical circuit $C_0$ proposed by~\cite{brod2015complexity}. 
Also, the mode number is given by $M = \Omega(N)$, which is required for the dual-rail encoding in the KLM scheme.
We remark that the mode number $M$ of $C_0$ can be arbitrarily enlarged by considering a circuit that ``contains" the circuit in~\cite{brod2015complexity}, which uses some of the modes for dual-rail encoding, leaving the rest of the modes as vacuum and adding only trivial (i.e., identity) gates at these unwanted modes. 
As the trivial gates leave the unwanted modes separated, this kind of mode number augmentation does not affect the post-selection probabilities.

To sum up, given an arbitrary quantum circuit $Q$, there exists a 4-depth linear optical circuit $C_0$ that can simulate the quantum circuit $Q$, by post-selecting successful CZ operation for all probabilistic CZ gates and post-selecting good measurement outcomes (that do not require any adaptive measurements) to implement each gate in the universal gate set \{CX, T, H\} composing the circuit $Q$. 
Here, as the explicit measurement basis sequence of brickwork qubit circuit for the universal gate set \{CX, T, H\} is given, one can find (i) proper coefficients of beam splitters in $C_0$ to implement the universal gate set via the KLM scheme and (ii) good measurement outcomes in the Fock basis. 
Hence, for the given quantum circuit $Q$, we can efficiently construct the constant-depth circuit $C_0$ (by means of determining its coefficients) and $M$-mode $N$-single photon Fock-states $\ket{I}$ and $\ket{J}$ such that Eq.~\eqref{quantumcircuittolinearopticalcircuit} holds. 
Now, $c_{Q}$ becomes a post-selection probability to post-select (i) successful CZ operations for all CZ gates and (ii) good measurement outcomes to implement each gate in the universal gate set composing $Q$ via MBQC.
In other words, $c_Q$ can be expressed as 
\begin{align}
    c_Q &= \prod_{k\in\{\text{CZ, CX, T, H}\}}p_{k}^{\Gamma_k}, 
\end{align}
where $p_k$ denotes post-selection probability to implement $k$ gate (e.g., $p_{\text{CZ}}$ is $2/27$ in~\cite{knill2002quantum}), and $\Gamma_k$ denotes the number of $k$ gate to implement the circuit $Q$. 
Hence, given the post-selection probability $p_k$ for each gate, by counting the number of each gate to implement the circuit $Q$, $c_Q$ can be computed efficiently. 
To obtain the size of $c_{Q}$, note that we set the photon number $N = \Theta(m_{Q})$ to implement $m_{Q}$ qubit brickwork qubit circuit.
As this $m_{Q}$ qubit brickwork qubit circuit can implement at most $O(m_{Q})$ number of gates, $Q$ has at most $O(m_{Q}) = O(N)$ number of gates in the universal gate set \{CX, T, H\}. 
Also, the number of CZ gates is given as $O(N)$, and thus the total number of gates to implement $Q$ is at most $\sum_k\Gamma_{k}= O(N)$.
Therefore, $c_Q$ has its amplitude $c_Q = 2^{-O(N)}$, since all the post-selection probabilities $p_k$ are given as constants.
This concludes the proof.

\end{proof}

}

{
Next, we use the proposition by~\cite{kondo2022quantum} which states that estimating the output probability of any BQP circuit is \#P-hard for a certain additive imprecision.}

\begin{lemma}[Kondo et al~\cite{kondo2022quantum}]\label{kondo}
It is \#$\rm{P}$-hard to compute $|\bra{0}^{\otimes n} Q \ket{0}^{\otimes n}|^2$ for an arbitrary given quantum circuit $Q$ within the additive error less than $2^{-2n}$. 
\end{lemma}

{
By the above two lemmas, one can deduce that there exists a constant-depth linear optical circuit that has the worst-case hardness for a certain level of imprecision. 
However, this does not directly imply the worst-case hardness for the fixed circuit architecture $\mathcal{B}\mathcal{B}^*$ as stated in Theorem~\ref{worstcase} because it is not guaranteed that the constant-depth worst-case circuit satisfying Eq.~\eqref{quantumcircuittolinearopticalcircuit} can be implemented in the fixed circuit architecture $\mathcal{B}\mathcal{B}^*$.

Hence, we verify that there exists a constant depth worst-case circuit $C_0$ in Lemma~\ref{Brod} that can be implemented in the fixed circuit architecture $\mathcal{B}\mathcal{B}^*$, by properly allocating the mode indices of $C_0$ in $\mathcal{B}\mathcal{B}^*$.

\begin{lemma}\label{lemma: embeding}
There exists a constant-depth circuit $C_0$ in Lemma~\ref{Brod} that can be implemented in $\mathcal{B}\mathcal{B}^*$ under a proper mode index permutation.
In other words, for an arbitrary given poly-sized $n$-qubit quantum circuit $Q$, there exists a circuit $C_0$ in $\mathcal{B}\mathcal{B}^*$ that satisfies Eq.~\eqref{quantumcircuittolinearopticalcircuit} for $M = \Omega(N)$ and $N = {\rm{poly}}(n)$, where $c_{Q} = 2^{-O(N)}$, and $\ket{I}$ and $\ket{J}$ are given by $M$-mode $N$-single photon Fock-states. 
\end{lemma}

\begin{proof}
See Appendix~\ref{appendix: section: encoding}.
\end{proof}

We prove Lemma~\ref{lemma: embeding} by showing that the constant-depth linear optical circuit $C_0$ in Lemma~\ref{Brod} can be embedded into a $3$-dimensional grid-structured circuit, which can be implemented in the $\mathcal{B}\mathcal{B}^*$ architecture under a suitable permutation of mode indices.
For interested readers, we provide the full proof of Lemma~\ref{lemma: embeding} in Appendix~\ref{appendix: section: encoding}.
}

Combining the above results, the proof of Theorem~\ref{worstcase} is now straightforward.

\begin{proof}[Proof of Theorem~\ref{worstcase}]
Given an arbitrary quantum circuit $Q$, by Lemma~\ref{Brod}, we can efficiently construct a constant-depth linear optical circuit $C_0$ and Fock-states $\ket{I}$ and $\ket{J}$ such that Eq.~\eqref{quantumcircuittolinearopticalcircuit} holds. 
Also, as depicted in Lemma~\ref{lemma: embeding}, we can efficiently embed this linear optical circuit $C_0$ in $\mathcal{B}\mathcal{B}^*$ by properly allocating mode indices.
This mode index allocation permutes $\ket{I}$ and $\ket{J}$; let us denote the permuted version of $\ket{I}$ and $\ket{J}$ as $\ket{I'}$ and $\ket{J'}$, respectively. 
Then, we can construct the linear optical circuit $C_0$ in $\mathcal{B}\mathcal{B}^*$ such that 
\begin{align}\label{eq: worst-case}
    c_{Q}^{-1}p_{\bm{s}_0}(C_0) = |\bra{0}^{\otimes n} Q \ket{0}^{\otimes n}|^2,
\end{align}
for collision-free input $\bm{t}$ and output $\bm{s}_0$ corresponding to $\ket{J'}$ and $\ket{I'}$, respectively. 
Note that, the input and output configuration $\ket{J'}$ and $\ket{I'}$ are fixed and invariant under the quantum circuit $Q$, because the architecture of $C_0$ does not depend on $Q$ (more precisely, only the coefficients of $C_0$ depend on $Q$). 
As $c_Q = 2^{-O(N)}$, if we can estimate $p_{\bm{s}_0}(C_0)$ to within additive imprecision $2^{-(2n+O(N))} = 2^{-O(N)}$, we can estimate the right-hand side of Eq.~\eqref{eq: worst-case} to within additive imprecision $2^{-2n}$, which is \#P-hard as given in Lemma~\ref{kondo}. 
Therefore, for collision-free input $\bm{t}$ and output $\bm{s}_0$, estimating an output probability $p_{\bm{s}_0}(C_0)$ to within additive imprecision $2^{-O(N)}$ for any $C_0$ in $\mathcal{B}\mathcal{B}^*$ is \#P-hard in the worst case. 
    
\end{proof}


\begin{remark}
In principle, encoding the constant-depth MBQC circuit from Lemma~\ref{Brod} to the logarithmic-depth $\mathcal{BB}^*$ circuit, as done in Lemma~\ref{lemma: embeding}, is not essential for deriving the hardness argument in shallow-depth regimes; 
one could directly use the constant-depth circuit from Lemma~\ref{Brod} itself to obtain worst-case hardness in such regimes.
Nevertheless, since we employ the architecture $\mathcal{BB}^*$ to implement required permutations during the worst-to-average-case reduction, we derive worst-case hardness within $\mathcal{BB}^*$ to simplify the overall circuit architecture to $(\mathcal{BB}^*)^{q}$ in Definition~\ref{def: kaleidoscope}.
\end{remark}


\section{Proof of Lemma~\ref{lemma: embeding}: Encoding constant-depth worst-case circuit in Lemma~\ref{Brod} to the circuit architecture $\mathcal{B}\mathcal{B}^*$}\label{appendix: section: encoding}

{

In the following, we prove Lemma~\ref{lemma: embeding}, showing that the constant depth worst-case linear optical circuit $C_0$ argued in~\cite{brod2015complexity} and Lemma~\ref{Brod} can be encoded in the Kaleidoscope circuit architecture $\mathcal{B}\mathcal{B}^*$ by properly allocating mode indices.
To summarize, this can be done in two steps:
First, we show how we can map the worst-case linear optical circuit $C_0$ into a 3-dimensional ``grid" linear optical circuit (i.e., a circuit with only nearest neighbor interactions for each dimension, with open boundaries). 
Next, we show that a single parallel application of any grid-structured linear optical circuit can be implemented in the unit butterfly circuit architecture $\mathcal{B}$ under a mode index permutation. 
We then combine these arguments, showing that the worst-case circuit $C_0$ can be encoded in $\mathcal{B}\mathcal{B}^*$.

\subsection{Mapping worst-case circuit in Lemma~\ref{Brod} to a grid linear optical circuit}

We begin by describing how to map the worst-case linear optical circuit $C_0$ in Lemma~\ref{Brod} into a 3-dimensional grid-structured linear optical circuit.
To do so, we first depict in Fig.~\ref{fig:mbqc_to_grid} (a) the brickwork qubit circuit for the MBQC~\cite{raussendorf2001one, broadbent2009universal, childs2005unified}, which can simulate any BQP circuit using post-selection.
Here, each vertex (i.e., black dot) represents input $\ket{+}$ state, and each edge represents the CZ gate applied on two qubits. 
After these CZ operations, each vertex is measured with rotated $\ket{+}$ basis, where the rotation angle is explicitly given to implement the universal gate set via MBQC.

\begin{figure*}[t]
\includegraphics[width=0.85\linewidth]{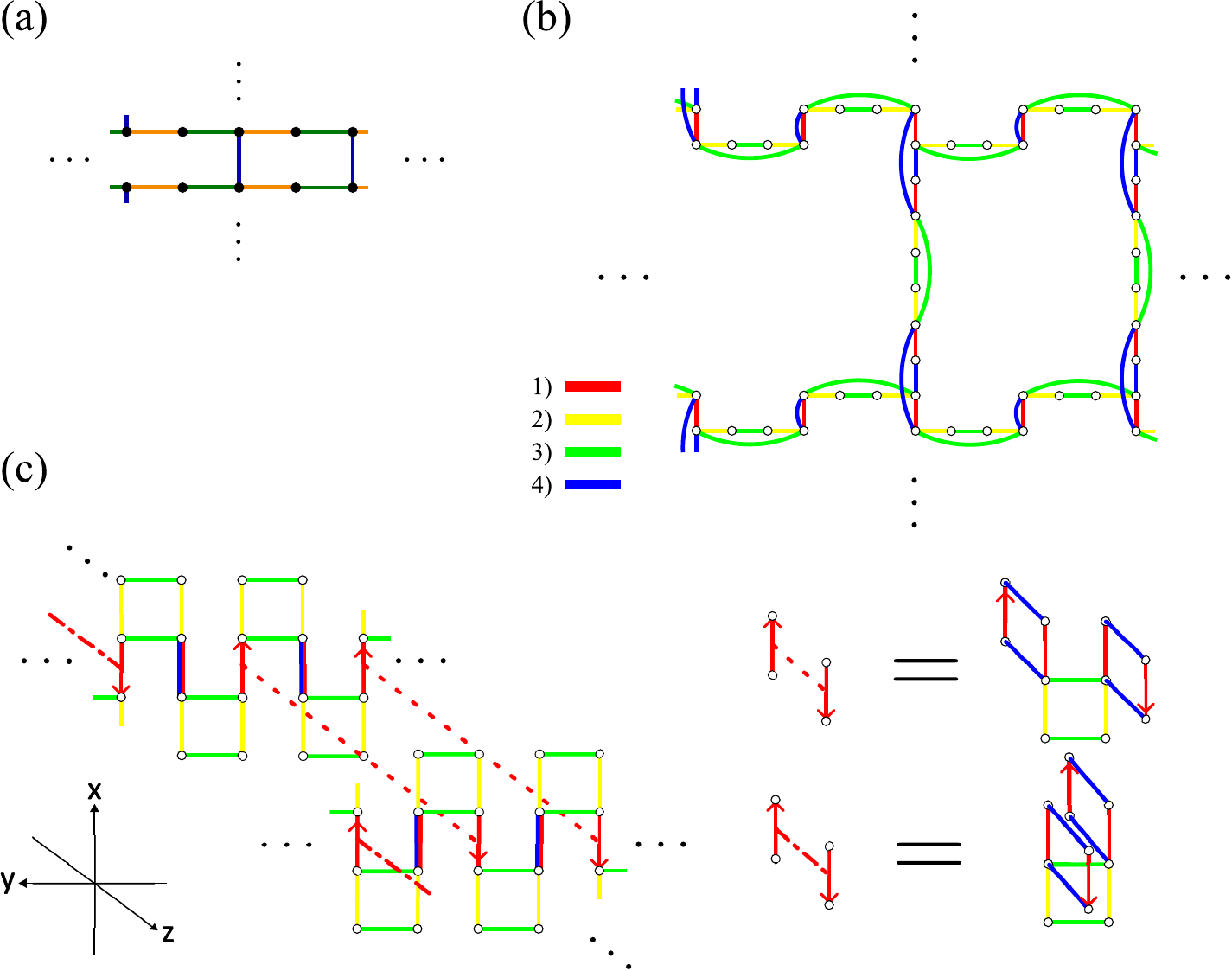}
\caption{(a) Schematics of the unit brickwork qubit circuit for measurement-based quantum computation~\cite{raussendorf2001one, broadbent2009universal, childs2005unified}. 
Each vertex (i.e., black dot) represents a $\ket{+}$ qubit, and each edge represents the CZ gate applied on two qubits. 
After CZ gate operations, each vertex is measured with a rotated $\ket{+}$ basis, where the rotation angle is explicitly given for measurement-based quantum computation. 
(b) Schematics of the 4-depth linear optical circuit proposed by~\cite{brod2015complexity} corresponding to the unit brickwork qubit circuit in (a) implemented by the KLM scheme (for more details, see Ref.~\cite{brod2015complexity}). 
Each vertex (i.e., white dot) represents zero or one photon Fock-state, and each edge represents the beam splitter with arbitrarily chosen coefficients. 
The initial beam splitter layer (red) corresponds to the input qubit (and ancillas), the following two beam splitter layers (yellow and green) correspond to the CZ gate implementation~\cite{knill2002quantum}, and the last beam splitter layer (blue) corresponds to the qubit rotation required for the measurement-based quantum computation. 
After the beam splitter operations, each vertex is measured with the Fock-state basis.
(c) Schematics of the unit brickwork linear optical circuit in (b) but mapped in 3-dimensional grid linear optical circuit architecture.
Here, the beam splitters are applied in parallel along each dimension, where the application sequence is along the direction $x \rightarrow y \rightarrow z \rightarrow x$.
}
\label{fig:mbqc_to_grid}
\end{figure*}

We can also construct this brickwork qubit circuit in linear optical systems under post-selections. 
As proposed in~\cite{brod2015complexity}, the brickwork qubit circuit in Fig.~\ref{fig:mbqc_to_grid} (a) can also be implemented in the linear optical circuit by the KLM scheme combined with post-selection, which is indeed the worst-case 4-depth linear optical circuit $C_0$ in Lemma~\ref{Brod}.
We illustrate the brickwork of this 4-depth linear optical circuit in Fig.~\ref{fig:mbqc_to_grid} (b).
Here, each vertex (i.e., white dot) represents a Fock-state with zero or one photon occupied, and each edge represents the beam splitter with appropriately chosen coefficients. 
More specifically, the initial beam splitter layer (red) corresponds to the input qubit preparation (and ancillas), the following two beam splitter layers (yellow and green) combined with single photon ancillas correspond to the CZ gate operation~\cite{knill2002quantum}.
The last beam splitter layer (blue) combined with Fock-state measurement corresponds to the rotated $\ket{+}$ measurement in the KLM scheme, which is required to implement the universal gate set via MBQC. 
}

{

By appropriately allocating the modes, the worst-case linear optical circuit $C_0$ depicted in Fig.~\ref{fig:mbqc_to_grid} (b) can be mapped into a 3-dimensional grid linear optical circuit architecture, as depicted in Fig.~\ref{fig:mbqc_to_grid} (c).
More specifically, the sequence of beam splitter application in the 3-dimensional architecture to implement the brickwork linear optical circuit is as follows.
At first, the red- and the yellow-colored beam splitters are applied in parallel along the $x$-axis. 
Next, the greed-colored beam splitters are applied in parallel along the $y$-axis. 
Lastly, the blue-colored beam splitters are applied in parallel along the $z$-axis and then applied in parallel along the $x$-axis again. 
Let us consider a single parallel application of a grid linear optical circuit as a sequence of parallel applications of beam splitters along each dimension. 
Then, one can deduce that the brickwork linear optical circuit in Fig.~\ref{fig:mbqc_to_grid} (b) can be implemented in a single parallel application of 3-dimensional grid linear optical circuit from $x$- to $z$-axis, with additional parallel beam splitter application along $x$-axis.
Here, we note that the size of this grid linear optical circuit can be arbitrarily enlarged as long as it contains the desired circuit in Fig.~\ref{fig:mbqc_to_grid} (c), where we can leave the unused modes as vacuums and add only trivial gates at these unused modes. 
Also, as the trivial gates leave the unused modes separated, this mode number augmentation does not affect the post-selection probabilities of MBQC and KLM scheme, thus leaving the total post-selection probability $c_{Q}$ given in Lemma~\ref{Brod} unchanged. 

}

{

\subsection{Encoding a grid linear optical circuit to the butterfly circuit architecture $\mathcal{B}$}

Second, we show that a single parallel application of a general-dimensional grid-structured linear optical circuit can be implemented in the unit butterfly circuit architecture $\mathcal{B}$, under a proper mode index permutation.
Here, as we mentioned previously, a single parallel application of a grid linear optical circuit is a sequence of parallel applications of beam splitters along each dimension.
More generally, we prove the following theorem.

\begin{theorem}\label{theorem: butterfly to grid}
Suppose that a mode number is given by $M = 2^{n}$ for $n \in \mathbb{Z}^{+}$.
For an arbitrary $M$-mode single parallel $d$-dimensional grid linear optical circuit $C$ with its size $2^{n_1}\times2^{n_2}\times\cdots\times2^{n_{d}}$ for $\sum_{i=1}^{d}n_{i} = n$, 
there exist an $M$-mode circuit $C'$ in $\mathcal{B}$ and an $M$-mode permutation matrix $\bm{P}$ that satisfy $C = \bm{P}C'\bm{P}^{T}$
\end{theorem}

To prove Theorem~\ref{theorem: butterfly to grid} we begin by showing how to implement a single parallel 1-dimensional grid circuit (i.e., a 1-dimensional chain, with open boundaries) by $\mathcal{B}$, whose mode number is given by $M = 2^n$. 
Using this result, we later prove Theorem~\ref{theorem: butterfly to grid} by extending this result to the general-dimensional case. 

}

{

\begin{lemma}\label{lemma: butterfly to 1d grid}
Suppose that a mode number is given by $M = 2^{n}$ for $n \in \mathbb{Z}^{+}$.
For an $M$-mode single parallel $1$-dimensional grid linear optical circuit $C$,  
there exist an $M$-mode circuit $C'$ in $\mathcal{B}$ and an $M$-mode permutation matrix $\bm{P}$ that satisfy $C = \bm{P}C'\bm{P}^{T}$
\end{lemma}

\begin{proof}

At first, we constructively establish the permutation matrix $\bm{P}$ to implement the grid circuit from the circuit in $\mathcal{B}$. 
Since $\mathcal{B}$ is composed of non-local interactions except for the first depth, we need a proper mode permutation to effectively make the grid structure. 
In short, this can be done as follows: 
As we previously defined in Definition~\ref{butterfly}, for the circuit architecture $\mathcal{B}$, we set the depth index $L \in [n]$ and the corresponding partition index $j \in [2^{n-L}]$ given by integers.
Then, in the ascending sequence of $j$ and $L$, permute the mode indices $\left\llbracket2^{L}(j-1) + 2^{L-1} + 1, 2^{L}(j-1) + 2^{L}\right\rrbracket$ from the ascending to the descending order, where the square bracket $\llbracket a, b\rrbracket$ indicates the integer interval between $a$ and $b$ included.
In other words, for each $L \in [n]$ and $j \in [2^{n-L}]$, we permute
\begin{align}\label{eq: permutation for 1 dimension}
    2^L(j-1) + 2^{L-1} + k\; \rightarrow \; 2^L(j-1) +2^{L}+1 -k,
\end{align}
for all $k \in \left[2^{L-1}\right]$ (nothing changes in Eq.~\eqref{eq: permutation for 1 dimension} for $L = 1$ so permutation does not occur for $L = 1$ case).
Let us denote $\bm{P}_{L,j}$ as the permutation described above, acting non-trivially on the mode indices $\left\llbracket2^{L}(j-1) + 2^{L-1} + 1, 2^{L}(j-1) + 2^{L}\right\rrbracket$, and acting trivially on the rest modes. 
Then, we can construct $\bm{P}$ as $\bm{P} = \prod_{L=2}^{n}\prod_{j = 1}^{2^{n-L}}\bm{P}_{L,j}$.

Let us explain why this permutation $\bm{P}$ can make the grid circuit from the circuit in $\mathcal{B}$.
First, suppose that the following is true: For given $L \geq 2$ and $j$, each of the sub-circuit of depth $L-1$ with mode indices $\left\llbracket2^{L-1}(2j-2) + 1, 2^{L-1}(2j-2) + 2^{L-1}\right\rrbracket$ and $\left\llbracket2^{L-1}(2j-1) + 1, 2^{L-1}(2j-1) + 2^{L-1}\right\rrbracket$ forms a grid circuit with mode number $2^{L-1}$; note that, this is trivial for $L = 2$ case. 
By the structure of $\mathcal{B}$, for given $L\geq 2$ and $j$, initially $(2^L(j-1)+ 2^{L-1})$-th mode and $(2^L(j-1)+ 2^{L})$-th mode are connected by a gate (i.e., a beam splitter) at depth $L$. 
Then, after the ascending-to-descending permutation $\bm{P}_{L,j}$ for given $L\geq 2$ and $j$, now $(2^L(j-1)+ 2^{L-1})$-th mode and $(2^L(j-1)+ 2^{L-1}+1)$-th mode are connected by the gate. 
Hence, by setting this gate as a non-trivial gate composing the grid circuit and setting the rest of the gates as trivial gates, we now have a grid circuit with size $2^{L}$ for each $L\geq 2$ and $j$, with mode indices given by $\left\llbracket2^{L}(j-1) + 1, 2^{L}(j-1) + 2^{L}\right\rrbracket$. 
Moreover, because we permute equivalently for all $j \in \left[ 2^{n-L} \right]$ for given $L$, these permutations do not affect the structure of $\mathcal{B}$ for larger depth $L+1, L+2,\dots,n$, such that after the permutations, $(2^{L+1}(j-1)+ 2^{L})$-th mode and $(2^{L+1}(j-1)+ 2^{L+1})$-th mode are still connected by a gate at depth $L+1$. 
Therefore, we can inductively repeat this permutation for ascending $j$ and $L$, thereby obtaining the permutation $\bm{P} = \prod_{L=2}^{n}\prod_{j = 1}^{2^{n-L}}\bm{P}_{L,j}$ that can implement the grid circuit of size $2^n$.

Next, we construct the circuit $C'$ in $\mathcal{B}$.
This can be done by choosing the non-trivial gate appropriately in $\mathcal{B}$ and leaving the rest of the gates trivial gates for each depth $L$ and partition $j$.
Firstly, set all the gates in $L = 1$ in $\mathcal{B}$ non-trivial gates because they already form a grid structure.  
Also, as we previously depicted, for each $L\geq 2$ and $j$, we set the gate between $(2^L(j-1)+ 2^{L-1})$-th mode and $(2^L(j-1)+ 2^{L})$-th mode as the non-trivial gate, where these modes are permuted from $(2^L(j-1)+ 2^{L-2} + 1)$-th mode and $(2^L(j-1)+ 2^{L-1} + 2^{L-2} + 1)$-th mode respectively in the previous step. 
Hence, given the circuit architecture $\mathcal{B}$, for each $L \geq 2$ and $j$, we set the gate between $(2^L(j-1)+ 2^{L-2} + 1)$-th mode and $(2^L(j-1)+ 2^{L-1} + 2^{L-2} + 1)$-th as the only non-trivial gate, and leave the rest of the gates as trivial gates. 
For a given single parallel 1-dimensional grid linear optical circuit $C$, let us denote $\{G_{s,t}\}$ as a set of two-mode gates composing the circuit $C$, where each gate $G_{s,t}$ is applied on $s$-th and $t$-th modes, respectively (for 1-dimensional case, $t$ is given by $s + 1$). 
Because $(2^L(j-1)+ 2^{L-1})$-th mode and $(2^L(j-1)+ 2^{L-1}+1)$-th mode are connected by the non-trivial gate after the permutation for each $L$ and $j$, we set this non-trivial gate as $G_{s,t}$ for $s = 2^L(j-1)+ 2^{L-1}$ and $t = s + 1$, thereby implementing the grid circuit $C$ by repeating this process in the ascending order of $j$ and $L$.

Therefore, we can construct the circuit $C'$ in $\mathcal{B}$ as follows. 
First, for depth $L = 1$ and $j \in [2^{n-1}]$, we set the gate between the modes $2^{L}(j-1) + 2^{L-1}$ and $2^{L}(j-1) + 2^{L}$ in $\mathcal{B}$ as $G_{s,t}$ for $s = 2^{L}(j-1) + 2^{L-1}$ and $t = s + 1$.
Next, for each depth $L \in \llbracket 2,n \rrbracket$ and $j \in [2^{n-L}]$, we set the gate between the modes $2^L(j-1)+ 2^{L-2} + 1$ and $2^L(j-1)+ 2^{L-1} + 2^{L-2} + 1$ in $\mathcal{B}$ as $G_{s,t}$ for $s = 2^{L}(j-1) + 2^{L-1}$ and $t = s+1$, and leave trivial for the rest of the gates.
Combining the arguments so far, this circuit $C'$ permuted by the proper permutation matrix $\bm{P}$ described above finally gives the desired grid circuit $C = \bm{P}C'\bm{P}^{T}$ (see, e.g., Fig.~\ref{fig:Rearranging index} (a)).

\end{proof}

\begin{figure*}[t]
\includegraphics[width=0.95\linewidth]{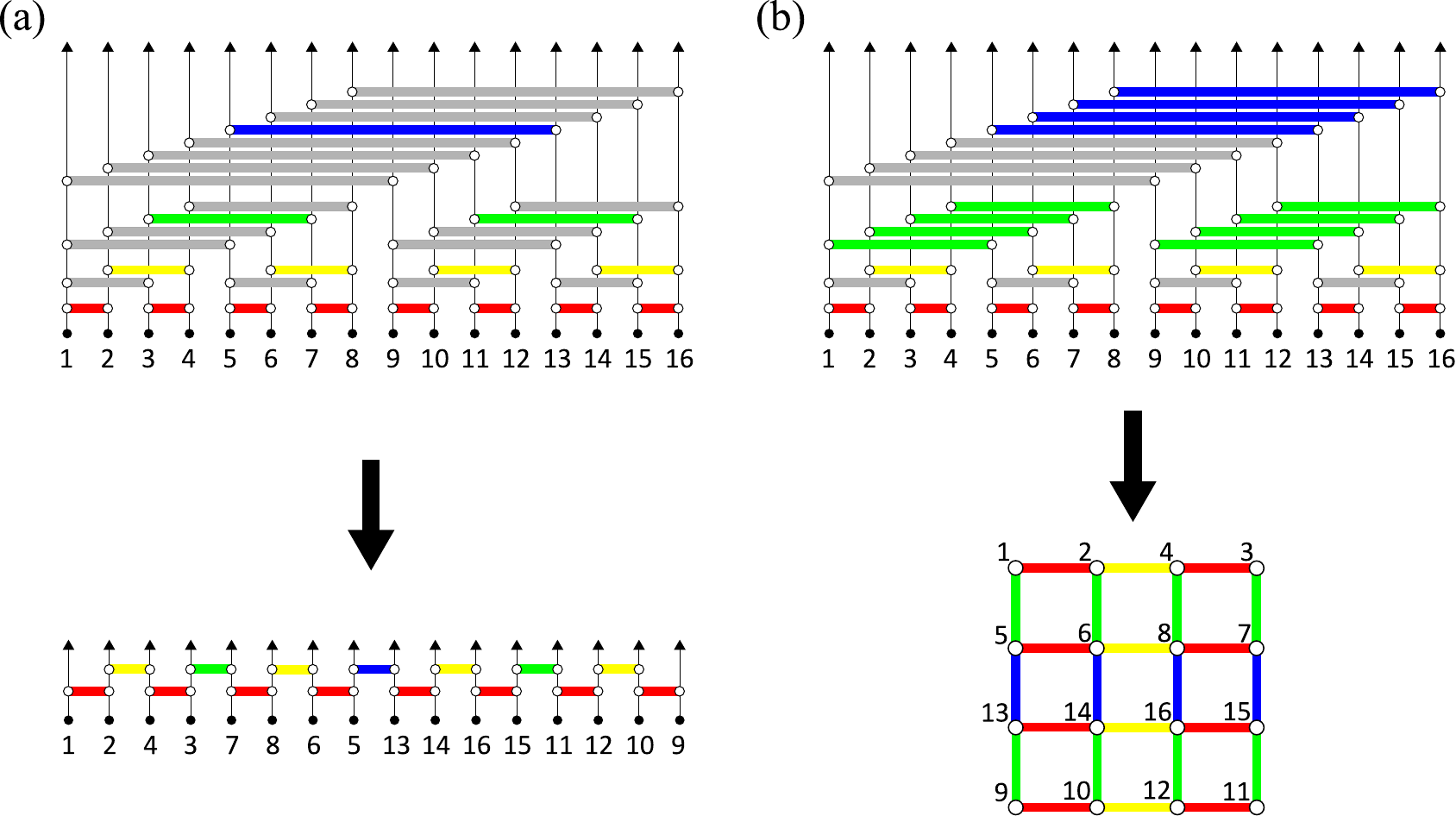}
\caption{Schematics of  $\mathcal{B}$ constructing (a) 1-dimensional grid linear optical circuit and (b) 2-dimensional grid linear optical circuit under a proper mode index permutation, for mode number $M = 16$. The grey-colored edges indicate trivial (i.e., identity) gates, and the other colored edges indicate non-trivial gates that compose the grid linear optical circuit we aim to construct.}
\label{fig:Rearranging index}
\end{figure*}

}

Based on this understanding, we now prove Theorem~\ref{theorem: butterfly to grid}, by extending the analysis in Lemma~\ref{lemma: butterfly to 1d grid} to the general-dimensional cases. 

\begin{proof}[Proof of Theorem~\ref{theorem: butterfly to grid}]

We first emphasize that, by decomposing the depth $n$ of $\mathcal{B}$ with $n_i$ such that $\sum_{i=1}^{d}n_{i} = n$, $\mathcal{B}$ can also be represented as $d$-dimensional circuit with its size $2^{n_i}$ along the $i$-th dimension; 
one can consider the gates of $\mathcal{B}$ between the depth $L \in \left\llbracket \sum_{p=1}^{i-1}n_p+1, \sum_{p=1}^{i}n_p \right\rrbracket$ are aligned along $i$-th dimension for $i \geq 2$, where the square bracket  $\llbracket a, b\rrbracket$ indicates the integer interval between $a$ and $b$ included. 
Hence, we naturally extend the analysis in Lemma~\ref{lemma: butterfly to 1d grid}, implementing a grid circuit along each dimension with size $2^{n_i}$, for all possible $n_i \in \{n_1, n_2, \dots, n_d\}$.
Here, we assume that $n_i \geq 2$ for all $i \in [d]$ without loss of generality because $\mathcal{B}$ already forms a grid along $i$-th dimension when $n_i = 1$.

More specifically, we set the size index for each dimension as $n_{i} \in \{n_1, n_2, \dots, n_d\}$, along with the depth index $L \in [n]$ and the corresponding partition index $j \in \left[2^{n - L}\right]$ as defined previously.
Also, let us denote $\bm{P}^{(i)}$ as the permutation along the $i$-th dimension, and $\{G_{s,t}^{(i)}\}$ as the set of gates applied on $s$-th and $t$-th modes along the $i$-th dimension composing the circuit $C$.
Then, along the first dimension up to $L = n_{1}$, we can similarly construct the permutation $\bm{P}^{(1)}$ and encode the gates $\{G_{s,t}^{(1)}\}$ into $C' \in \mathcal{B}$ as depicted in Lemma~\ref{lemma: butterfly to 1d grid}.

In the ascending order of the dimension $i \geq 2$, we note that each mode is departed by $2^{\sum_{p=1}^{i-1}n_{p}}$ mode index from its nearest-neighborhood along the $i$-th dimension. 
By the structure of $\mathcal{B}$, for given $L \in \left\llbracket\sum_{p=1}^{i-1}n_p + 1, \sum_{p=1}^{i}n_p \right\rrbracket $ and $j \in \left[ 2^{n - L} \right]$, initially $(2^{L}(j-1) + 2^{L-1} -2^{\sum_{p=1}^{i-1}n_{p}} + {l})$-th mode and $(2^{L}(j-1) + 2^{L} -2^{\sum_{p=1}^{i-1}n_{p}} + {l})$-th mode are connected by gates at $L$ for all ${l} \in \left[2^{\sum_{p=1}^{i-1}n_{p}} \right]$, where ${l}$ characterizes the indices along all the previous dimensions. 
Hence, now our permutation rule is that for each $L \in \left\llbracket\sum_{p=1}^{i-1}n_p + 1, \sum_{p=1}^{i}n_p \right\rrbracket $ and $j \in \left[2^{n - L} \right]$, we permute mode indices by 
\begin{align}\label{eq: permutation for general dimension}
    2^{L}(j-1) + 2^{L-1} + (k-1)2^{\sum_{p=1}^{i-1}n_{p}} + {l} \;\rightarrow\; 2^{L}(j-1) + 2^{L} -k2^{\sum_{p=1}^{i-1}n_{p}} + {l},
\end{align}
for all $k \in \left[2^{L-1- \sum_{p=1}^{i-1}n_{p}}\right]$ and  ${l} \in \left[2^{\sum_{p=1}^{i-1}n_{p}} \right]$.
Indeed, this is the ascending-to-descending permutation similar to Eq.~\eqref{eq: permutation for 1 dimension} but now along the $i$-th dimension instead. 
Note that nothing happens in Eq.~\eqref{eq: permutation for general dimension} for $L = \sum_{p=1}^{i-1}n_{p} + 1$ so permutation does not occur for $L = \sum_{p=1}^{i-1}n_{p} + 1$ cases. 
After these permutations, $(2^{L}(j-1) + 2^{L-1} -2^{\sum_{p=1}^{i-1}n_{p}} + {l})$-th mode and $(2^{L}(j-1) + 2^{L-1} + {l})$-th mode are connected by the gates for all ${l} \in \left[2^{\sum_{p=1}^{i-1}n_{p}} \right]$. 
Accordingly, we set these gates as non-trivial gates and set the rest of the gates as trivial gates for each $L$ and $j$, thereby forming a grid structure up to depth $L$. 
Also, notice that the permutation in Eq.~\eqref{eq: permutation for general dimension} can be done parallelly for all $k \in \left\llbracket 1, 2^{L-1- \sum_{p=1}^{i-1}n_{p}}\right\rrbracket$ and  ${l} \in \left[2^{\sum_{p=1}^{i-1}n_{p}} \right]$.
Therefore, by denoting the permutation described in Eq.~\eqref{eq: permutation for general dimension} for all $k$ and $x$ as $\bm{P}_{L,j}^{(i)}$, we can construct $\bm{P}^{(i)}$ as 
\begin{align}
    \bm{P}^{(i)} = \prod_{L = \sum_{p=1}^{i-1}n_p + 2}^{\sum_{p=1}^{i}n_p} \prod_{j=1}^{2^{n - L}}\bm{P}_{L,j}^{(i)},
\end{align}
and accordingly, $\bm{P} = \prod_{i=1}^{d}\bm{P}^{(i)}$.

We now construct the circuit $C'$ in $\mathcal{B}$ for $i \geq 2$. 
In the ascending order of $i \geq 2$, for each $i$,
set all the gates in $L = \sum_{p=1}^{i-1}n_p  + 1$ non-trivial gates at first, because they already form a grid structure. 
Also, as we previously discussed, for each $L \in \left\llbracket\sum_{p=1}^{i-1}n_p + 2, \sum_{p=1}^{i}n_p \right\rrbracket $ and $j \in \left[ 2^{n - L}\right]$, we set the gates between $(2^{L}(j-1) + 2^{L-1} -2^{\sum_{p=1}^{i-1}n_{p}} + {l})$-th mode and $(2^{L}(j-1) + 2^{L} -2^{\sum_{p=1}^{i-1}n_{p}} + {l})$-th mode as the non-trivial gates for all ${l} \in \left[2^{\sum_{p=1}^{i-1}n_{p}}\right]$, where these modes are permuted from $(2^{L}(j-1) + 2^{L-2} + {l})$-th mode and $(2^{L}(j-1) + 2^{L-1} + 2^{L-2} + {l})$-th mode respectively in the previous step. 
Hence, given the circuit architecture $\mathcal{B}$, for each $L \in \left\llbracket\sum_{p=1}^{i-1}n_p + 2, \sum_{p=1}^{i}n_p \right\rrbracket $ and $j \in \left[2^{n - L}\right]$, we set the gates between $(2^{L}(j-1) + 2^{L-2} + {l})$-th mode and $(2^{L}(j-1) + 2^{L-1} + 2^{L-2} + {l})$-th mode for all ${l} \in \left[2^{\sum_{p=1}^{i-1}n_{p}}\right]$ as the only non-trivial gates, and leave the rest of the gates as trivial gates.
To specify these non-trivial gates from the gate set $\{G_{s,t}^{(i)}\}$, note that $(2^{L}(j-1) + 2^{L-1} -2^{\sum_{p=1}^{i-1}n_{p}} + {l})$-th mode and $(2^{L}(j-1) + 2^{L-1} + {l})$-th mode are connected by the non-trivial gates for all ${l} \in \left[2^{\sum_{p=1}^{i-1}n_{p}}\right]$ after the permutation for each $L \in \left\llbracket\sum_{p=1}^{i-1}n_p + 2, \sum_{p=1}^{i}n_p \right\rrbracket $ and $j \in \left[2^{n - L}\right]$. 
We also note that the indices along previous dimensions (characterized by ${l}$) were permuted by $\prod_{p=1}^{i-1}\bm{P}^{(p)}$ in the previous dimensions (working on $2^{\sum_{p=1}^{i-1}n_{p}}$ mode partition).
Therefore, to correct these permutations along previous dimensions, we set these non-trivial gates along $(2^{L}(j-1) + 2^{L-1} -2^{\sum_{p=1}^{i-1}n_{p}} + {l})$-th mode and $(2^{L}(j-1) + 2^{L-1} + {l})$-th mode (after the permutation) as $G_{s,t}^{(i)}$ for $s = 2^{L}(j-1) + 2^{L-1} -2^{\sum_{p=1}^{i-1}n_{p}} + {l}'$ and $t = s + 2^{\sum_{p=1}^{i-1}n_{p}}$, where ${l}'$ is a permuted version of ${l} \in \left[2^{\sum_{p=1}^{i-1}n_{p}}\right]$ by $\prod_{p=1}^{i-1}\bm{P}^{(p)}$.

To sum up, we can construct the circuit $C'$ in $\mathcal{B}$ as follows. 

\begin{enumerate}
    \item Along the first dimension $i = 1$, for depth $L = 1$ and $j \in \left[2^{n - 1}\right]$, we set the gate between the modes $2^{L}(j-1) + 2^{L-1}$ and $2^{L}(j-1) + 2^{L}$ in $\mathcal{B}$ as $G_{s,t}^{(1)}$ for $s = 2^{L}(j-1) + 2^{L-1}$ and $t = s + 1$.
    Then, for each $L = \llbracket 2,n_{1} \rrbracket$ and $j \in \left[2^{n - L}\right]$, we set the gate between the modes $2^L(j-1)+ 2^{L-2} + 1$ and $2^L(j-1)+ 2^{L-1} + 2^{L-2} + 1$ in $\mathcal{B}$ as $G_{s,t}^{(1)}$ for $s = 2^{L}(j-1) + 2^{L-1}$ and $t = s+1$, and leave trivial for the rest of the gates.
    
    \item In the ascending order of the dimension $i \geq 2$, for $L = \sum_{p=1}^{i-1}n_p + 1$ and $j \in \left[2^{n - L}\right]$, we set the gate between the modes $2^{L}(j-1) + 2^{L-1} -2^{\sum_{p=1}^{i-1}n_{p}} + {l}$ and $2^{L}(j-1) + 2^{L} -2^{\sum_{p=1}^{i-1}n_{p}} + {l}$ in $\mathcal{B}$ as $G_{s,t}^{(i)}$ for $s = 2^{L}(j-1) + 2^{L-1} -2^{\sum_{p=1}^{i-1}n_{p}} + {l}'$ and $t = s + 2^{\sum_{p=1}^{i-1}n_{p}}$, for all ${l} \in \left[2^{\sum_{p=1}^{i-1}n_{p}}\right]$ and ${l}'$ permuted from ${l}$ by $\prod_{p=1}^{i-1}\bm{P}^{(p)}$.
    Then, for each $L \in \left\llbracket\sum_{p=1}^{i-1}n_p + 2, \sum_{p=1}^{i}n_p \right\rrbracket $ and $j \in \left[2^{n - L}\right]$,
    we set the gates between the modes $2^{L}(j-1) + 2^{L-2} + {l}$ and $2^{L}(j-1) + 2^{L-1} + 2^{L-2} + {l}$ in $\mathcal{B}$ as $G_{s,t}^{(i)}$ for $s = 2^{L}(j-1) + 2^{L-1} -2^{\sum_{p=1}^{i-1}n_{p}} + {l}'$ and $t = s + 2^{\sum_{p=1}^{i-1}n_{p}}$,
    for all ${l} \in \left[2^{\sum_{p=1}^{i-1}n_{p}}\right]$ and ${l}'$ permuted from ${l}$ by $\prod_{p=1}^{i-1}\bm{P}^{(p)}$.
\end{enumerate}

Combining the arguments so far, this circuit $C'$ permuted by the proper permutation matrix $\bm{P}$ described above finally gives the desired grid circuit $C = \bm{P}C'\bm{P}^{T}$ 
(see, e.g., Fig.~\ref{fig:Rearranging index} (b)).

\end{proof}

To help understand the encoding scheme in the proof, we leave some illustrations in Fig~\ref{fig:Rearranging index} for 16-mode circuit architecture $\mathcal{B}$ constructing 1- and 2-dimensional grid linear optical circuits using a proper mode index permutation.

\subsection{Summary: Encoding the worst-case circuit in Lemma~\ref{Brod} to $\mathcal{B}\mathcal{B}^*$}

To sum up, the worst-case linear optical circuit $C_0$ in Fig.~\ref{fig:mbqc_to_grid} (b) can be implemented in a single parallel application of 3-dimensional grid linear optical circuit from $x$- to $z$-axis, with additional parallel beam splitter application along $x$-axis.
Here, as we stated, the size of the grid linear optical circuit can be arbitrarily enlarged as long as it contains the worst-case circuit $C_0$, which leaves the total post-selection probability $c_{Q}$ given in Lemma~\ref{Brod} unchanged.
Also, by Theorem~\ref{theorem: butterfly to grid}, one can implement a single parallel application of the 3-dimensional grid linear optical circuit in the circuit architecture $\mathcal{B}$ under a mode index permutation.
Combining these arguments, the worst-case circuit $C_0$ in Lemma~\ref{Brod} can be encoded in $\mathcal{B}^*\mathcal{B}$ under a mode index permutation. 
Here, note that $\mathcal{B}\mathcal{B}^*$ and $\mathcal{B}^*\mathcal{B}$ are equivalent under a mode index permutation by reverse-indexing each binary number of mode index (e.g., $10010_2 \rightarrow 01001_2$, for mode indices given by $[00000_2, 11111_2]$).  
Therefore, this implies that $C_0$ can be encoded in $\mathcal{B}\mathcal{B}^*$ under a mode index permutation, thereby finally concluding the proof of Lemma~\ref{lemma: embeding}.

We remark that our encoding scheme to implement the worst-case circuit $C_0$ in the circuit architecture $\mathcal{B}\mathcal{B}^*$ is suboptimal. 
Since the worst-case circuit can be implemented in constant depth as depicted in Fig.~\ref{fig:mbqc_to_grid} (b), there would possibly exist easier or more straightforward strategies to encode the worst-case circuit in the log-depth circuit architecture $\mathcal{B}\mathcal{B}^*$ by allocating mode indices more carefully.
Still, our analysis in this appendix is sufficient for our worst-case hardness result in our current work, i.e., as depicted in Theorem~\ref{worstcase}.

{
\section{About Cayley transform}\label{appendix: section: cayley}
}

As we depicted in the main text, our goal is to perturb the average-case circuit with the worst-case circuit, and it is crucial to choose a proper circuit perturbation method to establish the worst-to-average-case reduction successfully.
Throughout this work, we use the Cayley path for the perturbation, which was employed in Refs.~\cite{movassagh2023hardness, bouland2022noise, kondo2022quantum} for the hardness proposals of the random circuit sampling.

\begin{definition}[Cayley transform~\cite{movassagh2023hardness}]\label{cayley}
The Cayley transform of an $n$ by $n$ unitary matrix $H$ parameterized by $\theta \in [0,1]$ is a unitary matrix defined as 
\begin{equation}
    H(\theta) \coloneqq ((2-\theta)H + \theta I_{n})(\theta H + (2-\theta)I_{n})^{-1},
\end{equation}
where $I_{n}$ is the $n$ by $n$ identity matrix. Also, for the diagonalization of the $n$ by $n$ unitary matrix $H = LDL^{\dag}$, with unitary matrix $L$ and diagonal matrix $D = \text{diag}(e^{i\phi_1},\dots,e^{i\phi_n})$, the equivalent form of the Cayley transform is 
\begin{equation}
    H(\theta) = \frac{1}{q(\theta)} L\;\text{diag}( \left\{ p_{j}(\theta) \right\}_{j=1}^{n} )\,{\it{L}}^{\dag},
\end{equation}
where
\begin{align}
    &q(\theta) = \prod_{j=1}^{n} (1 + i\theta e^{i\frac{\phi_j}{2}}\sin{\frac{\phi_j}{2}}),  \label{q theta} \\ 
    &p_j(\theta) = e^{i\phi_j} (1 - i\theta e^{-i\frac{\phi_j}{2}}\sin{\frac{\phi_j}{2}}) \prod_{k\in[n]\setminus j} (1 + i\theta e^{i\frac{\phi_k}{2}}\sin{\frac{\phi_k}{2}}).
\end{align}

\end{definition}
Using the Cayley transform defined above, we now define the perturbed random circuit distribution, which is the local random circuit ensemble in Definition~\ref{randomcircuit} with each Haar random gate perturbed by the Cayley transform.

\begin{definition}[Perturbed local random circuit ensemble]\label{perturbed circuit instances}

Let $\mathcal{A}$ be the circuit architecture with $m$ number of gates. For the given circuit $C$ in $\mathcal{A}$ with gates $\{G_i\}_{i=1}^{m}$, the circuit $V(\theta)$ is defined with each gate of $C$ replaced by $G_i \rightarrow H_i(\theta)G_i$, where each $H_i(\theta)$ is a Cayley transform of independently distributed local Haar random gate $H_i$ $(i\in [m])$ parameterized by $\theta \in [0,1]$ in Definition~\ref{cayley}. We define $\mathcal{H}_{\mathcal{A}, \theta}^{C}$ as the distribution for such $V(\theta)$. Here, the distribution of the $V(0)$ is the local random circuit ensemble $\mathcal{H}_{\mathcal{A}}$ in Definition~\ref{randomcircuit}, and $V(1) = C$. 

\end{definition}

We should make sure that the success probability of average-case approximation over circuits is still large enough after the perturbation described in Definition~\ref{perturbed circuit instances}, to establish the reduction process successfully.
This is evident in the case that the total variation distance over circuits induced by the perturbation is small enough, as the success probability over circuits perturbs by, at most, the total variation distance.
In fact, Ref.~\cite{movassagh2023hardness} proved that total variation distance between $\mathcal{H}_{\mathcal{A}}$ and $\mathcal{H}_{\mathcal{A}, \theta}^{C}$ is small for sufficiently small perturbation $\theta$ compared to the inverse gate number of the circuit architecture $\mathcal{A}$.

\begin{lemma}[Movassagh~\cite{movassagh2023hardness}]\label{totalvariationdistance}
Let $\mathcal{A}$ be the circuit architecture with $m$ number of gates. For $\theta \ll 1$ and for any circuit $C$ in $\mathcal{A}$, total variation distance between $\mathcal{H}_{\mathcal{A}, \theta}^{C}$ and $\mathcal{H}_{\mathcal{A}}$ is $O(m\theta)$.
\end{lemma}
Therefore, by using small $\theta = O(m^{-1})$, one can upper-bound the total variation distance by an arbitrarily small constant, which implies that the success probability of average-case approximation over circuits (i.e., $1 - \delta$ in Theorem~\ref{averagehardness}) also perturbs by at most a small constant.

{
As we have sketched in Sec.~\ref{Section:Averagecase}, we infer the worst-case output probability $p_{\bm{s}}(C\bm{P})$ for some fixed circuit $C$ and permutation matrix $\bm{P}$ via average-case output probabilities $p_{\bm{s}}(V(\theta)\bm{P})$ for perturbed random circuit $V(\theta)$.
To investigate the viability of the inference of the worst-case value, we examine the behavior of the function $p_{\bm{s}}(V(\theta)\bm{P})$ characterized by the parameter $\theta$.
Using Definition~\ref{cayley}, we find that $p_{\bm{s}}(V(\theta)\bm{P})$ can be represented as a rational function in $\theta$.

\begin{lemma}\label{rationalfunction}
Let $\mathcal{A}$ be the $q$-Kaleidoscope circuit architecture $(\mathcal{B}\mathcal{B}^*)^{q}$ with $m = qM\log M$ number of gates, and $V(\theta) \sim \mathcal{H}_{\mathcal{A},\theta}^{C}$ for any $C$ in $\mathcal{A}$. Then for any collision-free outcome $\bm{s}$ and $M$ by $M$ permutation $\bm{P}$, the output probability $p_{\bm{s}}(V(\theta)\bm{P})$ can be represented as a degree $(4mN, 4mN)$ rational funtion in $\theta$.
\end{lemma}
}
\begin{proof}
For given circuit unitary matrix $V(\theta) \sim \mathcal{H}_{\mathcal{A},\theta}^{C}$ with $C$ composed of $\{G_i\}_{i=1}^{m}$ gates, one can decompose $V(\theta)$ with $m = qM\log M$ product of unitary matrices, such that each matrix element of $V(\theta)$ can be represented as
\begin{equation}
    \left[V(\theta)\right]_{j,k} = \sum_{l_1=1}^{M}\sum_{l_2=1}^{M}\cdots\sum_{l_{m-1}=1}^{M} V_{j,l_1}^{(1)}V_{l_1,l_2}^{(2)}\cdots V_{l_{m-1},k}^{(m)},
\end{equation}
where each $V^{(i)}$ denotes an $M$-dimensional unitary matrix, with a single gate unitary matrix $H_i(\theta)G_i$ applied to the modes participating in the gate and identity for the rest of the modes.
For example, if the $i$th gate $H_i(\theta)G_i$ is a two-mode gate between the first two modes, $V^{(i)} $ is a block diagonal matrix of $ H_i(\theta)G_i$ and identity matrix, namely, $V^{(i)} = H_i(\theta)G_i \bigoplus I_{M-2}$.

For circuit architecture $\mathcal{A} = (\mathcal{B}\mathcal{B}^*)^{q}$ which is composed of only two-mode gates, matrix elements of $V^{(i)}$ can be represented as degree $(2,2)$ rational functions in $\theta$, where the common denominator for the elements is given by $q_i(\theta)$, defined in Eq. (\ref{q theta}) but with index $i$ appended for $i$th random gate $H_i(\theta)$.
Using reduction to the common denominator for all of the $m$ gates, $\left[V(\theta)\right]_{j,k}$ can be represented as $(2m,2m)$ rational function in $\theta$ with the common denominator $\prod_{i=1}^{m} q_i(\theta)$; note that it does not change with the indices $j,k$.

{
Now, by Eq.~\eqref{outputprobability}, the output probability $p_{\bm{s}}(V(\theta)\bm{P})$ for collision-free outcome $\bm{s}$ has the form of
\begin{align}
    \begin{split}
        p_{\bm{s}}(V(\theta)\bm{P}) &= \left|\sum_{\sigma \in S_N}\prod_{j = 1}^{N} \left[V(\theta)_{\bm{s},\bm{t}^*}\right]_{\sigma_j, j} \right|^2,
    \end{split}
\end{align}
where we denote $\bm{t}^*$ as an input configuration vector permuted by $\bm{P}^{-1}$ from the fixed input configuration $\bm{t}$ in Eq.~\eqref{outputprobability}, and $S_N$ is $N$-mode permutation group. } 
One can easily check that the common denominator for $\prod_{j = 1}^{N} \left[V(\theta)_{\bm{s},\bm{t}^*}\right]_{\sigma_j, j}$ is $[\prod_{i=1}^{m} q_i(\theta) ]^N$, and it does not change with permutation $\sigma$.
Let us define $Q(\theta) = [\prod_{i=1}^{m} |q_i(\theta)|^2]^N$, which is a degree $4mN$ polynomial in $\theta$.
Then, $Q(\theta)$ serves as the common denominator for the output probability.
Hence, the output probability can be represented as $p_{\bm{s}}(V(\theta)) = \frac{P(\theta)}{Q(\theta)}$, with $P(\theta)$ also a degree $4mN$ polynomial in $\theta$.
\end{proof}

\section{Proof of Theorem~\ref{averagehardness} (average-case hardness)}\label{proof of averagecase hardness}

In the following, we provide a proof of Theorem~\ref{averagehardness}.
As we have already sketched the proof in the main text, for the sake of completeness, we will provide technical details here.
Specifically, we establish the worst-to-average-case reduction:
We prove that, for $\delta,\,\eta \geq 0$ with $\delta + \eta < \frac{1}{4}$, the output probability $p_{\bm{s}_0}(C_0)$ in Theorem~\ref{worstcase} can be estimated in additive imprecision $2^{-O(N)}$ in complexity class $\text{BPP}^{\text{NP}}$, given output probability estimations $p_{\bm{s}}(U)$ within additive imprecision $\epsilon = 2^{-O(N^{\gamma+1}(\log N)^2)}$, with probability at least $1 - \delta$ over $U \sim \mathcal{H}_{\mathcal{A}}\mathcal{P}$ with $\mathcal{A} = (\mathcal{B}\mathcal{B}^*)^{q\ge 3}$ for at least $1 - \eta$ over $\bm{s} \sim \mathcal{G}_{M,N}$.

\begin{proof}[Proof of Theorem~\ref{averagehardness}]

Let $\mathcal{O}$ be an oracle that solves the problem in Theorem~\ref{averagehardness}, i.e., on input $\bm{s} \sim \mathcal{G}_{M,N}$ and $U \sim \mathcal{H}_{\mathcal{A}}\mathcal{P}$ with $\mathcal{A} = (\mathcal{B}\mathcal{B}^*)^{q\ge 3}$, the oracle outputs $p_{\bm{s}}(U)$ within additive error $\epsilon$ with $1-\delta$ portion of $U$ for $1 - \eta$ portion of $\bm{s}$. 
In other words, for $1 - \eta$ of $\bm{s} \sim \mathcal{G}_{M,N}$ and $\mathcal{A} = (\mathcal{B}\mathcal{B}^*)^{q\ge 3}$, we have 
\begin{equation}\label{eq: average-case oracle}
	\Pr_{U \sim \mathcal{H}_{\mathcal{A}}\mathcal{P}} [|\mathcal{O}(\bm{s}, U) - p_{\bm{s}}(U) | > \epsilon ] < \delta .
\end{equation}

Let $C_0$ be the worst-case circuit in $(\mathcal{B}\mathcal{B}^*)^{q_0}$ with $q_0 \geq 1$, and let $\bm{s}_0$ be the fixed collision-free output given in Theorem~\ref{worstcase}. 
In the following, we show that approximating $p_{\bm{s}_0}(C_0)$ to within additive error $2^{-O(N)}$ (i.e., solving the problem in Theorem~\ref{worstcase}) is in $\text{BPP}^{\text{NP}^{\mathcal{O}}}$, which implies that the average-case approximation of $p_{\bm{s}}(U)$ to within $\epsilon$ is \#P-hard under $\text{BPP}^{\text{NP}}$ reduction.

{
As we have argued in the main text, we first prepare $M$ by $M$ permutation matrices $\bm{P}_0$ and $\bm{P}_1$ each sampled uniformly over the $M!$-number of possible permutations. 
Using these sampled permutations, we newly define permuted outcome $\bm{s}_p$ and permuted circuit $C_p$ as 
\begin{align}\label{def: permuted outcome and permuted circuit}
    \bm{s}_{p} \coloneqq \bm{P}_0\bm{s}_0, \quad C_{p} \coloneqq \bm{P}_0C_0\bm{P}_1^{-1} ,
\end{align}
where one can easily check that the permuted outcome follows random outcome ensemble $\bm{s}_{p} \sim \mathcal{G}_{M,N}$ in Definition~\ref{randomoutcome}. 
Also, the permuted circuit $C_{p}$ can be implemented in the circuit architecture $(\mathcal{B}\mathcal{B}^*)^{q_0+2}$, because each of the sampled $\bm{P}_0$ and $\bm{P}_1$ can be implemented in the architecture $\mathcal{B}\mathcal{B}^*$ by Lemma~\ref{permutation}.

Note that, by definition in Eq.~\eqref{def: permuted outcome and permuted circuit}, the worst-case output probability can also be represented as $p_{\bm{s}_0}(C_0) = p_{\bm{s}_{p}}(\bm{P}_0C_0) = p_{\bm{s}_{p}}(C_{p}\bm{P}_1)$.
Hereafter, we consider $C_{p}$ as a ``revised" worst-case circuit, and conduct circuit level worst-to-average-case reduction, from the worst-case circuit $C_{p}$ to the average-case circuit in $(\mathcal{B}\mathcal{B}^*)^{q_0+2}$.

For the average-case circuit, we sample the perturbed random circuit $V(\theta)$ from $\mathcal{H}_{\mathcal{A},\theta}^{C_{p}}$ defined in Definition~\ref{perturbed circuit instances}, for the circuit architecture $\mathcal{A} = (\mathcal{B}\mathcal{B}^*)^{q}$ with $q = q_0 + 2 \geq 3$ from now on, and the worst-case circuit $C_{p}$.
This can be done by sampling independently distributed local Haar random gates $\{H_i\}_{i=1}^{m}$ for gate number $m = qM\log M$, perturbing them by the Cayley transform parameterized by $\theta$ according to the gates composing the worst-case circuit $C_{p}$.
Then, by Definition~\ref{perturbed circuit instances}, $V(1) = C_{p}$ such that $p_{\bm{s}_{p}}(V(1)\bm{P}_1) = p_{\bm{s}_0}(C_0)$ is the worst-case output probability we aim to estimate.

Given the above $V(\theta)$, we multiply the sampled permutation $\bm{P}_1$ at the right side of $V(\theta)$. 
Because $V(0)\bm{P}_1$ follows the distribution $\mathcal{H}_{\mathcal{A}}\mathcal{P}$, one can expect that the distribution of $V(\theta)\bm{P}_1$ is close to $\mathcal{H}_{\mathcal{A}}\mathcal{P}$ for sufficiently small $\theta$.
Then, one can also expect that the success probability of the oracle $\mathcal{O}$ (i.e., the right-hand side of Eq.~\eqref{eq: average-case oracle}) on input a perturbed circuit $V(\theta)\bm{P}_1$ instead of $V(0)\bm{P}_1$ would not much be deviated from $\delta$ as long as $\theta$ is small enough.

Based on this understanding, we input $\bm{s}_{p}$ and $V(\theta)\bm{P}_1$ in the oracle $\mathcal{O}$.
Then, by Eq.~\eqref{eq: average-case oracle}, for at least $1 - \eta$ over $\bm{s}_{p}$, the failure probability of $\mathcal{O}$ is at most 
\begin{align}
\Pr \left[\left|\mathcal{O}(\bm{s}_{p}, V(\theta)\bm{P}_1) - p_{\bm{s}_{p}}(V(\theta)\bm{P}_1) \right| > \epsilon \right] 
< \delta + D_{\rm{TV}}(\mathcal{H}_{\mathcal{A},\theta}^{C_{p}}, \mathcal{H}_{\mathcal{A}}), \label{failureforcircuit}
\end{align}
where $D_{\rm{TV}}$ denotes total variation distance. 
The inequality in Eq.~\eqref{failureforcircuit} is evident based on the two facts: We can interpret the total variation distance as the supremum over events of the difference in probabilities of those events (Viz., circuits corresponding to the failure)~\cite{bouland2019complexity}, and the equal permutation on the circuit does not affect the total variation distance over the circuit space.
In the right-hand side of Eq.~\eqref{failureforcircuit}, by Lemma~\ref{totalvariationdistance}, $D_{\rm{TV}}(\mathcal{H}_{\mathcal{A},\theta}^{C_{p}}, \mathcal{H}_{\mathcal{A}})$ is $O(m\theta)$.
By setting $0 \le \theta \le \Delta$ with $\Delta = O(m^{-1})$, we can upper bound $D_{\rm{TV}}(\mathcal{H}_{\mathcal{A},\theta}^{C_{p}}, \mathcal{H}_{\mathcal{A}})$ by an arbitrarily small constant. 
}

By Lemma~\ref{rationalfunction}, $p_{\bm{s}_{p}}(V(\theta)\bm{P}_1)$ is a $(4mN, 4mN)$ degree rational function $\frac{P(\theta)}{Q(\theta)}$, where the denominator is given as $Q(\theta) = \left[\prod_{i=1}^{m} |q_i(\theta)|^2\right]^N$ for $q_i(\theta)$ defined in Definition~\ref{cayley} but with subindex $i$ corresponding to $i$th random gate $H_i(\theta)$.
We note that $Q(\theta)$ can be computed in $\Theta(m)$ time, as it only depends on the constant number of eigenvalues of local gate matrices (i.e., $\phi_j$ values in Eq.~\eqref{q theta} for each local gate matrix $H_i$).
Also, given that $\theta \le \Delta = O(m^{-1})$, $Q(\theta)$ is very close to the unity, because $Q(\theta) \ge 1$ and
\begin{align}\label{Q theta}
    Q(\theta) &= \left[\prod_{i=1}^{m} |q_i(\theta)|^2\right]^N \nonumber \\
    &=  \left[\prod_{i=1}^{m}\prod_{j=1}^{2} |(1 + i\theta e^{i\frac{\phi_{i,j}}{2}}\sin{\frac{\phi_{i,j}}{2}})|^2 \right]^N \\
    &\le \left[\prod_{i=1}^{m}\prod_{j=1}^{2} (1+\theta^2)\right]^N  \nonumber\\
    &\le (1 + O(m^{-2}))^{2mN}  \nonumber\\
    &= 1 + O(Nm^{-1}), \nonumber
\end{align}
where $\phi_{i,j}$ denotes the phase of the $j$th eigenvalue of the $i$th gate.

Therefore, $p_{\bm{s}_{p}}(V(\theta)\bm{P}_1)$ is very close to the degree $d = 4mN$ polynomial $P(\theta)$ in $\theta \in [0,\Delta]$, which allows us to use polynomial interpolation technique for $P(\theta)$.
Specifically, we obtain estimations of $P(\theta)$ for different values of $\theta \in [0,\Delta]$ by querying the oracle $\mathcal{O}$, use polynomial interpolation for given $P(\theta)$ values to estimate $P(1)$, and infer the value $p_{\bm{s}_{p}}(V(1)\bm{P}_1) = p_{\bm{s}_0}(C_0)$ by multiplying $Q(1)^{-1}$.
However, $Q(1)^{-1}$ becomes arbitrarily large for the case that even a single $\phi_{i,j}$ in Eq.~\eqref{Q theta} is near $\pm\pi$, which will arbitrarily enlarge the imprecision of the approximation of $p_{\bm{s}_{p}}(V(1)\bm{P}_1)$.
To avoid this issue, we employ the strategy from Ref.~\cite{bouland2022noise}, which only considers the case that all $\phi_{i,j}$ values of randomly chosen gates $\{H_i\}_{i=1}^{m}$ are in $[-\pi + \zeta, \pi - \zeta]$, and regards the other case as failure. 
This happens with probability at least $1 - O(m\zeta)$ over the random circuit instances.
By setting $\zeta = O(m^{-1})$, we can make $O(m\zeta)$ arbitrarily small constant, and as a result, we can upper bound $Q(1)^{-1}$ (see Eq.~\eqref{lowerboundofQ(1)} below) with high probability over random circuit instances.

Now the problem reduces to approximating degree $d = 4mN$ polynomial $P(\theta)$ with the value $\mathcal{O}(\bm{s}_{p}, V(\theta)\bm{P}_1)Q(\theta)$ in $\theta \in [0,\Delta]$ within additive error smaller than $\epsilon(1 + O(Nm^{-1})) \simeq \epsilon$; such approximations will later be used for the estimation of the value $P(1)$ via polynomial interpolation technique.
The failure of $\mathcal{O}$ depends on the outcome $\bm{s}_{p} \sim \mathcal{G}_{M,N}$ whose failure probability is at most $\eta$, and the circuit $V(\theta) \sim \mathcal{H}_{\mathcal{A},\theta}^{C_{p}}$ whose failure probability is at most $\delta + O(m\Delta)$ from Eq.~\eqref{failureforcircuit}.
Also, the probability that at least one $\phi_{i,j}$ of randomly chosen gates $\{H_i\}_{i=1}^{m}$ is outside of the regime $[-\pi + \zeta, \pi - \zeta]$ is at most $O(m\zeta)$. 
Putting everything together and applying a simple union bound, the total failure probability of the approximation of $P(\theta)$ is at most
{
\begin{align}\label{failureforall}
    \Pr \left[\left|\mathcal{O}(\bm{s}_{p}, V(\theta)\bm{P}_1)Q(\theta) - P(\theta) \right| > \epsilon \right] &< \eta + \delta + O(m\Delta) + O(m\zeta)  \le \delta',
\end{align}
}
where $\delta'$ is an upper bound of $\eta + \delta + O(m\Delta) + O(m\zeta)$, and given $\eta + \delta < \frac{1}{4}$, we can make $\delta' < \frac{1}{4}$ by setting $O(m\Delta)$ and $O(m\zeta)$ arbitrary small constants.

Let $\{\theta_i\}_{i=1}^{O(d^2)}$ be the set of equally spaced points in the interval $[0,\Delta]$.
For each $\theta_i$, we obtain the unitary matrix $U(\theta_i)$ using the same random gate $\{H_i\}_{i=1}^{m}$ and worst-case circuit $C_{p}$. 
Let $y_i = \mathcal{O}(\bm{s}_{p}, V(\theta_i)\bm{P}_1)Q(\theta_i)$. 
By Eq.~\eqref{failureforall}, each set of points $(\theta_i, y_i)$ satisfies
\begin{equation}
    \Pr [|y_i - P(\theta_i) | > \epsilon ] \le \delta' < \frac{1}{4}.
\end{equation}
By using the interpolation algorithm introduced in Theorem~\ref{bouland}, we can obtain the additive approximation of $P(1)$ as $\tilde{p}$ with an access to NP oracle, such that
\begin{equation}\label{failureofrobustBW}
    \Pr[|\tilde{p} - P(1)| > \epsilon'] < \frac{1}{3},
\end{equation}
where $\epsilon' = \epsilon e^{-d\log{\Delta}} = \epsilon 2^{O(N^{\gamma+1}(\log N)^2)}$ using $d = 4mN$ and $m = qM\log M$. 
Note that the failure probability in Eq.~\eqref{failureofrobustBW} can be arbitrarily reduced by taking a polynomial number of trials, and thus we can obtain the estimated value $P(1)$ within additive error $\epsilon'$ with arbitrarily high probability.

From the estimated value $P(1)$, we can infer the worst-case output probability value $p_{\bm{s}_0}(C_0) = p_{\bm{s}_{p}}(V(1)\bm{P}_1) = P(1)/Q(1)$.
As the value of $Q(1)$ depends on the values $\phi_{i,j}$ in Eq.~\eqref{Q theta}, the $\phi_{i,j}$ independent lower bound of $Q(1)$ is required to set an upper bound of the additive imprecision of $p_{\bm{s}_{p}}(V(1)\bm{P}_1)$. 
Since we only consider the case that all of $\phi_{i,j}$ values are in $[-\pi + \zeta, \pi - \zeta]$ with $\zeta = O(m^{-1})$ for all randomly chosen gates $\{H_i\}_{i=1}^{m}$, we have
\begin{align}\label{lowerboundofQ(1)}
    \begin{split}
        Q(1) &=  \left[\prod_{i=1}^{m}\prod_{j=1}^{2} |(1 + ie^{i\frac{\phi_{i,j}}{2}}\sin{\frac{\phi_{i,j}}{2}})|^2 \right]^N \\
        &= \left[\prod_{i=1}^{m}\prod_{j=1}^{2} \left( 1 - \sin^2{\frac{\phi_{i,j}}{2}} \right) \right]^N \\
        &\ge \left[\prod_{i=1}^{m}\prod_{j=1}^{2} \left( 1 - \sin^2{\frac{\pi - \zeta}{2}} \right) \right]^N \\
        &= \left( O(m^{-2}) \right)^{2mN} \\
        &= 2^{2mN \log O(m^{-2})}.
    \end{split}
\end{align}

Therefore, the total additive error for estimating $p_{\bm{s}_{p}}(V(1)\bm{P}_1)$ is bounded by $\epsilon'2^{-2mN \log O(m^{-2})} = \epsilon 2^{O(N^{\gamma+1}(\log N)^2)}$. 
By setting $\epsilon = 2^{-O(N^{\gamma+1}(\log N)^2)}2^{-O(N)} = 2^{-O(N^{\gamma+1}(\log N)^2)}$, we can estimate the worst-case output probability value $p_{\bm{s}_{p}}(V(1)\bm{P}_1) = p_{\bm{s}_0}(C_0)$ within additive error $2^{-O(N)}$, and the whole reduction process is in $\text{BPP}^{\text{NP}}$. 
This completes the proof. 
    
\end{proof}

For the polynomial interpolation, we employ the Robust Berlekamp-Welch algorithm that requires the NP oracle, recently proposed in Ref.~\cite{bouland2022noise}.

\begin{theorem}[Robust Berlekamp-Welch~\cite{bouland2022noise}]\label{bouland}
Let $P(x)$ be a degree $d$ polynomial in $x$. 
Suppose there is a set of points $D = \{(x_i,y_i)\}$
such that $|D| = O(d^2)$ and $\{x_i\}$ is equally spaced in the interval $[0,\Delta]$. 
Suppose also that each points $(x_i,y_i)$ satisfies
\begin{equation}
    \Pr[|y_i - P(x_i)|\ge \epsilon] \le \delta,
\end{equation}
with $\delta < \frac{1}{4}$. Then there exists a $\rm{P}^{\rm{NP}}$ algorithm that takes input $D$ and outputs $\tilde{p}$ such that
\begin{equation}
    |\tilde{p} - P(1)| \le \epsilon e^{-d\log\Delta},
\end{equation}
with success probability at least $\frac{2}{3}$.
\end{theorem}

\section{Proof of Lemma~\ref{aaronsonandarkhipov}}\label{proof of averagecase approximation}

Let $\bar{p}_{\bm{s}}(C)$ be the output probability distribution from the approximate sampler $\mathcal{S}$ with the given linear optical circuit $C$. 
Also, let $\mathcal{C}_{M,N}$ be the set of collision-free outcomes of Boson Sampling, for mode number $M$ and photon number $N$.
Then $\bar{p}_{\bm{s}}(C)$ satisfies 
\begin{align}\label{averageofdifference}
\E_{\bm{s}\sim\mathcal{G}_{M,N}} \left[|\bar{p}_{\bm{s}}(C) - p_{\bm{s}}(C)|\right] = \frac{1}{\binom{M}{N}} \sum_{\bm{s}\in\mathcal{C}_{M,N}}|\bar{p}_{\bm{s}}(C) - p_{\bm{s}}(C)|
\le \frac{2\beta}{\binom{M}{N}}.
\end{align}
Using Eq.~\eqref{averageofdifference} and Markov's inequality, $\bar{p}_{\bm{s}}(C)$ satisfies
\begin{align}
    \Pr_{\bm{s}\sim\mathcal{G}_{M,N}} \left[|\bar{p}_{\bm{s}}(C) - p_{\bm{s}}(C)| \ge \frac{\beta k}{\binom{M}{N}}\right] &\le \frac{2}{k} ,
\end{align}
for all $k > 2$.
Also, using Stockmeyer's algorithm~\cite{stockmeyer1985approximation} whose complexity is in $\rm{BPP}^{\rm{NP}}$, obtaining the estimate $\tilde{p}_{\bm{s}}(C)$ of $\bar{p}_{\bm{s}}(C)$ satisfying 
\begin{equation}
    \Pr\left[|\tilde{p}_{\bm{s}}(C) - \bar{p}_{\bm{s}}(C)| \ge \alpha \bar{p}_{\bm{s}}(C)\right] \le \frac{1}{2^N},
\end{equation}
in polynomial time in $N$ and $\alpha^{-1}$ is in $\rm{BPP}^{\rm{NP}^{\mathcal{S}}}$. Using $\E_{\bm{s} \sim \mathcal{G}_{M,N}}[\bar{p}_{\bm{s}}(C)] = \binom{M}{N}^{-1} \sum_{\bm{s}\in\mathcal{C}_{M,N}}\bar{p}_{\bm{s}}(C) \le \binom{M}{N}^{-1}$,  
\begin{align}
    \Pr \left[|\tilde{p}_{\bm{s}}(C) - \bar{p}_{\bm{s}}(C)| \ge \frac{\alpha l}{\binom{M}{N}} \right] &\le \Pr \left[ \bar{p}_{\bm{s}}(C) \ge \frac{l}{\binom{M}{N}} \right] + \Pr \left[|\tilde{p}_{\bm{s}}(C) - \bar{p}_{\bm{s}}(C)| \ge \alpha \bar{p}_{\bm{s}}(C)\right] \\
    &\le \frac{1}{l} + \frac{1}{2^N},
\end{align}
for all $l > 1$. 
Putting all together, by applying a triangular inequality, finding an average-case approximation $\tilde{p}_{\bm{s}}(C)$ of $p_{\bm{s}}(C)$ satisfying
\begin{align}
    \Pr \left[|\tilde{p}_{\bm{s}}(C) - p_{\bm{s}}(C)| \ge \frac{\beta k + \alpha l}{\binom{M}{N}} \right] &\le \frac{2}{k} + \frac{1}{l} + \frac{1}{2^N}
\end{align}
is in $\rm{BPP}^{\rm{NP}^{\mathcal{S}}}$.
Let $\kappa$ and $\xi$ be fixed error parameters such that $k/2 = l = 3/\xi$ and $\beta = \kappa\xi/{12} = \alpha/2$. As $\beta k + \alpha l = \kappa$ and $\frac{2}{k} + \frac{1}{l} + \frac{1}{2^N} = \frac{2}{3}\xi + \frac{1}{2^N} \le \xi$, we finally obtain Eq.~\eqref{averagecaseapproximation}.

\section{On average-case hardness including collision outcomes}\label{appendix: section: collision}

As we have remarked in the main text, extending our hardness result of shallow-depth Boson Sampling to the saturated regime $M = \Theta(N^{\gamma})$ for $1 \leq \gamma < 2$ remains a crucial open problem, and to do so, one also needs to consider collision outcomes in the average-case hardness argument.
In this appendix, we present a possible approach that can be further improved to the average-case hardness result including collision outcomes.

Our intuition is that while most of the outcomes in the saturated regime are not collision-free, the key point is that most outcomes are expected to be ``nearly" collision-free, in the sense that for most of the outcomes, a large portion of the measured modes are occupied only by a single photon.
Also, 
we can establish worst-to-average-case reduction across different system sizes (i.e., mode number \& photon number), i.e.,  the system size for the worst-case hardness result can be smaller at most polynomially than the system size for the average-case hardness result.  
Hence, by post-selecting outcomes with the number of measured single photons at least as large as the photon number required in the worst-case problem (in Theorem~\ref{worstcase}, which requires a collision-free outcome), one would possibly establish worst-to-average-case reduction from the randomly chosen outcome over all possible outcomes.


More specifically, let the mode number $M = \Theta(N^{\gamma})$ for any $\gamma \geq 1$. 
Let us define $\,\mathcal{Q}_{M,N}$ as the uniform distribution over all the possible outcomes of Boson Sampling with $M$ modes and $N$ photons, such that each outcome $\bm{s} \sim \mathcal{Q}_{M,N}$ is an $M$-dimensional output configuration vector $\bm{s} = (s_1,\dots,s_M)$ with the constraint $\sum_{i=1}^{M}s_i = N$.
Also, we denote the number of $s_i = 1$ in $\bm{s}$ as the number of \textit{singletons} throughout this appendix. 
Here, let us assume that there exists $N_0 = cN^{\lambda}$ for a constant $c$ and $0<\lambda <1$ such that it is likely that the number of singletons is more than $N_0$; we later provide the intuition to justify this assumption at the end of this appendix.

Hereafter, we denote each $(M_0, N_0)$ and $(M, N)$ as the system sizes (mode number \& photon number) for the worst-case hardness (i.e., Theorem~\ref{worstcase}) and average-case hardness (we aim to prove), respectively.  
More specifically, given a worst-case circuit $C_0$ in $M_0 = \Omega(N_0)$ mode circuit architecture $\mathcal{B}\mathcal{B}^*$ and a fixed $M_0$-mode and $N_0$-photon collision-free outcome $\bm{s}_0$ depicted in Theorem~\ref{worstcase}, by our assumption, one can find a photon number $N$ that satisfies $N_0 = cN^{\lambda}$ for a constant $c$ and $0<\lambda <1$, and mode number $M = \Theta(N^{\gamma})$ for $\gamma \geq 1$, such that most of the outcomes over $\mathcal{Q}_{M,N}$ have at least $N_0$ singletons.
In the following, we show a possible approach for worst-to-average-case reduction for general outcomes, by describing how one can estimate $p_{\bm{s}_0}(C_0)$ to within desired imprecision given oracle access to the well-estimated values $p_{\bm{s}}(U)$ for average-case outcomes $\bm{s} \sim \mathcal{Q}_{M,N}$ and $M$-mode average-case circuits over $U \sim \mathcal{H}_{\mathcal{A}}$ with $\mathcal{A} = (\mathcal{B}\mathcal{B}^*)^{q}$.

At first, for the average-case instances, we sample $\bm{s} \sim \mathcal{Q}_{M,N}$ that has at least $N_0$ singletons, and sample $U \sim \mathcal{H}_{\mathcal{A}}$ specifically for $M$-mode circuit architecture $\mathcal{A} = (\mathcal{B}\mathcal{B}^*)^{q\geq2}$.
Subsequently, we sample $M/2$-dimensional configuration vector $\bm{s}_1$ through the following steps: $(1)$ discard $N_0$ singletons from $\bm{s}$; $(2)$ retain only the non-zero elements of $\bm{s}$ and discard all the zeros; $(3)$ append zeros as needed to construct $M/2$-dimensional vector $\bm{s}_1$, and $(4)$ randomly permute the elements of $\bm{s}_1$ uniformly. 
Then, given $\bm{s}$, $\bm{s}_0$ and $\bm{s}_1$, one can efficiently construct $M$ by $M$ permutation matrix $\bm{P}$ that satisfies
\begin{align}
    \bm{s} = \bm{P}(\bm{s}_0\oplus 0_{M/2-M_0}\oplus \bm{s}_1),
\end{align}
where we denote $0_{M/2-M_0}$ as a $(M/2 - M_0)$-dimensional vector composed on only zeros. 
In addition, let $\bm{t} = \bm{t}_0\oplus 0_{M/2-M_0}\oplus \bm{t}_1$ be an input configuration for the average-case output probability, where $\bm{t}_1$ is some fixed $M/2$-mode and $(N-N_0)$-photon configuration vector.

Next, given the sampled $U$, we denote the first-rightmost $M/2$-mode subcircuit of $U$ as $U_0$, such that $U_0$ is in $M/2$-mode circuit architecture $\mathcal{A} = \mathcal{B}\mathcal{B}^*$ and essentially follows the distribution $\mathcal{H}_{\mathcal{A}}$; we note that, $M/2$-mode $\mathcal{B}\mathcal{B}^*$ can be embedded into $M$-mode $\mathcal{B}\mathcal{B}^*$ (see, e.g., Fig.~\ref{butterflyforgaussian}).
Using the definitions so far, we construct the revised worst-case circuit $C \in (\mathcal{B}\mathcal{B}^*)^{q=2}$ as 
\begin{align}
    C = \bm{P}(C_0\oplus I_{M/2 - M_0} \oplus U_0), 
\end{align}
where we denote $I_{M/2-M_0}$ as a $(M/2 - M_0)$-sized identity matrix. 
We note that the revised worst-case circuit can be embedded in the higher depth circuit architecture $C \in (\mathcal{B}\mathcal{B}^*)^{q\geq2}$, by adding trivial gates in the front. 
Then, by definition, for input configuration $\bm{t}$, the output probability $p_{\bm{s}}(C)$ can be expressed as 
\begin{align}
    p_{\bm{s}}(C) = p_{\bm{s}_0}(C_0)p_{\bm{s}_1}(U_0).
\end{align}
Therefore, if one can well-estimate the revised worst-case value $p_{\bm{s}}(C)$ and the average-case value $p_{\bm{s}_1}(U_0)$, one can also infer the desired worst-case output probability $p_{\bm{s}_0}(C_0)$, thereby establishing worst-to-average-case reduction.

We now argue how to estimate the worst-case value $p_{\bm{s}}(C)$ and the average-case value $p_{\bm{s}_1}(U_0)$.
First, estimating the average-case value $p_{\bm{s}_1}(U_0)$ can be done by similarly accessing the oracle, but now we input the outcome $\bm{s}_1$ and the circuit $U_0$ instead.         
Hence, the success probability for the oracle would increase at most the total variation distance between the distribution of $\bm{s}_1$ and $\mathcal{Q}_{M/2,N-N_0}$, and if this can be sufficiently bounded, we can obtain the estimated value of $p_{\bm{s}_1}(U_0)$ with high probability. 
Also, one can estimate the value $p_{\bm{s}}(C)$ by the polynomial interpolation, as in the proof of Theorem~\ref{averagehardness}.
Specifically, by parameterizing the circuit $U$ as $U(\theta)$, where each gate of $U$ (except those in $U_0$) is perturbed via the Cayley transform using the corresponding gate from $C$, one can similarly infer the value $p_{\bm{s}}(C)$ through oracle access to $p_{\bm{s}}(U(\theta))$ and applying polynomial interpolation (see Appendix~\ref{proof of averagecase hardness} for details).

To sum up, the overall worst-to-average-case reduction process can be schematized as follows.
\begin{align}
    p_{\bm{s}_0}(C_0) =  \frac{p_{\bm{s}}(C)}{p_{\bm{s}_1}(U_0)} \quad\xrightarrow[\substack{U\rightarrow C  \\ \text{Cayley transform} \\ 
    \text{and} \\ \text{polynomial interpolation}}]{}\quad \frac{p_{\bm{s}}(U)}{p_{\bm{s}_1}(U_0)} \approx \frac{\mathcal{O}(\bm{s},U)}{\mathcal{O}(\bm{s}_1,U_0)},
\end{align}
where $\mathcal{O}$ represents the oracle for the average-case estimation problem.

We highlight the necessary ingredients that need to be addressed for the successful worst-to-average-case reduction by the above process.
First, to make sure that we can also well-estimate $p_{\bm{s}_1}(U_0)$ with high probability over $\bm{s}_1$ via the oracle access, one needs to justify that the distribution of $\bm{s}_1$ is sufficiently close to the distribution $\mathcal{Q}_{M/2,N-N_0}$.
To help understand how the distribution of $\bm{s}_1$ can be close to $\mathcal{Q}_{M/2,N-N_0}$, consider the limiting case that $N_0 = 0$. 
Here, given that $M$ is much larger than $N$ (but still can be $M = \Theta(N)$), sampling $\bm{s}_1$ is equivalent to sampling from $\mathcal{Q}_{M,N}$ and post-selecting outcomes that have all the non-zero elements in the first $M/2$ mode, which yields the uniform distribution $\mathcal{Q}_{M/2,N}$.
Therefore, given that $M$ is much larger than $N$ and $N_0$ is sufficiently small relative to $N$, we expect that the distribution of $\bm{s}_1$ would not largely deviate from the distribution $\mathcal{Q}_{M/2,N-N_0}$.

Also, as the final imprecision level increases by the estimated value of $p_{\bm{s}_1}(U_0)^{-1}$, it must be promised that $p_{\bm{s}_1}(U_0)$ is lower bounded.
More precisely, $p_{\bm{s}_1}(U_0)$ must be at least as large as the allowed imprecision level of average-case estimation (i.e., imprecision level of the oracle $\mathcal{O}$), commonly referred to as the anti-concentration property~\cite{aaronson2011computational, bouland2019complexity, Bremner2016average, Hangleiter2018anticoncentration, Bermejo2018architectures, Dalzell2022random}.  
Nevertheless, we believe that $p_{\bm{s}_1}(U_0)$ is at least as large as the allowed imprecision level for most of $U_0$, based on two reasons:
First, the imprecision level required by our current average-case hardness result is sufficiently small.
Also, although this has only been demonstrated in finite-sized numerical experiments, random circuits for our circuit architecture closely resemble global Haar-random unitary circuits~\cite{go2024exploring}, which are strongly conjectured to exhibit the anti-concentration property~\cite{aaronson2011computational}.

We now present a toy model that could offer an insight that, for a large portion of outcomes of Boson Sampling, a sufficient number of singletons would exist.
In the following, we present the classical ``balls-into-bins" problem where we throw $N$ balls into $M$ bins uniformly at random, which is indeed the classical analogue of Boson Sampling.
We show that, given that $M = \Theta(N^{\gamma})$ with $\gamma \geq 1$, it is highly unlikely that the number of bins getting exactly one ball is less than $cN^{\lambda}$ for a constant $c$ and $0<\lambda < 1$.

Let us first denote the random variable $X_i$ as the number of balls in $i$-th bin (such that $\sum_{i=1}^{M}X_i = N$), when we throw $N$ balls into $M$ bins uniformly at random. 
Here, our goal is to verify that over the experiments, it is highly unlikely to observe the outcomes that the number of $X_i = 1$ is smaller than $cN^{\lambda}$.
Since all the variables $X_i$ are correlated due to the constraint $\sum_{i=1}^{M}X_i = N$, instead, we consider i.i.d. random variable $Y_i$ each drawn from Poisson distribution $P(\frac{N}{M})$, such that $\E[\sum_{i=1}^{M}Y_{i}] = N$ and $\Pr[\sum_{i=1}^{M}Y_{i} = N] = e^{-N}\frac{N^N}{N!} = \Theta(\frac{1}{\sqrt{N}})$.
The motivation for using Poissonian random variables is that the distribution of $(Y_1,\dots, Y_M)$ conditioned on the sum $\sum_{i=1}^{M} Y_i = N$ is equivalent to the distribution of the balls-into-bins problem $(X_1,\dots, X_M)$ (see, e.g., Ref.~\cite{ross2014introduction}). 
Accordingly, if it is exponentially unlikely that the number of $Y_i = 1$ is smaller than $cN^{\lambda}$ as we sample $(Y_1,\dots, Y_M)$, this indicates that it is still exponentially unlikely that the number of $X_i = 1$ is smaller than $cN^{\lambda}$ as we sample $(X_1,\dots, X_M)$, given that the post-selection probability for $\sum_{i=1}^{M} Y_i = N$ is at most inverse polynomial (i.e., $\Theta(\frac{1}{\sqrt{N}})$).

Now, let us newly define $Z = \sum_{i=1}^{M} Z_i$ as the summation of the independent Bernoulli random variable $Z_i \in \{0,1\}$, where each $Z_i = 1$ only if $Y_i = 1$, and $Z_i = 0$ otherwise.
Then, $Z$ indicates the number of $Y_i = 1$, and thus our problem reduces to verifying that the probability $\Pr[Z \leq cN^{\lambda}]$ is exponentially small.
To proceed, as $Y_i \sim P(\frac{N}{M})$, note that the probability of obtaining $Z_i = 1$ is given by
\begin{align}
    \Pr[Z_i = 1] = \frac{N}{M}\exp(-\frac{N}{M}).
\end{align}
Then, the expectation value of $Z$ is given by
\begin{align}
    \E[Z] =  \sum_{i=1}^{M}\E[Z_i] =  \sum_{i=1}^{M}\Pr[Z_i = 1]  = N\exp(-\frac{N}{M}).
\end{align}
As all the $Z_i$ are independent Bernoulli random variables, we use the Chernoff bound, which gives the probability bound as
\begin{align}
    \Pr[Z \leq \varepsilon \E[Z])] \leq \exp\left( -\frac{(1-\varepsilon)^2}{2}\E[Z]\right) .
\end{align}
Hence, by setting $\varepsilon \E[Z] = cN^{\lambda}$, we finally obtain the bound
\begin{align}
    \Pr[Z \leq cN^{\lambda}] \leq \exp\left( -\frac{1}{2}\left(1 -  cN^{\lambda - 1}\exp(\frac{N}{M})\right)^2N \exp(-\frac{N}{M}) \right)= \exp\left(-\Theta(N)\right).
\end{align}
Therefore, for the classical balls-into-bins problem, it is exponentially unlikely that the number of bins getting exactly one ball is less than $cN^{\lambda}$.
Note that this bound even allows $\lambda = 1$ as long as $c\exp(\frac{N}{M}) < 1$, which implies that a constant portion of bins will get exactly one ball.

We emphasize that the distribution of the classical balls-into-bins problem is not equivalent to the uniform distribution over the possible outcomes of Boson Sampling, because the classical balls are distinguishable while identical bosons are indistinguishable in nature (and indeed, bosons are even more ``gregarious" than classical balls~\cite{aaronson2011computational}). 
Nevertheless, this example is sufficient to provide an insight that a sufficiently large number of singletons would appear in a large fraction of Boson Sampling outcomes.


We conclude this appendix by introducing an alternative approach, based on the closely related result in~\cite{bouland2023complexity} that establishes a hardness argument for general Boson Sampling in the saturated regime.
Note that our convention on singletons closely parallels the notion of ``clicks" introduced in~\cite{bouland2023complexity}, which refers to the number of modes that are occupied by a \textit{non-zero} number of photons.
Analogous to our analysis, their key observation is that, even in the saturated regime, most outcomes have a sufficiently large number of clicks.
Based on this observation, they show that, given an $N$-photon outcome $\bm{s}$ with collisions that has a sufficient number of clicks $N_0$, there exists a poly-time constructible matrix $A \in \{0,1 \}^{N_{0} \times N}$ such that computing $\text{Per}(A_{\bm{s}})$ is worst-case \#P-hard, for $A_{\bm{s}}$ being a submatrix of $A$ by repeating rows according to $\bm{s}$. 
Then they establish a worst-to-average-case reduction by perturbing $A$ into their average-case matrix and performing polynomial interpolation (for details, see~\cite{bouland2023complexity}).

Accordingly, since we employ a conceptually similar strategy in the main text, establishing worst-case hardness for shallow-depth circuits with a large enough number of clicks is sufficient to obtain an analogous hardness result for shallow-depth Boson Sampling.
This requires constructing a shallow-depth circuit (as done in the proof of Theorem~\ref{worstcase}) for which estimating an output probability with a large enough number of clicks is worst-case \#P-hard, in direct analogy to the construction of the worst-case matrix $A$ above.
Since this lies beyond the scope of the presented work, we leave this as future work.


{

\section{Numerical simulation for the collision patterns of Boson Sampling for the local random circuit over $(\mathcal{B}\mathcal{B}^*)^{q}$}\label{appendix: section: numerical evidence of bbp}

In this appendix, we present numerical simulation results of collision patterns of Boson Sampling, for solely the local random circuit ensemble in Definition~\ref{randomcircuit} over our shallow-depth circuit architecture $(\mathcal{B}\mathcal{B}^*)^{q}$ in Definition~\ref{def: kaleidoscope}.

To show the bosonic birthday paradox property for the local random circuit ensemble over $(\mathcal{B}\mathcal{B}^*)^{q}$, we show numerically that the local random circuit ensemble over $(\mathcal{B}\mathcal{B}^*)^{q}$ samples sufficiently many collision-free outcomes over the experiments. 
Specifically, we prepare random unitary circuits drawn from the local random circuit ensemble over $(\mathcal{B}\mathcal{B}^*)^{q}$, sample the outcomes by classical Boson Sampling algorithm for each randomly chosen circuit, and compute the ratio of the collision-free outcomes over all the sampled outcomes for each randomly chosen circuit. 
To evaluate whether these local random circuits output sufficiently many collision-free outcomes, as a comparison, we also prepare the global Haar random circuits, which have the bosonic birthday paradox property as depicted in~\cite{aaronson2011computational}.
In other words, if the ratio of collision-free outcomes for local random circuit ensemble over $(\mathcal{B}\mathcal{B}^*)^{q}$ is as large as the ratio of collision-free outcomes for global Haar random circuit, this clearly gives evidence of the bosonic birthday paradox for our circuit architecture.

\begin{figure*}[t]
\includegraphics[width=1\linewidth]{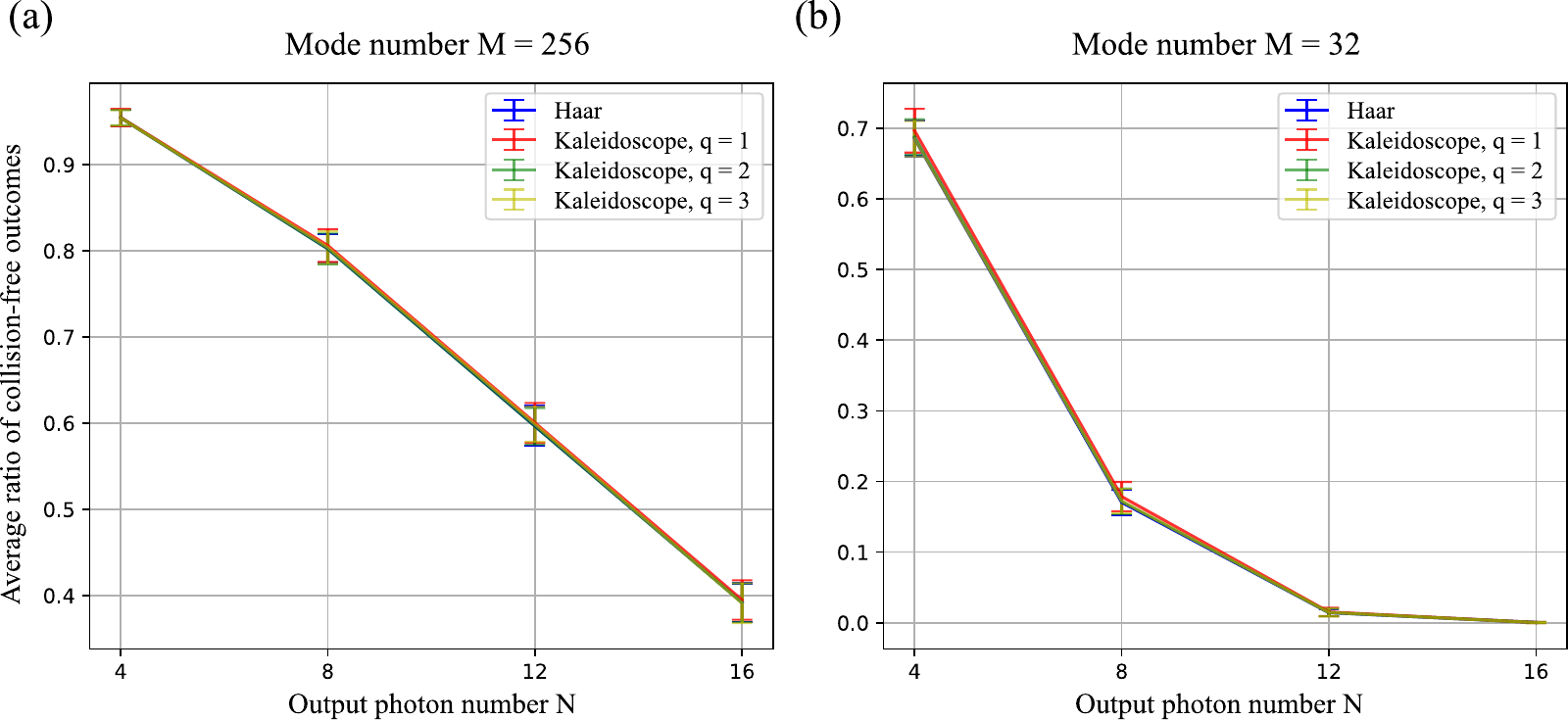}
\caption{
The average ratio of collision-free outcomes from (a) $256$-mode and (b) $32$-mode unitary circuits drawn from local random circuit ensemble $\mathcal{H}_{\mathcal{A}}$ over the $q$-Kaleidoscope circuit architecture $\mathcal{A} = (\mathcal{B}\mathcal{B}^*)^{q}$ with the repetition number $q = 1,2,3$, and global Haar random circuit ensemble. 
We sample 500 unitary circuits from each of the circuit ensembles, and we sample 500 outcomes on input $N = 4, 8, 12, 16$ photons for each of the randomly chosen circuits. 
Then we count the collision-free outcomes over the 500 sampled outcomes and compute the average ratio of collision-free outcomes over 500 randomly chosen circuits for each of the circuit ensembles. 
We present in the figure the average ratio of collision-free outcomes for local random circuits over $(\mathcal{B}\mathcal{B}^*)^{q}$ with $q = 1,2,3$ and global Haar random circuit ensemble, for the photon number given by $N = 4, 8, 12, 16$ (the error bar represents standard deviation). 
For all of the repetition numbers $q = 1,2,3$ and the photon numbers $N = 4, 8, 12, 16$, local random circuits for $(\mathcal{B}\mathcal{B}^*)^{q}$ exhibits almost equivalent average collision-free ratio to the global Haar random unitary circuits. 
}
\label{fig: Ratio}
\end{figure*}

More specifically, we consider $256$-mode random unitary circuits drawn from local random circuit ensemble $\mathcal{H}_{\mathcal{A}}$ over the $q$-Kaleidoscope circuit architecture $\mathcal{A} = (\mathcal{B}\mathcal{B}^*)^{q}$ with the repetition number $q = 1,2,3$, and global Haar random circuit ensemble. 
First, we sample 500 unitary circuits from each of the circuit ensembles described above.
For each of the randomly chosen circuits, on input $N = 4, 8, 12, 16$ photons in the first $N$ modes over $256$ modes, we sample 500 outcomes using the Clifford-Clifford algorithm for Boson Sampling~\cite{clifford2018classical}. 
Then, we count the number of collision-free outcomes over the 500 sampled outcomes and compute the average ratio of collision-free outcomes for each $N$ and circuit ensemble, where the average is taken over the $500$ sampled unitary circuits for each circuit ensemble.

In Fig.~\ref{fig: Ratio}, we present the average ratio of collision-free outcomes and standard deviation (error bar) for the photon number given by $N = 4, 8, 12, 16$, from the local random circuits over $(\mathcal{B}\mathcal{B}^*)^{q}$ with the repetition number $q = 1,2,3$ and global Haar random circuits.
Remarkably, local random circuits for $(\mathcal{B}\mathcal{B}^*)^{q}$ shows great resemblance to the global Haar random circuits: 
For all of the photon numbers $N = 4, 8, 12, 16$ and the repetition numbers $q = 1,2,3$, local random circuits for $(\mathcal{B}\mathcal{B}^*)^{q}$ exhibits almost identical average collision-free ratio (and also the standard deviation) to the global Haar random circuits. 
Given that the bosonic birthday paradox is guaranteed for the global Haar random circuit ensemble, also given that the mode number we used covers all recent large-sized Boson Sampling experiments~\cite{zhong2020quantum, zhong2021phase, madsen2022quantum, deng2023gaussian}, our numerical result shows that the local random circuit ensemble for $(\mathcal{B}\mathcal{B}^*)^{q}$ guarantees the bosonic birthday paradox at least for the near-term finite-sized experiments.
Even though this does not guarantee the bosonic birthday paradox of the local random circuit ensemble for $(\mathcal{B}\mathcal{B}^*)^{q}$ in the asymptotic limit as the system size scales, still, we believe that constant repetition number $q$ would be sufficient for the bosonic birthday paradox, due to the close resemblance to the global Haar random circuit ensemble as shown in the figure.

Additionally, we similarly present a numerical simulation result in Fig.~\ref{fig: Ratio} for the mode number $M = 32$ and photon number $N = 4, 8, 12, 16$, which indicates the saturated regime where the bosonic birthday paradox does not hold. 
Even in this regime, the collision patterns of local random circuits over $(\mathcal{B}\mathcal{B}^*)^{q}$ also show a great resemblance to those of global Haar random circuits.
This result gives hope that one can possibly extend our result to the saturated regime, by conducting a similar analysis for our circuit ensemble to the analysis for global Haar random circuits~\cite{bouland2023complexity}.

}

\end{document}